\documentclass[sigconf, ]{acmart}
\AtBeginDocument{%
  }

\setcopyright{acmlicensed}
\copyrightyear{2024}
\acmYear{2024}
\acmDOI{XXXXXXX.XXXXXXX}

\acmConference[Conference'17]{ACM Conference}{July 2017}{Washington,
  DC, USA}

\acmISBN{978-1-4503-XXXX-X/18/06}
\usepackage{graphicx}
\usepackage{subfigure}
\usepackage{amsmath}
\usepackage[toc,page]{appendix}
\usepackage{ragged2e} 
\usepackage{booktabs,makecell, multirow, tabularx}
\usepackage{tikz}
\usepackage{algpseudocode}
\usepackage{algorithm}
\makeatletter
\newenvironment{breakablealgorithm}
{
		\begin{center}
			\refstepcounter{algorithm}
			\hrule height.8pt depth0pt \kern2pt
			\renewcommand{\caption}[2][\relax]{
				{\raggedright\textbf{\ALG@name~\thealgorithm} ##2\par}%
				\ifx\relax##1\relax 
				\addcontentsline{loa}{algorithm}{\protect\numberline{\thealgorithm}##2}%
				\else 
				\addcontentsline{loa}{algorithm}{\protect\numberline{\thealgorithm}##1}%
				\fi
				\kern2pt\hrule\kern2pt
			}
		}{
		\kern2pt\hrule\relax
	\end{center}
}
\makeatother

\usetikzlibrary{calc}

\usepackage{pgf-umlsd}
\usepackage{pifont}
\usepackage{tcolorbox}
\usepackage{float}
\usepackage{enumitem}
\usepackage{needspace}

\usetikzlibrary{graphs,quotes, shapes.geometric, positioning,shapes,arrows}
\theoremstyle{definition}
\newtheorem{definition}{Definition}[section]
\newtheorem{theorem}{Theorem}
\newtheorem{lemma}{Lemma}

\newcommand{\xmark}{\ding{55}}

\newcommand{\CP}{\mathcal{P}  }
\newcommand{\C}{\mathcal{C}  }
\newcommand{\R}{\mathcal{R} }

\newcommand{\FC}{\mathcal{F}_{Channel}} 
\newcommand{\Fclk}{\mathcal{F}_{clk}}   
\newcommand{\FL}{\mathcal{L}}   
\newcommand{\FR}{\mathcal{H}}   
\newcommand{\Fan}{\mathcal{F}_{auth}}   
\newcommand{\Fex}{\mathcal{F}_{mex}}    
\newcommand{\Fjc}{\mathcal{F}_{jc}}     
\newcommand{\Fsig}{\mathcal{F}_{Sig}}   

\newcommand*{\Sim}{\textit{Sim}}
\newcommand{\env}{\mathcal{E}}
\newcommand{\Fpre}{\mathcal{F}_{prelim}}
\newcommand{\Fideal}{\mathcal{F}_{ideal}}
\newcommand\leftarrowS{\leftarrow\joinrel\smalldollar}
\makeatletter
\newcommand{\smalldollar}{\mathrel{\mathpalette\small@dollar\relax}}
\newcommand{\small@dollar}[2]{%
  \vcenter{\hbox{%
    $#1\textnormal{\fontsize{0.7\dimexpr\f@size pt}{0}\selectfont\$}$%
  }}%
}
\makeatother
\usepackage{adjustbox}
\usepackage{mdframed}
\mdfdefinestyle{MyFrame}{%
    linecolor=black,
    outerlinewidth=1pt,
    roundcorner=5pt,
    innertopmargin=2pt,
    innerbottommargin=2pt,
    innerrightmargin=5pt,
    innerleftmargin=5pt,
    font=\small}

\begin{document}

\title{FairRelay: Fair and Cost-Efficient Peer-to-Peer Content Delivery through Payment Channel Networks}

\author{Jingyu Liu}
\email{jingyuliu@hkust-gz.edu.cn}
\affiliation{%
  \institution{The Hong Kong University of Science and Technology (Guangzhou)}
  \streetaddress{No.1 Du Xue Rd, Nansha District}
  \city{Guangzhou}
  \state{Guangdong}
  \country{China}
}

\author{Yingjie Xue}
\email{yingjiexue@hkust-gz.edu.cn}
\affiliation{%
  \institution{The Hong Kong University of Science and Technology (Guangzhou)}
  \streetaddress{No.1 Du Xue Rd, Nansha District}
  \city{Guangzhou}
  \state{Guangdong}
  \country{China}
}

\author{Zifan Peng}
\email{zpengao@connect.hkust-gz.edu.cn}
\affiliation{%
  \institution{The Hong Kong University of Science and Technology (Guangzhou)}
  \streetaddress{No.1 Du Xue Rd, Nansha District}
  \city{Guangzhou}
  \state{Guangdong}
  \country{China}
}

\author{Chao Lin}
\email{cschaolin@163.com}
\affiliation{%
  \institution{Jinan University}
  \streetaddress{No.601 West Huangpu Rd, Tianhe District}
  \city{Guangzhou}
  \state{Guangdong}
  \country{China}
}

\author{Xinyi Huang}
\email{xyhuang81@gmail.com}
\affiliation{
   \institution{Jinan University}
  \streetaddress{No.601 West Huangpu Rd, Tianhe District}
  \city{Guangzhou}
  \state{Guangdong}
  \country{China}
}

\renewcommand{\shortauthors}{Liu et al.}

\begin{abstract}

\textit{Peer-to-Peer} (P2P) content delivery, known for scalability and resilience, offers a decentralized alternative to traditional centralized \textit{Content Delivery Networks} (CDNs). A significant challenge in P2P content delivery remains: the fair compensation of relayers for their bandwidth contributions. Existing solutions employ blockchains for payment settlements, however, they are not practical due to high on-chain costs and over-simplified network assumptions.
In this paper, we introduce \textit{FairRelay}, a fair and cost-efficient protocol that ensures all participants get fair payoff in complex content delivery network settings. 
We introduce a novel primitive, \textit{Enforceable Accumulative Hashed TimeLock Contract} (\textit{Enforceable A-HTLC}), designed to guarantee payment atomicity \textemdash ensuring all participants receive their payments upon successful content delivery.

The fairness of FairRelay is proved using the Universal Composability (UC) framework. Our evaluation demonstrates that, in optimistic scenarios, FairRelay employs \textit{zero} on-chain costs.
In pessimistic scenarios, the on-chain dispute costs for relayers and customers are constant, irrespective of the network complexity. Specifically, empirical results indicate that the on-chain dispute costs for relayers and customers are 24,902 gas (equivalent to 0.01 USD on Optimism L2) and 290,797 gas (0.07 USD), respectively.  
In a 10-hop relay path, FairRelay introduces less than 1.5\% additional overhead compared to pure data transmission, showcasing the efficiency of FairRelay.

\end{abstract}



\settopmatter{printfolios=true}
\maketitle

\section{Introduction}

\textit{Peer-to-Peer} (P2P)  content delivery embraces a decentralized approach that has fundamentally transformed the landscape of digital content distribution. Distinct from centralized content delivery networks (CDNs), P2P systems  reduce dependence on centralized infrastructures, thereby achieving enhanced robustness, resilience, and cost-efficiency in distributing content \cite{zolfaghari_content_2021}. 
Due to its effectiveness, this paradigm has gained immense popularity and led to the development of various protocols, such as BitTorrent \cite{levin_bittorrent_2008}, FairDownload \cite{9929262}, and Bitstream \cite{linus_bitstream_nodate}.
 BitTorrent, as a remarkable representative protocol, has a substantial user base with more than 150 million monthly active users since 2012\footnote{https://en.wikipedia.org/wiki/BitTorrent}, manifesting the widespread interest from the general public.

\begin{figure}[!htbp]
    \centering
    \setlength{\abovecaptionskip}{0.2cm}
    \includegraphics[width=0.4\textwidth]{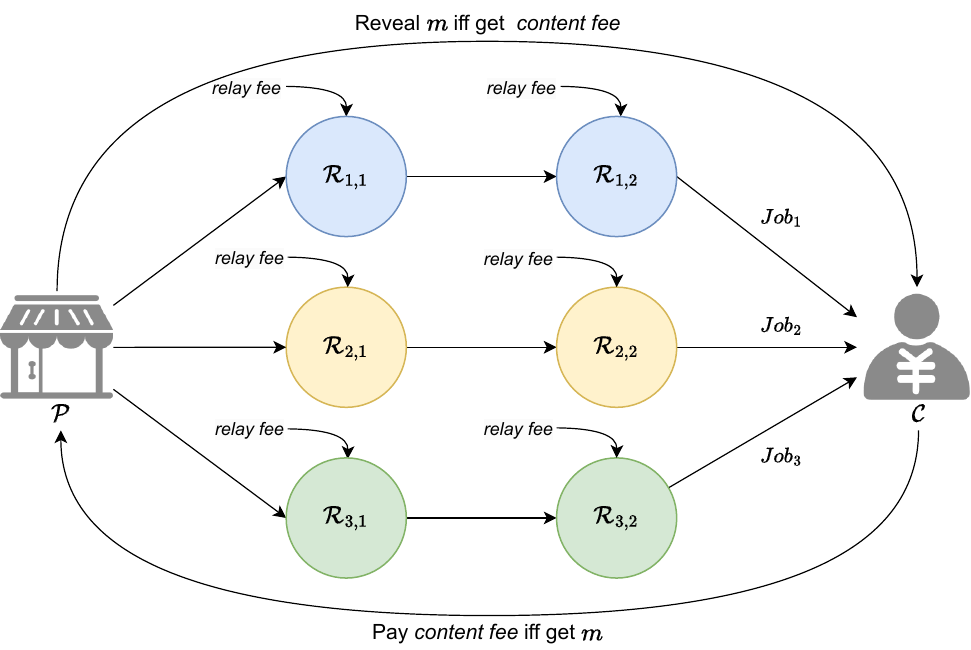}
    \caption{A digital content delivery process with multiple pathways succeeds only if all relayers are paid.}
    \label{fig:intro_model}
    \vspace{-0.5cm}
\end{figure}

A fundamental challenge in P2P content delivery is the assurance of fair payoff for each participant. Consider a scenario where a customer $\mathcal{C}$ seeks to purchase digital content $m$ from a content provider $\CP$.  There are three roles: a customer $\C$, a content provider $\CP$ and multiple relayers $\R$. Content $m$ is transmitted through many relayers. To optimize the relay service, the transmission of $m$ which comprises multiple data chunks, is divided into several \textit{delivery jobs}. Each job corresponds to a multi-hop pathway, leveraging the computing power and bandwidth of relayers' devices \cite{jain_using_nodate, padmanabhan_resilient_2003, anjum_survey_2017}, as shown in Fig. \ref{fig:intro_model}. To make sure every party gets a fair payoff, the following properties should be achieved : 
\begin{enumerate}[leftmargin=*]
    \item  $\mathcal{C}$ pays a \textit{content fee}  to $\CP$ only if $\mathcal{C}$ receives the content $m$;
    \item $\CP$ provides the content $m$ to  $\mathcal{C}$ only if $\CP$  gets the \textit{content fee};
    \item $\mathcal{R}$ gets \textit{relay fee} only if it provides relay service;
    \item The relayers should not know the content of $m$.
\end{enumerate}

By (4), the content $m$ should be encrypted during transmission.  Then, by (1) and (2), the content delivery process between $\C$ and $\CP$ is a two-party \textit{fair exchange} problem, where $\C$ pays $\CP$ only if $\CP$ reveals the encryption key for content $m$. 

\sloppy It is important to recognize that two models exist to incentivize participation in P2P content delivery: the altruistic player model and the rational player model. BitTorrent \cite{levin_bittorrent_2008} employs a reputation-based framework, effective in an altruistic setting where participants share contents based on personal interests. Here, we would like to explore the rational player model. By introducing monetary incentives, we anticipate that a greater number of rational players will join in the P2P content delivery network, thereby promoting the widespread adoption of this decentralized technology.

Existing fee-related P2P content delivery  protocols can be categorized into two classes: 1) Using a centralized party to distribute fees. Protocols like Floodgate \cite{nair_floodgate_2008}, Meson\cite{meson}, and Saturn\cite{saturn} rely on a trusted node or committee to distribute rewards. This leads to potential issues such as a single point of failure, where the trusted entity could become a target for attacks or manipulation. 
2) Using a blockchain to distribute fees. Various researches \cite{bertino_fair_2021, 9929262,  dziembowski_fairswap_2018, zkcp} leverage public blockchain to serve as a decentralized trusted third party to facilitate fair exchanges. For example, \textit{Zero-knowledge Contingent Payment} (ZKCP) \cite{zkcp} and FairSwap \cite{dziembowski_fairswap_2018} allow parties to perform fair exchange utilizing on-chain payment. However, the on-chain payments result in high costs and latency. In FairSwap,  a single exchange costs 1,050,000 gas (0.24 USD on Optimism L2) and at least three on-chain communication rounds (around 6s on Optimism L2). This on-chain overhead is unacceptable considering a user might retrieve hundreds of videos per day.


To avoid the costly on-chain payments, many solutions \cite{linus_bitstream_nodate, miller_sprites_2017}  adopt payment channels. 
However, payment channels typically solve the problem of fair exchange between two parties. They are not able to address the relay fee problem in (3) as there are no relayers in their models.


 Consider a simple case where $m$ is relayed by a set of relayers in a single path. The relayers must get relay fee for providing bandwidth for relaying. A straightforward attempt would be including relayers in a multi-hop payment channel network to relay payments, which would fail to solve the problem. The challenge lies in the inability to materialize the completion of a relay job. FairDownload \cite{9929262}  proposed the encryption of incoming data by relayers before relaying, which shifts the settlement of the relay job to the release of each relayer's decryption key. 
 The protocol succeeds to solve the relay fee problem with on-chain settlements when a single relayer is involved in the delivery jobs. For complex delivery jobs where $m$ is relayed by multiple relayers and multiple paths, the following challenges remain to be tackled.
 
\noindent \textbf{The costs.} The costs would be unacceptable if every payment is settled on-chain. For example, in FairDownload, each relayer additionally requires at least 171,591 gas (0.05 USD in Optimism L2) to get their relay fee.
Thus, it is desirable to move the payments off-chain using payment channel networks.  When operations are moved to off-chain payment channels, providing fair payments for every participant becomes challenging due to the lack of public synchronized state. 

 \noindent \textbf{Payment atomicity in complex networks.}    A desired protocol should guarantee the \emph{atomicity} of payments for all participants.  Atomicity means that, for relayers, once the customer gets the content $m$ implying all relayers' delivery jobs are successfully completed, then all relayers should get their relay fees. Otherwise, the delivery jobs fail, and no relayer gets relay fee and the customer does not get the content, either. 
 In a network where there are multiple relayers, relayers may exhibit complex behaviors, such as colluding relayers launching a \textit{wormhole attack} \cite{malavolta_anonymous_2019} to avoid paying some intermediate relayers. Therefore, \textbf{guaranteeing the atomicity of paying all relayers is an open problem.} 
 
 When the content is delivered through multiple paths, the problem becomes more intricate. Each channel updates independently off-chain, necessitating a synchronization scheme to ensure all paths are settled atomically.

\noindent \textbf{Contributions.} 
In this paper, we address the fairness problem in P2P content delivery involving multiple relayers and multiple paths, which is common in practice.  We provide \emph{FairRelay}, a fair and cost-efficient protocol  for P2P content delivery through payment channel networks. 
Our contributions are summarized as follows.
\begin{itemize}[leftmargin=*,topsep=5pt]
    \item \textbf{An atomic multi-hop payment scheme to pay all relayers in a single path.} To facilitate the fair payment for all participants, we design \textit{Accumulative Hashed TimeLock Contract} (\textit{A-HTLC}), based on which we construct a multi-hop payment scheme. Compared to traditional HTLC, A-HTLC enables hashlocks to accumulate along the path, which is essential to guarantee all relayers get fair payment.
    \item \textbf{An atomic multi-path payment scheme to settle all paths' payments.} To ensure all payments on multiple paths are settled atomically \textemdash all paths are settled or none is settled, we further introduce \textit{Enforceable A-HTLC} where an \textit{acknowledge-pay-enforce} procedure is introduced to enforce all sub-payments of all paths. 
    If any party fails to settle the payment, the provider can submit a challenge on-chain. The challenged party is enforced to settle the payment to avoid a penalty.

    \item \textbf{A fair and cost-efficient content delivery  protocol.} We introduce \emph{FairRelay}, a cost-efficient protocol enabling fair content delivery. 
    When all participants are honest, no on-chain costs are introduced. 
    If some participants are dishonest, a \textit{proof of misbehavior} scheme is designed to tackle the dishonest behavior.
    In the worst case, the on-chain overheads for dispute are constant for customers and relayers, regardless of the network complexity. 

    \item \textbf{Security analysis, implementation and evaluation.} We prove the protocol's fairness in the \textit{Universal Composability} (UC) framework. We implement FairRelay, and conduct performance evaluation. The evaluation results manifest the protocol's practicality and efficiency.

\end{itemize}

\noindent \textbf{Organization}. The paper is structured as follows. Section \ref{sec:prelimaries} introduces necessary background and primitives. Section \ref{sec:prob} delineates the protocol's security and privacy objectives. Section \ref{sec:overview} offers a technical overview of FairRelay. Detailed construction of the protocol is elaborated in Section \ref{sec:protocol_construction}. Section \ref{sec:security} sketches a security analysis, and Section \ref{sec:eval} evaluates the efficiency of our protocol. Related works are reviewed in Section \ref{sec:related_works}. The paper concludes in Section \ref{sec:conclusion}.

\section{Preliminaries}\label{sec:prelimaries}

We intend to enable fair content delivery through PCNs. In this section, we first introduce the payment channels and an important primitive \textit{Hashed TimeLock Contracts} (\textit{HTLCs}), then introduce commitment schemes and proof of misbehavior mechanisms to uphold the integrity of all participants. zk-SNARKs are used are as tools to reduce computational overhead in verifying proof of misbehavior. We also provide a brief introduction to the universal composability framework, setting the context for our subsequent security analysis. Some basic cryptographic primitives are given in the end.

\vspace{0.3em}
\noindent \textbf{Payment Channels and HTLC.} Payment channel provides an off-chain solution enabling two parties to transact without involving the blockchain. Two parties can open a payment channel by depositing funds into a multi-signature address and then update the channel balance by exchanging signed transactions. Parties within a payment channel can perform complex conditioned payment \cite{dziembowski_general_2018} via HTLC. With HTLC, a party can pay another party if the latter reveals a preimage of a hash value before a deadline. In this work, we employ conditioned payment and model payment channel as a single ideal functionality $\mathcal{F}_{Channel}$.

\vspace{0.3em}
\noindent \textbf{Commitment and Proof of Misbehavior.} A commitment scheme allows one to commit a chosen value or statement while keeping it hidden to others, with the ability to reveal the committed value/statement later \cite{goldreich2009foundations}. \textit{Proof of misbehavior} (PoM) scheme is a cryptographic construction that enables a prover to demonstrate  that an entity has violated its previous claim. PoM schemes can be combined with commitment schemes, serving to substantiate inconsistencies between revealed information and commitments. 

\vspace{0.3em}
\noindent \textbf{zk-SNARKs.} \textit{Zero-Knowledge Succinct Non-interactive Argument of Knowledge} ({zk-SNARKs}) are cryptographic paradigms that enables a prover to validate the truthfulness of a statement to a verifier without divulging any private witness pertaining to the statement itself \cite{bitansky2012extractable}. 
Our research adopts a zk-SNARK system  to construct succinct validity proofs within the proof of misbehavior framework, significantly reducing the computational and storage overheads associated with the on-chain verification of misbehavior proofs. 

\vspace{0.3em}
\noindent \textbf{Universal Composability.} We model the security of a \textit{real world} protocol $\Pi$ within the framework of universal composability \cite{canetti_universally_2001}, considering static corruptions where the adversary $\mathcal{A}$ declares which parties are corrupted at the outset. 
To analyze the security of $\Pi$ in the real world, we "compare" its execution with the execution of an ideal functionality  $\mathcal{F}$ in an \textit{ideal world}, where $\mathcal{F}$ specifies the protocol's interface and can be considered as an abstraction of what properties $\Pi$ shall achieve. In the ideal world, the functionality can be "attacked" through its interface by an ideal world adversary called simulator $\Sim$. 
Intuitively, a protocol $\Pi$ realizes an ideal functionality $\mathcal{F}$ when the environment $\env$, acting as a distinguisher, cannot distinguish between a real execution of the protocol and a simulated interaction with the ideal functionality. In the UC framework, $\env$ use session id (sid) to distinguish different instances. 

\vspace{0.3em}
\sloppy 
\noindent \textbf{Fundamental cryptographic primitives.} Our research builds upon some cryptographic primitives, including a commitment scheme $(Commit, Open)$ that relies on a hash function ($Commit(\cdot) = \mathcal{H}(\cdot)$)\footnote{A standard commitment scheme employs a collision-resistant hash function along with random padding.}, a symmetric encryption scheme $(SE.\text{KGen},\text{Enc}, \text{Dec})$, and an asymmetric encryption scheme $(AE.\text{KGen},\text{Enc}, \text{Dec}, \text{Sign}, \text{Ver})$. 
We assume that each participant possesses a set of key pairs, and the public key is known to all.  We use $K(u)$ to denote the private key of a user $u$, $Pk(u)$ to denote the public key of a user $u$. 
Furthermore, our protocol employs \textit{Merkle Trees} $(MT.\text{Root},\text{Member}, \text{Ver})$ with \textit{Merkle Multi-proof} for efficient membership verification and integrity check \cite{merklemulti}. 

\section{Assumptions and Goals}    \label{sec:prob}
In this section, we elaborate on the assumptions underlying the content delivery system and outline the security and privacy objectives, namely fairness and confidentiality.

\subsection{Assumptions}    \label{sec:assumption}
\begin{itemize}[leftmargin=*]
 \item \textbf{Path Existence. } Consider a  content delivery with multiple relayers over a PCN: a content provider $\CP$ connects to a customer $\C$ through multiple payment paths. We assume the existence of a content relay network that overlaps with a PCN. Specifically, we assume a valid payment path from $\C$ to $\CP$ with sufficient liquidity to pay the content fee $\mathfrak{B}_m$, and $\eta$ multi-hop payment paths with enough bandwidth and liquidity for transferring content from $\CP$ to $\C$, covering the relay fees. 
    \item \textbf{Adversarial Model. } 
    We consider a static adversary $\mathcal{A}$ that operates in \textit{probabilistic polynomial time}. The adversary $\mathcal{A}$ has the ability to corrupt any participant in this content delivery process before the protocol begins. 
    \item \textbf{Communication Model. } All communication between parties occurs over an authenticated communication channel $\Fan$ \cite{SigUC}.
    We assume a synchronous communication model, where all parties are constantly aware of the current round, which is modeled by a global clock functionality \cite{syncUC}\footnote{In the real world, the global clock can be achieved through the block height of a specific public blockchain}. If a party (including the adversary $\mathcal{A}$) sends a message to another party in round $i$, then that message is received by the recipient before the start of the next round.
    \item \textbf{TTP Model. }  We deploy a smart contract called \textit{Judge Contract} on a public-verifiable Turing-complete decentralized ledger as the \textit{trusted third party} (TTP). Any node in this delivery can interact with the smart contract to settle any dispute. 
   
\end{itemize}

\subsection{Security and Privacy Goals} \label{sec:prob_def}

\noindent \textbf{Notation.} We model $\CP$ delivering a content $m$ to $\C$ through $\eta$ delivery paths as a directed graph $\mathcal{G}$.  Content $m$ consists of $n$ chunks ($\{m_1, \dots, m_n\}$), and is committed by $com_m$, where $com_m$ is the merkle root built on all chunks ($com_m:= MT.Root(m)$). 
Each path $p$ is assigned with a delivery job $Job(p)$.  $\R_{k, i}$ is the $i$-th relayer in the $k$-th path ($p_k$).
$\mathfrak{B}_m$ is the content fee paid to $\CP$, and $\mathfrak{B}_{k, i}$ is the relay fee paid to $\R_{k, i}$. We denotes $|p_k|$ as the number of relayers in path $p_k$. 
In this work, we focus on two fundamental properties, fairness and confidentiality, which we define below.

\begin{definition}[\textbf{Fairness for $\mathcal{P}$}]
    \label{def:exf_p}
    The fairness for $\mathcal{P}$ is guaranteed if both of the following conditions are satisfied. 
    \vspace{0.3em}
    \begin{itemize}[leftmargin=*]
        \item For any corrupted PPT $\mathcal{C}, \mathcal{R} \in \mathcal{G}$ controlled by $\mathcal{A}$, honest $ \CP$ reveals $m$ to $\C$ only if $ \mathcal{P}$ gets the content fee $\mathfrak{B}_m$.
        \item For any corrupted PPT $\mathcal{C}, \mathcal{R} \in \mathcal{G}$ controlled by $\mathcal{A}$, if an honest entity $\mathcal{P}$ does not receive the content fee $\mathfrak{B}_m$, it does not pay any relay fee.

    \end{itemize}

\end{definition}

\begin{definition}[\textbf{Fairness for $\mathcal{C}$}]
    \label{def:exf_c}
    For  any corrupted PPT $  \mathcal{P}, \mathcal{R} \in \mathcal{G}$ controlled by $\mathcal{A}$, honest $ \mathcal{C}$ pays $ \mathfrak{B}_{m}$ to $\mathcal{P}$ only if $\mathcal{C}$ learns $m$.     
\end{definition}

\begin{definition}[\textbf{Fairness for $\mathcal{R}$}]
    \label{def:exf_r}
    For any corrupted PPT $  \mathcal{C}, \mathcal{R} \in \mathcal{G}$ controlled by $\mathcal{A}$,  $\C$ learns the data chunks relayed through paths $p_k$ only if every $\mathcal{R}_{k,i} \in p_k$  gets the relay fee $\mathfrak{B}_{{k,i}}$. 
\end{definition}

\begin{definition}[\textbf{Confidentiality}]
    \label{def:conf}
    For any corrupted PPT $ \R \in \mathcal{G}$ controlled by $\mathcal{A}$, $\mathcal{A}$ can not learn $m$ if $\CP$ and $\C$ are honest. 
\end{definition}

\section{Technical Overview}    \label{sec:overview}


\subsection{Key idea}   \label{sec:key_idea}

In order to keep the content confidential during the relay process, $\mathcal{P}$ encrypts all chunks in $m$ with its symmetric encryption key. For simplicity, let us first consider a single delivery path where a data chunk $m_i$ is relayed through multiple relayers. To ensure that relayers receive a relay fee, each relayer in the multi-hop path encrypts the relayed content before delivering it to the customer $\C$. This encryption reduced the content delivery problem into two fair exchange problems: the fee-key exchange between $\mathcal{P}$ and $\mathcal{C}$ and the fee-key exchange between  $\mathcal{P}$ and each relayer $\mathcal{R}$, since $P$ pays the relay fee. 

The introduction of multiple hops poses challenges to ensure the atomicity of key distribution and payment. To decrypt the encrypted content received by the customer in the last hop, the customer must obtain all encryption keys along the path. A straightforward solution is using HTLC: the customer escrows its content fee and locks it using the hashes of all encryption keys. The provider sends the encryption keys to redeem the payment. However, this trivial solution compromises the confidentiality of the content: 
the encryption key used for redemption might be revealed on-chain, leading the encrypted content accessible to all relayers and network eavesdroppers\footnote{In case of dispute, the provider might submit the keys on-chain to redeem the payment.}.
To address this issue,  we design a sophisticated key delivery scheme by masking the encryption key with a secret, and then using the secret to redeem the payment. This approach allows the secrets to be released for content fee redemption without compromising the confidentiality of the content. At this point, the problem is formally modeled as a multi-party fee-secret exchange as shown in Fig. \ref{fig:relay_content_exchange}.

However, another challenge remains. By exchanging a secret with a content fee using HTLC, it only guarantees that the secret corresponds to the lock specified in the HTLC. The validity of the secret for decryption is not guaranteed in the aforementioned process. To tackle this challenge, we construct a \textit{proof of misbehavior} scheme to ensure the validity of the provided secret. In this scheme, both the content provider and relayers are required to make a deposit. Before the content delivery, the provider and all relayers make commitments on the secret they used for decryption. When a payment is made and the secret is released, if the secret is not consistent with the commitment, it can be uploaded to an on-chain \textit{Judge Contract} along with the commitment which will penalize the provider/relayer by slashing their deposits. 

In a P2P content delivery instance as shown in Fig. \ref{fig:relay_content_exchange}, a protocol achieves fairness if all payments and secret reveals are settled atomically.

\begin{figure}[htbp]
    \setlength{\abovecaptionskip}{0.1cm}
    \centering
    \includegraphics[width=0.47\textwidth]{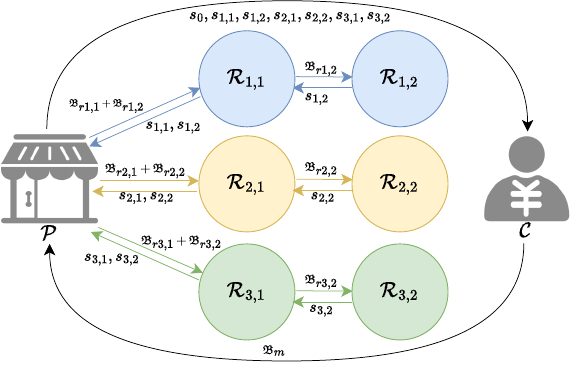}
    \caption{Multi-party Money-Secret Exchange}
    \label{fig:relay_content_exchange}
    \vspace{-0.2cm}
\end{figure}

\vspace{0.5em}
\noindent \textbf{Fair Multi-hop Payment.} \label{key:ahtlc} To ensure atomicity in multi-hop fair payments, we propose an enhanced HTLC scheme called \textit{Accumulative Hashed TimeLock Contract} (\textit{A-HTLC}). Instead of directly revealing secrets to the customer $\mathcal{C}$, each relayer progressively reveals its secret to the preceding hop through the payment channel. While following HTLC's lock-unlock procedure, the unlock condition propagates incrementally from the last relayer to the provider of the payment path.

\begin{figure}[htbp]
    \setlength{\abovecaptionskip}{0.1cm}
    \includegraphics[width=0.47\textwidth]{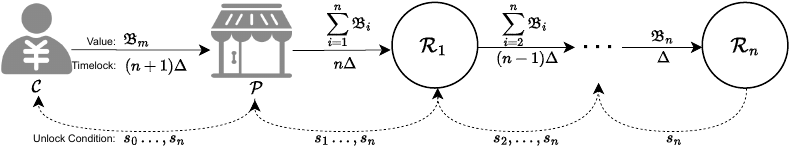}
    \caption{\textit{A-HTLC} Workflow }
    \label{fig:payment_path}
    \vspace{-0.2cm}
\end{figure}

Fig. \ref{fig:payment_path} illustrates the workflow of one specific path in \textit{A-HTLC}. In this scheme, $\C$ locks the content fee payment to the provider $\CP$, requiring the revelation of all secrets $\{s_0, s_{1}, \ldots, s_{n}\}$ before a timeout. Subsequently, $\CP$ locks a conditioned payment to $\mathcal{R}_{1}$ with the unlocking condition being the release of all secrets $\{s_{1}, s_{2}, \ldots, s_{n}\}$. This sequential process continues until the last relayer $\mathcal{R}_{n}$ in this path, where the relayer $\R_{n-1}$ locks the payment using $\R_n$'s secret $s_n$. 
The lock time for each conditioned payment reduces linearly from left to right, ensuring strong atomicity against the \textit{wormhole attack}, 
where no one can unlock their outcome payments until all their incoming payments are unlocked.
In general, \textit{A-HTLC} consolidates atomic secret reveals in a multi-hop payment, and overcomes the \textit{weak atomicity} \cite{malavolta_anonymous_2019} of HTLC in multi-hop payments, thereby preventing relayers from exploiting intermediate relay fees.
 
\vspace{0.5em}
\noindent \textbf{Fair Multi-path Payment.} \label{key:enforce}
While \textit{A-HTLC} addresses the atomicity problem in multi-hop payments, the challenge of achieving atomicity over multiple paths persists due to the asynchronous nature of updates in different paths.
To overcome this challenge, we introduce an \textit{acknowledge-pay-enforce} paradigm based on \textit{A-HTLC}: all relayers reach an "agreement" with the provider on a global deadline $T$ by which all sub-paths must be settled (all payments are unlocked). If any sub-path remains unsettled by $T$, $\mathcal{P}$ has the authority to request the relayers' secrets on-chain, thereby enforcing the settlement of the remaining payment path.
The following outlines the workflow of \textit{acknowledge-pay-enforce} paradigm.

To ensure payment settlement only after reaching an "agreement", $\CP$ generates a synchronizer secret $s_{sync}$ and incorporates its commitment $h_{sync}$ into the unlock condition for each \textit{A-HTLC} payment initiated from $\CP$. A relayer $\mathcal{R}_i$ sends a \textit{lock receipt} to $\CP$ when its incoming channel is locked. The  \textit{lock receipt} states that "If all subsequent relayers reveal their secrets and $s_{sync}$ is revealed on-chain, I will reveal my secret $s_i$". 
Once $\CP$ gathers all receipts, $\CP$ distributes $s_{sync}$, and then each path commences unlocking the payment channel from the last relayer to the first. 
If a path is not unlocked as expected, $\CP$ initiates an on-chain challenge with all receipts to enforce the unlock process in this path, requiring each relayer to sequentially reveal its secret on-chain. Failure to reveal a secret in time results in punishment for the corresponding relayer, with a compensation transferred towards $\CP$.

\subsection{Protocol Overview}

Here is the sketch of our protocol. We transform the content delivery problem into multi-party fee-secret exchanges. To ensure the correct release of secrets, we propose a commitment scheme along with proof of misbehavior mechanism. 
Our protocol proceeds in four main phases: \textit{setup}, \textit{delivery}, \textit{payment} and \textit{decryption}, with two additional \textit{challenge} mechanisms if misbehavior occurs, as described below and visualized in Fig. \ref{fig:overview}. 
\begin{figure}[htbp]
\setlength{\abovecaptionskip}{0.cm}
    \includegraphics[width=0.47\textwidth]{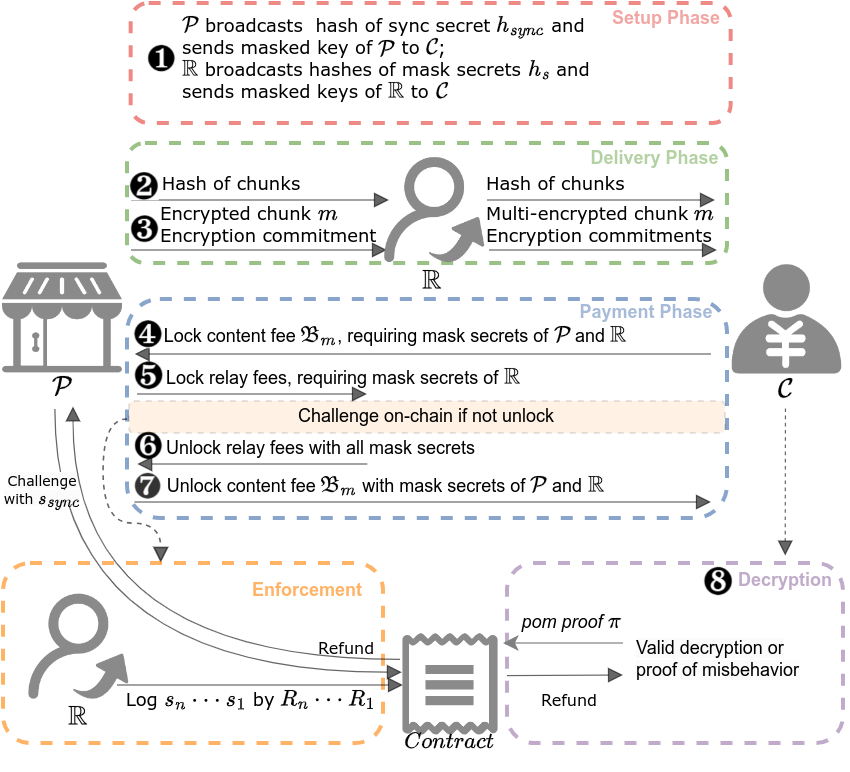}
    \caption{Protocol Overview}
    \label{fig:overview}
    \vspace{-0.2cm}
\end{figure}

During the setup phase, each party involved, excluding $\mathcal{C}$, generates an encryption key and masks the key with a secret. The results after masking are delivered to $\mathcal{C}$ privately, and the hash values of these secrets are utilized as hashlocks in \textit{A-HTLC}.
In the delivery phase, $\mathcal{P}$ encrypts the plaintext chunks and transmits them to the first relayer along with its encryption commitments. Each subsequent relayer further encrypts the incoming chunks and appends encryption commitments before forwarding them.
Once all ciphertexts are delivered, exchanges happen in the payment phase, where $\C$ buys all encryption keys from $\CP$, and $\CP$ buys keys from $\R$s. We utilize \textit{Enforceable A-HTLC} to achieve atomic multi-hop fee-secret exchange. 
In the decryption phase, $\mathcal{C}$ unmasks all the secrets to obtain the encryption keys and verifies them against the commitments. $\mathcal{C}$ decrypts the ciphertexts with the encryption keys and verifies the decryption results against the commitments. If any inconsistency is detected, $\mathcal{C}$ can submit a proof of misbehavior on-chain to request a compensation.

\section{Protocol Construction} \label{sec:protocol_construction}

\subsection{Building Blocks}


\noindent \textbf{Ledger and Channels.} We utilize a ledger $\FL$ modeled in \cite{egger_atomic_2019} and payment channel networks $\FC$ based on the model in \cite{aumayr_thora_2022}. The corresponding functionalities are elaborated in Fig. \ref{fig:ledger_if} and \ref{fig:offchain_payment_channel}. $\FL$ publicly maintains the balances of each user and provides two interfaces, namely \textit{transfer} and \textit{query}. The \textit{transfer} interface allows other ideal functionalities to transfer balances from one user to another, while the \textit{query} interface enables anyone to query a user's latest balance.
$\FC$ enables two parties to perform arbitrary off-chain conditioned payments. 
$\FC$ supports two core operations: the \textit{query} interface allows the participants in a payment channel to look up the latest balance, while the \textit{update} interface allows two parties to reach an agreement on a new balance update with a specific condition $\phi$. Once the condition is fulfilled, any party can update the channel by submitting this agreement with the parameters to $\FC$.
For example, in a channel utilizing a HTLC as the update condition, $\phi$ could specify "provide a preimage of $h$ before time $t$". The condition $\phi(s, ct) = 1$ holds only if the \textit{redeem parameter} $s$ is the preimage of $h$ and the current round time $ct$ is less than $t$. To simplify the notation, we define two functions to build such agreements (detailed in Appendix \ref{app:functions}, Algorithm \ref{algo:lock_unlock}):
\begin{itemize}[leftmargin=*]
    \item $lock(v, sid, cid,amt, \phi)$: This function allows one party $v$ in channel $cid$ to compose a partially-signed update agreement $tx$, transferring $amt$ tokens to the other party $w$. The transferring can be further redeemed by $w$ only if the condition $\phi$ is met. 
    \item $unlock(tx, w, s)$: This function allows the other party $w$ to generate an update message $(\textit{update}, \textit{sid}, cid, lb', rb', \phi, s)$ from the partially-signed update agreement $tx$ and a \textit{redeem parameter}  \textit{s}. If the current round is $ct$ and $\phi( s, ct) = 1$, $w$ can update the latest channel state to $(lb', rb')$, which represents the updated balances in the channel. Specifically, $amt$ tokens are transferred from $v$'s balance $lb$ to $w$'s balance $rb$.
\end{itemize}
\begin{figure}
    \setlength{\abovecaptionskip}{0.1cm}
    \centering
    \begin{mdframed}[userdefinedwidth= \linewidth,frametitlealignment=\center,  frametitle={ $\FL$}, frametitlerule=true, frametitlebackgroundcolor=gray!20, linecolor=gray!50, font=\small,  innertopmargin=2pt, innerbottommargin=2pt, innerrightmargin=5pt, innerleftmargin=5pt]

    \textbf{Local Variable}: $\FC$ maintains a map  $ Balance[uid] \mapsto x  $, where $x$ is the balance of user with $uid$. $Balance$ is public to all parties. 

    \textbf{API}
    \begin{itemize}[leftmargin=*]
        \item \textit{Transfer.} Upon receiving $( \textit{transfer},\textit{sid}, uid_s, uid_r, amt)$ from an ideal functionality of session $sid$:
        If $Balance[uid_s] \geq amt$, then $Balance[uid_s] := Balance[uid_s] - amt$, $Balance[uid_r] := Balance[uid_r] + amt$, send $(\textit{transferred}, sid,uid_s, uid_r, amt ) $ to the session $sid$. Otherwise, send $(\textit{insufficient}, sid, uid_s,  uid_r, amt ) $ to the session $sid$.
        \item \textit{Query.} Upon receiving $(\textit{query}, \textit{sid}, uid)$ from a party of session $sid$: send $(sid, uid, Balance[uid]) $
    \end{itemize}
    
    \end{mdframed}

    \caption{Ideal Functionality of Ledger }
    \label{fig:ledger_if}
    \vspace{-0.5cm}
\end{figure}
\begin{figure}
\setlength{\abovecaptionskip}{0.1cm}
\begin{mdframed}[userdefinedwidth= \linewidth, frametitlealignment=\center, frametitle={$\mathcal{F}_{Channel}$}, frametitlerule=true, frametitlebackgroundcolor=gray!20,  linecolor=gray!50, font=\small, innertopmargin=2pt, innerbottommargin=2pt, innerrightmargin=5pt, innerleftmargin=5pt ]

    \noindent \textbf{Local Variable}: 
    \begin{description}
        \item[$\mathbb{C}$]: a map  $ \mathbb{C}[cid] \mapsto \gamma  $, where $\gamma = \{lu, ru, lb, rb \}$. Here, $cid$ is the channel identifier, $lu$ is the left user id, $ru$ is the right user id, $lb$ is the left user balance, $rb$ is the right user balance. 
        \item[$ct$]: current round time retrieved from the global clock $\Fclk$. 
    \end{description}
\textbf{API}: 
\begin{itemize}[leftmargin=*]
    \item \textit{Update}.  Upon receiving $(\textit{update}, \textit{sid}, cid, lb', rb', \phi , s)$ from $\mathbb{C}[cid].lu $ or $\mathbb{C}[cid].ru$, where $(lb', rb')$ is the new state balance of channel $cid$, $\phi$ is the update condition, and $s$ is the condition parameters. 
    If both party agree the payment condition $ (cid, lb', rb', \phi)$ and the update condition is satisfied
    ($\phi(s, ct) = 1$), $\FC$ updates $\mathbb{C}[cid].lb = lb'$, $\mathbb{C}[cid].rb = rb'$, and sends $(\textit{updated}, \textit{sid}, cid, s)$ to both parties in the channel. 
    Otherwise, $\FC$ sends $(\textit{update-fail}, \textit{sid}, cid, s)$. $\FC$ will leak the \textit{update}/\textit{update-fail} message to $\Sim$. 
    \item \textit{Query}. Upon receiving $(\textit{query}, \textit{sid}, cid)$ from $\mathbb{C}[cid].lu $ or $ \mathbb{C}[cid].ru$, $\FC$ returns the latest channel state $(sid, cid,\mathbb{C}[cid].lb, \mathbb{C}[cid].rb)$.

\end{itemize}

\end{mdframed}
        \caption{Ideal Functionality of Payment Channels}
	\label{fig:offchain_payment_channel}
 \vspace{-0.5cm}
\end{figure}

\noindent \textbf{A-HTLC and Enforceable A-HTLC.} 
Extending HTLC's single-hash lock, \textit{A-HTLC} introduces a multi-hash lock mechanism where the payment condition $\phi_{\textit{A-HTLC}}$ is defined by a list of hashes $\mathbb{H}$ and a timelock $t$ ($\phi_{\textit{A-HTLC}} := \text{Construct}(\mathbb{H}, t)$). A payment locked by \textit{A-HTLC} can only be released if all preimages $\mathbb{S}$ corresponding to $\mathbb{H}$ are provided before the deadline $t$. This framework supports an atomic multi-hop payment scheme. 
In this context, the provider $\CP$ disburses a relay fee $\mathfrak{B}_i$ to each relayer $\R_i$ in a sequence of $n$ relayers, where each $\R_i$ holds the secret $s_i$ (with corresponding hash $h_i$), and the channel between $\R_{i-1}$ ($\R_0$ denotes the provider $\CP$) and $\R_i$ is denoted as $cid_i$.

$\CP$ initiates a conditional payment to $\R_1$ by creating the condition payment $tx_1$ ($tx_1 := \textit{lock}(\CP, sid, cid_1, \sum_{j=1}^{n} \mathfrak{B}_j, \phi_1)$), where $\phi_1 := \text{Construct}(\{h_1, \ldots, h_n\}, t_1)$. Each subsequent relayer $\R_i$ forwards a sub-payment $tx_{i+1}$ to $\R_{i+1}$, requiring the disclosure of secrets $\{s_i, \ldots, s_n\}$ before round $t_i$. Upon receiving $tx_n$, $\R_n$ redeems the payment by submitting $unlock(tx_n, \R_n, s_n)$, thereby disclosing $s_n$ to $\R_{n-1}$. This cascade ensures each relayer redeems its incoming payment, allowing secrets to accumulate from $\R_n$ back to $\CP$. 

Leveraging the "acknowledge-pay-enforce" paradigm (introduced in Section \ref{key:enforce}), we introduce the \textit{Enforceable A-HTLC} to guarantee the atomicity of all payments across multiple paths. In this scheme, $\CP$ selects a random synchronizer secret $s_{sync}$ and incorporates its hash ($h_{sync}$) into the hashlock of its all outgoing \textit{A-HTLC} payments. Upon receiving a valid payment $tx_i$ locked with $h_{sync}$, each $\R_i$ includes $h_{sync}$ in the subsequent payment's hashlock and sends a \textit{lock receipt} back to $\CP$. This receipt asserts that $\R_i$ will disclose its secret $s_i$ if all subsequent secrets $\{s_{i+1}, \ldots, s_n\}$ and $s_{sync}$ are revealed on-chain. Once $\CP$ collects all \textit{lock receipts}, $\CP$ releases $s_{sync}$, initiating the payment redemption process. Each payment must be sequentially redeemed by all relayers from $\R_n$ to $\R_1$. If the redemption process stalls, $\CP$ can enforce the disclosure by submitting an \textit{enforcement} request along with $s_{sync}$ to \textit{Judge Contract}, requesting relayers to reveal their secrets on-chain. Non-compliance results in penalties.

\vspace{0.3em}
\noindent \textbf{Commitment and PoM Construction} 
To ensure that the customer can access the content once certain secrets are revealed, we introduce commitment and proof of misbehavior schemes for masking and encryption.
The commitment scheme on masking ($MCOM$) allows a node $v$ to generate a commitment asserting that "$v$ possesses a key $sk$ and a mask secret $s$, and the mask result is $ck$ ($ck:= s \oplus sk$)". The commitment $com_{mask}$ consists of $(h_{sk}, h_s, ck, \sigma)$, where $h_{sk} := Commit(sk)$, $h_s := Commit(s)$, and $\sigma$ represents the signature of $v$ for $(h_{sk}, h_s, ck)$. Once the mask secret $s$ is revealed, the holder of $com_{mask}$ can decrypt $ck$ and obtain an encryption key $sk'$ ($sk' := ck \oplus s$).
The commitment scheme on encryption ($ECOM$) enables a node $v$ to generate a commitment asserting that "$v$ encrypts a data chunk $m_i$ in content with index $i$ using the key $sk$, and the encryption result is $c$". The commitment $com_{enc}$ consists of $(h_{m}, h_c, h_{sk}, i, \sigma)$, where $h_{m} := Commit(m_i)$, $h_c := Commit(c)$, $h_{sk}$ denotes the commitment of the key $sk$, and $\sigma$ refers to the signature of $v$ on $(h_{m}, h_c, h_{sk})$. When the encryption key $sk$ and the ciphertext $c$ committed by $h_{c}$ are revealed, the recipient of $com_{enc}$ can unveil a plaintext $m_i'$.
If the decryption result is inconsistent with the commitment on masking or encryption, the recipient can generate a proof of misbehavior against $v$ (using PoM scheme $PoME$ or $PoMM$) and submit to the \textit{Judge Contract}, alleging this inconsistency. The commitment and PoM scheme on masking and encryption formally defined in Appendix \ref{app:pomm} and  \ref{app:pome}.

\begin{figure*}[ht]
    \center 
    \setlength{\abovecaptionskip}{0.cm}
    \includegraphics[width=1\textwidth]{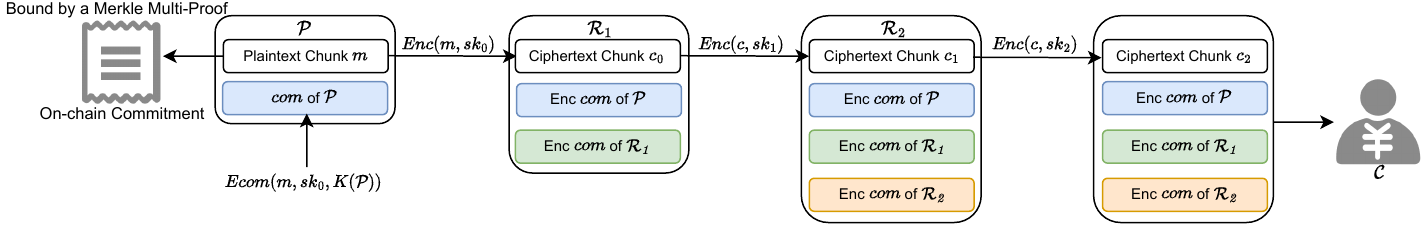}
    \caption{Encryption and Commitments in FairRelay: chunk $m$ is relayed from $\mathcal{P}$ to $\mathcal{C}$ through 2 relays. Each relay appends its encryption commitment to ciphertext and the commitments in a relay path form a chain $com_m$.}
    \label{fig:enc}
\end{figure*}

\begin{figure*}[htb]
    \setlength{\abovecaptionskip}{0.1cm}
    \vspace*{-0.2cm}
    \begin{mdframed}[frametitle={Ideal Functionality of \textit{Judge Contract} $\Fjc$}, frametitlealignment=\centering, frametitlerule=true, frametitlebackgroundcolor=gray!20, font=\small, innerleftmargin=1mm, innerrightmargin=1mm]

         All providers and relayers should have a sufficient security deposit locked in $\FL$. 

        \noindent \textbf{Local variables:} 
        \begin{description}
            \item[$\mathbb{RT}$]: the registration table mapping the content commitment $com_m$ to the committers' IDs and content price.
            \item[$\mathfrak{B}_{max}$]: maximum price for a content.
            \item[$Logs$]: a set of challenges submitted.
            \item[$ct$]: the current global round in the global clock.
        \end{description}
\textbf{API} 
\begin{itemize}[leftmargin=*]
    \item \textit{Register.} Upon receiving $(\textit{register}, sid, com_m, \mathfrak{B}_m)$ from $uid$: query balance $(sid, uid, x) \leftarrow \FL $, if $x \gg \mathfrak{B}_{max} \gg \mathfrak{B}_m$, set $\mathbb{COM}[com_m] = (uid, \mathfrak{B}_m)$.

    \item \textit{PoMM.} Upon receiving $( \textit{pomm}, sid, \pi_{pomm}, tid)$ from a user $uid$: If $PoMM.Ver(\pi_{pomm}, tid) = 1$: $(\textit{transfer}, sid, tid, uid, \mathfrak{B}_{max} ) \rightarrow\FL$. 

    \item \textit{PoME.} Upon receiving $(\textit{pome}, sid, \pi_{pome}, tid)$ from a user $uid$: If $PoME.Ver(\pi_{pome}, tid) = 1$:$(\textit{transfer}, sid, tid, uid, \mathfrak{B}_{max} ) \rightarrow \FL$. 
\end{itemize}
\begin{center}
    \textit{Enforcement Handler} 
\end{center}
\begin{itemize}[leftmargin=*]
    \item \textit{Enforce.} (Round $t_r$) Upon receiving $(\textit{enforce}, sid,  Ch, \Sigma, s' )$ from $\CP$:
    \begin{itemize}[leftmargin=*]
        \item $n:= |Ch.ADDR| - 1$, parse $\{h_1, \ldots, h_n\} := Ch.\mathbb{H}$. Parse identity of all participants from $Ch.ADDR$, $(\CP, \R_1, \ldots, \R_n)$
        \item If $(t_r < Ch.T ) \land (Ch \notin Logs) \land (Open(s', Ch.h_{0}) = 1)$, and ($\Sigma$ contains all parties' lock receipt): $Logs: = Logs \cup \{Ch\}$, broadcast $(\textit{enforced}, sid,   s' )$. Create a new empty log $\mathbb{S}$. 
    \end{itemize}

    \item \textit{Response.} (Round $t_r + n - i + 1$ ) Upon receiving (\textit{log}, $sid$,   $i$, $s_i$) from $R_i$ : if $Open(s_i, h_i) = 1$ and ($\mathbb{S}[i+1] \neq \bot$ or $i = n$) , set $\mathbb{S}[i] := s_i$,  and broadcast $(\textit{logged}, sid, i, \mathbb{S})$ to all parties. 
    If no valid $\textit{log}$ message received in this round, set $\mathbb{S}[i] := \bot$

    \item \textit{Punishment.} (Round $t_r + n + 1$ ) Upon receiving (\textit{punish}, $sid$) from $\CP$, if not all relayers log their secrets, let $\mathcal{R}_x$ be the relayer that fails to log the secret and has the highest index.
    Send $(\textit{transfer}, sid, \R_x, \CP, \mathfrak{B}_{max}) \rightarrow \FL$.
    Upon receiving $(\textit{transferred}, sid, \R_x, \CP,  \mathfrak{B}_{max})$ from $\FL$, broadcast $(\textit{punished}, sid, \R_x, \CP, \mathfrak{B}_{max})$. 
\end{itemize}
\end{mdframed}
\caption{Ideal functionality of \textit{Judge Contract}}
\label{fig:if_contract}
\vspace{-0.2cm}
\end{figure*}

\vspace{0.3em}
\noindent \textbf{Judge Contract.} We present the \textit{Judge Contract} as an ideal functionality $\Fjc$, depicted in Fig. \ref{fig:if_contract}, which manages disputes and facilitates content registration.
For global content integrity, any content $m$ intended for sale must be registered on the \textit{Judge Contract}'s \textit{registration table} $\mathbb{RT}$ using the \textit{register} interface. Each provider and relayer must have a sufficient deposit locked on $\FL$ accessible by $\Fjc$ for accountability.
The \textit{PoMM}/\textit{PoME} interfaces accept valid proofs of misbehavior related to masking/encryption issues involving a provider or a relayer. Upon receiving such proofs of misbehavior, it initiates a compensation process, returning a predetermined amount $\mathcal{B}_\text{max}$ to the customer(compensation amount $\mathcal{B}_\text{max}$ is greater than any content price $\mathfrak{B}_m$, restricted by the \textit{register} interface). 
The \textit{Enforcement Handler} manages the enforcement logic in the \textit{Enforceable A-HTLC}.

\subsection{Protocol Details} \label{sec:construction}

Consider a customer $\C$ seeks to obtain a digital content $m$ (comprises $n$ fixed-size data chunks denoted as $\{m_1, \ldots, m_n\}$) with a price of $\mathfrak{B}_m$. This price is determined by a Merkle root $com_m$ and is offered by a provider $\CP$.
Assuming that multiple relay paths can be found, the content delivery process involves the participation of several relayers who contribute their bandwidth in exchange for relay fees over the PCN.
The delivery graph $\mathcal{G}$ is publicly accessible to all participants.
For ease of exposition, we will first describe the fair content delivery over a singular relay path $p$. Then, we will discuss the extension of this solution to multi-path scenarios.

\subsubsection{\textbf{FairRelay in the single-path scenario}}
The singular path solution operates as follows.

\vspace{0.3em}
\noindent \textbf{Setup Phase.}
$\mathcal{C}$ initiates the protocol by broadcasting an \textit{init} message to $\CP$ and all relayers, triggering the setup phase. 
Denotes all relayers and the provider as $\mathbb{R} \cup \CP$, a node $v$ in  $\mathbb{R} \cup \CP$ generates:  
(1) a symmetric encryption key $sk$ for encrypting all chunks and its commitment $h_{sk}:= Commit(sk)$,
(2) a secret $s$ used to mask the encryption key and its commitment $h_{s}:= Commit(s)$, 
(3) a commitment on masking ($com_{mask}:= (h_{sk}, h_s, ck)_{\sigma}$) including the mask result $ck$ ($ck:= s \oplus sk$) signed by $v$. 
Afterward, $v$ privately sends $com_{mask}$ to $\mathcal{C}$ by encrypting it using $\mathcal{C}$'s public key. Upon receiving all valid commitments on masking and hashes, $\mathcal{C}$ triggers the delivery phase. A formal decryption is demonstrated in Fig. \ref{constrction:setup2}.

\vspace{0.3em}
\noindent \textbf{Delivery Phase.}
To ensure the integrity of the selling content $m$, the provider $\CP$ first sends the hashes ($\mathbb{H}_m$) of all content chunks to $\mathcal{C}$ through the relay path $p$. 
$\mathbb{H}_m$ serves as the digest for the chunks be relayed, ensuring that the final decryption result will match the content $m$. 
 Upon receiving them, $\mathcal{C}$ constructs a Merkle root $com_m'$ from these hashes and checks if $com_m' = com_m$. A match indicates that once the preimages (chunks of the content) for all these hashes are revealed, $\mathcal{C}$ must have received all the chunks of $m$.
Next, $\CP$ encrypts each content chunk using its encryption key and forwards the encrypted chunks to the next hop, along with its encryption commitment. Each relayer then encrypts the incoming chunks again using its own key, attaches its encryption commitment, and forwards them to the next hop. The encryption commitments accumulate with the chunks relayed hop-by-hop, forming a \textit{commitment chain}, where each commitment in this chain mapping to a layer of encryption.

Let's consider the $j$-th chunk $m_j$ as an example. $\CP$ first encrypts it using its own key $sk_0$, generating ciphertext $c_{0, j}$. $\CP$ then constructs a commitment on the encryption, denoted as $com_{enc}^{0,j}$, for this chunk. $\CP$ forwards $c_{0, j}$ and $com_{enc}^{0,j}$ to the first relayer $\R_1$. $\R_1$ encrypts the ciphertext $c_{0, j}$ using its own key $sk_1$, resulting in ciphertext $c_{1, j}$, and generates a commitment $com_{enc}^{1,j}$ for this encryption.
$\R_1$ then forwards $c_{1, j}$ and $\{com_{enc}^{0,j}, com_{enc}^{1,j}\}$ to the customer $\C$. Upon receiving the ciphertext and the commitment chain, $\mathcal{C}$ verifies the following:
 (1) each encryption commitment is properly signed by each encoder, (2) the \textit{commitment chain} linked from the ciphertext $c_{1,j}$ to the $j$-th hash in the hash list $\mathbb{H}_m$ (detailed on customer $\C$'s description in Fig. \ref{constrction:delivery2}). A valid linked \textit{commitment chain} guarantees that once $\C$ gets all the encryption keys committed in this chain, $\C$ can properly decrypts the ciphertext $c_{1,j}$ layer-by-layer, generating chunk $m_j$. Otherwise, $\C$ can locate the dishonest encoder and submit a proof of misbehavior to $\Fjc$, seeking a compensation (detailed on the Extract function in Appendix \ref{alg:extcontent}). 
Fig. \ref{fig:enc} demonstrates how a plaintext chunk is committed and encrypted in a two-hop path, and formal descriptions are demonstrated in Fig. \ref{constrction:delivery2}. 
        
\vspace{0.3em}
\noindent \textbf{Payment Phase.}
To pay the content fee $\mathfrak{B}_m$, there exists a payment path from $\C$ to $\CP$. For better clarity of the protocol description, we consider this payment path as a direct payment channel, as its security is not affected (see Appendix \ref{app:loose}).
Once the relay path completes the delivery of all ciphertext chunks, $\C$ enters the payment phase. $\C$ initiates a conditioned payment of $\mathfrak{B}_m$ to $\CP$, with the unlocking condition that $\CP$ must reveal all mask secrets $s$ from $\R$ and $\CP$ before round $t_0$.
Subsequently, $\CP$ makes a conditioned payment to the first relayer $\R_1$ with an amount equal to the total relay fees in the path. The unlocking condition is that $\R_1$ must reveal all subsequent secrets ($\R_2$ to $\CP$) before round $t_1$. This process stops at the last relayer in the path, where the unlocking condition reduces one hash at each lock (as shown in Fig. \ref{fig:payment_path}).
Once the incoming channel for the last relayer $\R_{|p|}$ is properly locked, $\R_{|p|}$ redeems this payment by calling the \textit{update} interface of $\FC$, revealing its mask secret $s_{|p|}$ to $\R_{|p|-1}$. The secrets accumulated from $\R_{|p|}$ to $\CP$ through the payment channel update process, and once $\CP$'s outcoming channel is updated, $\mathcal{P}$ gathers all necessary secrets to unlock $\mathcal{C}$'s payment, redeeming the content fee $\mathfrak{B}_m$. If any relayer $R_i$ fails to redeem its incoming payment, all its left party ($\R_r$, $r < i$) will refuse to reveal its secrets, expiring the conditioned payment, resulting the termination of the protocol, see Fig. \ref{constrction:payment2}.

\vspace{0.3em}
\noindent \textbf{Decryption Phase.}
    Upon redeem of payment from $\mathcal{C}$, $\mathcal{C}$ obtains all necessary mask secrets to derive the encryption keys:
    $ sk' := ck \oplus s$.
    $\mathcal{C}$ then verifies whether the revealed $sk'$ match the corresponding mask commitment (see Fig. \ref{constrction:decryption2}). If not, $\mathcal{C}$ will generate and submit a proof of misbehavior on masking on-chain to punish the dishonest node and seek a compensation. ($ExtKey$ in Algorithm \ref{alg:extkey}) 
    Once all keys are validated, $\mathcal{C}$ proceeds to decrypt each ciphertext chunk layer by layer using these keys to recover the content $m$.
    At each decryption layer, $\mathcal{C}$ confirms the encryption commitment with the intermediate decryption result. If the encryption commitment is not consistent with the decryption result, $\mathcal{C}$ can submit a proof of misbehavior on encryption on-chain to punish the dishonest node then request a compensation ($Extract$ in Algorithm \ref{alg:extcontent}).  
    
\subsubsection{\textbf{FairRelay in the multi-path scenario}}
Consider the multi-path content delivery $\mathcal{G}$, which consists of $\eta$ relay paths.
In the $k$-th relay path $p_k$, $\R_{k, i}$ denotes the $i$-th relayer, $\mathfrak{B}_{k,i}$ denotes the corresponding relay fee, $t_{k, i}$ denotes the incoming payment timelock to $\R_{k, i}$. 
The maximum length of the relay path determines the publicly known ciphertext delivery deadline $T_1$ and the payment deadline $T_2$ (ensuring protocol \textit{liveness}).
Next, we outline the differences between the multi-path solution and the single-path solution. 

\vspace{0.3em}
\noindent \textbf{Setup Phase. } (Detailed in Fig. \ref{constrction:setup})
In the single-path scenario, $v \in \mathbb{R}\cup \CP$ only needs to send the commitment on masking to $\C$ privately. However, in the multi-path case, $v$ needs to additionally broadcast the hashes of the mask secret ($h_{s} := \text{Commit}(s)$), which will be used to construct the \textit{enforcement challenges} $Ch$ in the \textit{enforceable A-HTLC} scheme. 
Furthermore, $\mathcal{P}$ will generate a synchronizer secret $s_{\text{sync}}$ and broadcast its hash $h_{\text{sync}} := \text{Commit}(s_{\text{sync}})$ to all nodes. This synchronizer secret is used to ensure the atomicity of all payments. 
Each relay path $p_k$ has its own \textit{enforcement challenge} $Ch_k$, which consists of the enforcement deadline $T_2$, the hashes of all mask secrets ($\mathbb{H}_k$) from each relayer in this path, the hash $h_{\text{sync}}$, and the addresses $ADDR_k$ of the provider and all relayers in $p_k$.

\vspace{0.3em}
\noindent \textbf{Delivery Phase.} (Detailed in Fig. \ref{constrction:delivery})
Instead of delivering hashes of all chunks in one path, $\CP$ distributes hashes along with a Merkle multi-proof $\pi_{\text{merkle}}^k$ to $\mathcal{C}$ in each relay path $p_k$. The chunk delivery process remains the same as in the single-path solution.

\vspace{0.3em}
\noindent \textbf{Payment Phase.} (Detailed in Fig. \ref{constrction:payment})
Once all $\eta$ paths have completed the ciphertext delivery, $\mathcal{C}$ enters the payment phase. $\mathcal{C}$ initiates a conditioned payment of $\mathfrak{B}_m$ to $\mathcal{P}$ in exchange for all mask secrets $s$ from each relayer $\mathcal{R}$ and the provider $\mathcal{P}$.
For each relay path, $\mathcal{P}$ locks a payment to the first relayer $\R_{k,1}$ in $p_k$ with an amount of $\sum_{i=1}^{|p_k|}\mathfrak{B}_{k,i}$ (the sum of relay fees in path $p_k$). In the multi-path case, the first relayer additionally is required to reveal $s_{\text{sync}}$ to unlock this payment, enabling $\CP$ to control the initiation of the redeem process.

Once a relayer $\R_{k,i}$ confirms the proper locking of its incoming channel, it acknowledges this by sending a \textit{lock receipt} (a signature on $Ch_k$) back to $\CP$ and locks the next channel. 
Upon collecting all receipts from all relayers, $\mathcal{P}$ broadcasts $s_{\text{sync}}$, allowing $\R_{k, |p_k|}$ to unlock the payment by revealing their own secret along with $s_{\text{sync}}$ to $\FC$. Intermediate relayers append their own secret to unlock their incoming channel upon receiving the corresponding secrets from the update of their outgoing channel.
If all paths are successfully unlocked, $\mathcal{P}$ obtains all secrets to unlock $\mathcal{C}$'s payment.

If a payment path $p_k$ fails to redeem the payment within the specified time, $\CP$ has the option to submit $Ch_k$ (including $s_{\text{sync}}$) and all receipts from the stalled path on-chain. This submission triggers a request for all relayers in the path to reveal their mask secrets on-chain within a designated time window. The revealing process starts from the last relayer and proceeds towards the first, with each intermediate relayer disclosing their secret only after all subsequent relayers have done so.
Upon completion of the revealing process, $\mathcal{P}$ can unlock the payment from $\mathcal{C}$. The first relayer who fails to reveal their secret on-chain will face consequences. This includes a predefined compensation $\mathcal{B}_{\text{max}}$ ($\mathcal{B}_{\text{max}}$ is greater than the sum of relay fees) being transferred from the non-compliant relayer to the provider.

\vspace{0.3em}
\noindent \textbf{Decryption Phase.}
The decryption phase is the same as the single solution, illustrated on Fig. \ref{constrction:decryption}.

\subsubsection{\textbf{Timelocks}}
In this section, we will discuss the timelocks in single-path and multi-path solutions.

\noindent \textit{Single-path scenario:} In the single-path protocol, $\C$ first makes the conditioned payment redeemable until round $t_0$. Once $\CP$ receives this payment, if the payment amount and the redeem condition, including the revelation of $s_0$, are satisfied, $\CP$ makes a new conditioned payment to $\R_1$ with a timelock of $t_0 - 1$. This ensures that once $\R_1$ redeems $\CP$'s payment, $\CP$ has an additional round to redeem the payment from $\CP$, ensuring fairness for $\CP$. Each honest relayer $\R_i$ applies the same logic to handle the incoming conditioned payment, ensuring that the outgoing timelock is earlier than the incoming one.

\sloppy 
\noindent \textit{Multi-path scenario:} In the multi-path protocol, each participant in a delivery instance $\mathcal{G}$ configures a delivery deadline $T_1 := 5 + \max(|p_k|)_{p_k \in \mathcal{G}}$ and an enforcement deadline $T_2 := T_1 + 2\max(|p_k|)_{p_k \in \mathcal{G}} + 5$ based on the length of the longest relay path. Honest $\CP$ accepts its incoming payment only if the corresponding timelock $t_0 := T_2 + \max(|p_k|)_{p_k \in \mathcal{G}} + 2$, and honest $\R_{k,i}$ accepts its incoming payment only if its timelock $t_{k, i} = T_2 + |p_k| - i + 2$. This configuration ensures the following: (1) For a provider or a relayer, the outgoing payment timelock is smaller than the incoming payment timelock. This is evident as $t_0 = T_2 + \max(|p_k|)_{p_k \in \mathcal{G}} + 2 > t_{k, 1} = T_2 + |p_k| + 1$, and $t_{k, i} > t_{k, i+1}$. (2) Once a relayer $\R_{k, i}$ is enforced to reveal its secret on-chain, it has additional rounds to redeem its incoming payment. In our protocol, the provider can only submit an enforcement against a path $p_k$ before round $T_2$. Therefore, the deadline for $\R_{k, i}$ to reveal its secret on-chain is round $T_2 + |p_k| - i$, where $t_{k, i} = T_2 + |p_k| - i + 2 > T_2 + |p_k| - i$.
\begin{figure*}[htbp]
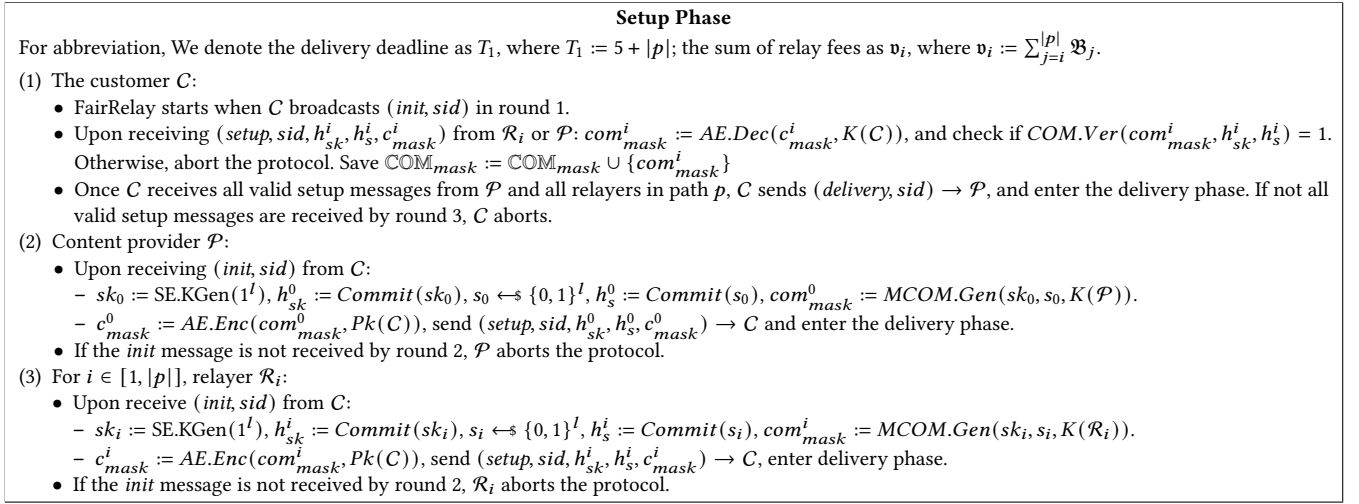

\setlength{\abovecaptionskip}{0.1cm}
    \vspace*{-0.3cm}
    \begin{mdframed}[style=MyFrame]

        \begin{center}
            \textbf{Setup Phase}
        \end{center}
    For abbreviation, We denote the delivery deadline as $T_1$, where $T_1 := 5 + |p|$; the sum of relay fees as $\mathfrak{v}_{i}$, where $\mathfrak{v}_{i}:= \sum_{j=i}^{|p|} \mathfrak{B}_{j}$.

        \begin{enumerate}[leftmargin=*]
            \item The customer $\C$:
                \begin{itemize}
                \item FairRelay starts when $\C$ broadcasts $(\textit{init}, sid)$ in round 1. 
                
                \item Upon receiving $(\textit{setup}, sid, h_{sk}^{ i},h_{s}^{ i},  c_{mask}^{i})$ from $\R_{i}$ or $\CP$: $com_{mask}^{i}: = AE.Dec(c_{mask}^{ i}, K(\C)) $, and check if $COM.Ver(com_{mask}^{ i},  h_{sk}^{ i}, h_{s}^{i}) =1$. Otherwise, abort the protocol. Save $ \mathbb{COM}_{mask}:=  \mathbb{COM}_{mask} \cup \{ com_{mask}^{i} \}$
                
                \item Once $\C$ receives all valid setup messages from $\CP$ and all relayers in path $p$, $\C$ sends $(\textit{delivery}, sid) \rightarrow \CP$, and enter the delivery phase. If not all valid setup messages are received by round 3, $\C$ aborts. 
                \end{itemize}

            \item Content provider $\CP$: 
            \begin{itemize}
                \item Upon receiving $(\textit{init}, sid)$ from $\mathcal{C}$:  
                  \begin{itemize}[leftmargin=*]
                      \item $sk_{0}: =  \text{SE.KGen}(1^l)$, $h_{sk}^0 := Commit(sk_{0})$, $s_{0} \leftarrowS \{0, 1\} ^l $,  $h_{s}^0 := Commit(s_{0})$, $com_{mask}^0 := MCOM.Gen(sk_0, s_0,  K(\mathcal{P}) )$. 
                      \item    $c_{mask}^{0}:= AE.Enc(com_{mask}^{0}, Pk(\C))$, send $(\textit{setup}, sid,h_{sk}^0, h_{s}^0, c_{mask}^{0}) \rightarrow \C$ and enter the delivery phase. 
                  \end{itemize}
                   \item If the \textit{init} message is not received by  round 2, $\CP$ aborts the protocol. 
                  
            \end{itemize}
               
            \item For $i \in [1, |p|]$, relayer $\mathcal{R}_{i}$: 
            \begin{itemize}[leftmargin=*]
                \item  Upon receive $(\textit{init}, sid)$ from $\mathcal{C}$:  
                  \begin{itemize}[leftmargin=*]
                      \item $sk_{i}: = \text{SE.KGen}(1^l)$, $h_{sk}^{i} := Commit(sk_{i})$, $s_{i} \leftarrowS \{0, 1\} ^l $, $h_{s}^{i} := Commit(s_{ i})$, $com_{mask}^{i} := MCOM.Gen(sk_{i}, s_{i}, K(\R_{i})) $. 
                       \item   $c_{mask}^{i}:= AE.Enc(com_{mask}^{i}, Pk(\C))$, send $(\textit{setup}, sid, h_{sk}^{ i},h_{s}^{i}, c_{mask}^{i}) \rightarrow \C$, enter delivery phase. 
                  \end{itemize}
                  \item If the \textit{init} message is not received by round 2, $\R_{i}$ aborts the protocol. 
                  
            \end{itemize}

        \end{enumerate}
    \end{mdframed}
    \caption{Setup Phase of FairRelay in single path scenario}
    \label{constrction:setup2}
\end{figure*}

\begin{figure*}[htbp]
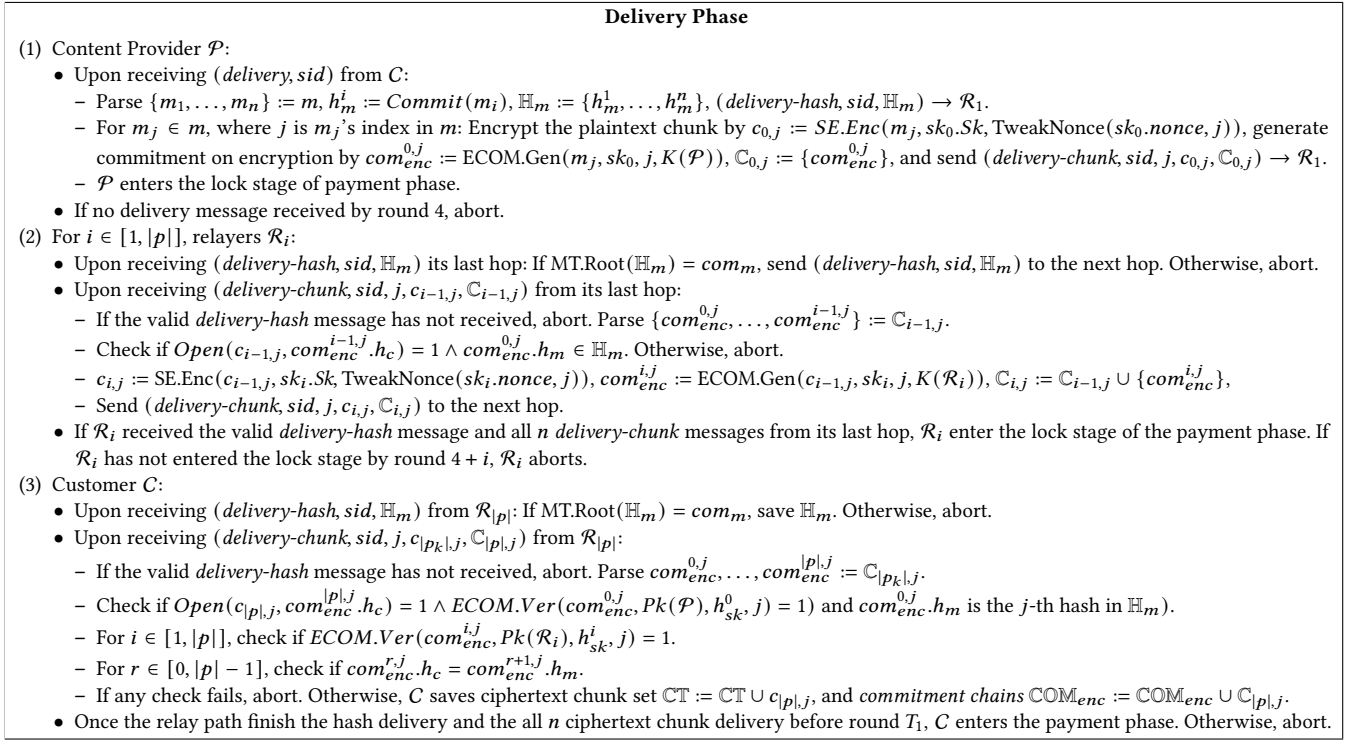

\setlength{\abovecaptionskip}{0.1cm}
    \begin{mdframed}[style=MyFrame]
        \begin{center}
            \textbf{ Delivery Phase }
        \end{center}
    
        \begin{enumerate}[leftmargin=*]

        \item Content Provider $\CP$: 
            \begin{itemize}
                \item Upon receiving $(\textit{delivery}, sid)$ from $\C$: 
                \begin{itemize}[leftmargin=*]
                      \item Parse $\{m_1, \ldots, m_n\}:= m$, $h_m^i := Commit(m_i)$, $\mathbb{H}_m := \{h_m^1, \ldots, h_m^n\}$,  $(\textit{delivery-hash}, sid, \mathbb{H}_m) \rightarrow \mathcal{R}_{1}$. 
                      \item For $m_j \in m$, where $j$ is  $m_j$'s index in $m$: Encrypt the plaintext chunk by $  c_{0,j} := SE.Enc(m_j, sk_0.Sk, \text{TweakNonce}( sk_{0}.nonce , j))$, generate commitment on encryption by $com_{enc}^{0,j} := \text{ECOM.Gen}( m_j,  sk_{0}, j, K(\CP) )$, $\mathbb{C}_{0,j} : = \{com_{enc}^{0,j}\}$, and send $( \textit{delivery-chunk}, sid, j, c_{0,j},\mathbb{C}_{0,j}) \rightarrow \mathcal{R}_{1}$.
        
                        \item $\CP$ enters the lock stage of payment phase. 
                \end{itemize}
                \item If no delivery message received by round 4, abort. 
               
            \end{itemize}
        \item For  $i \in [1, |p|]$, relayers $\R_{ i}$: 
            \begin{itemize}
                \item  Upon receiving $(\textit{delivery-hash}, sid, \mathbb{H}_m)$ its last hop: If $\text{MT.Root}( \mathbb{H}_m) =com_m$,  send $(\textit{delivery-hash}, sid, \mathbb{H}_m)$ to the next hop. Otherwise, abort. 

                \item  Upon receiving $( \textit{delivery-chunk}, sid, j,  c_{i-1,j}, \mathbb{C}_{i-1,j})$ from its last hop: 
                \begin{itemize}[leftmargin=*]
                    \item If the valid $\textit{delivery-hash}$ message has not received, abort. Parse $\{ com_{enc}^{0,j}, \ldots,  com_{enc}^{i-1,j} \} := \mathbb{C}_{i-1,j}$.
                      \item Check if $ Open(c_{i-1,j}, com_{enc}^{i-1,j}.h_c) = 1 \land com_{enc}^{0,j}.h_m \in \mathbb{H}_m$. Otherwise, abort. 
                      \item $c_{i,j} := \text{SE.Enc}( c_{i-1,j}, sk_{ i}.\textit{Sk}, \text{TweakNonce}(sk_{ i}.nonce, j) )$, $com_{enc}^{i,j} := \text{ECOM.Gen}(c_{i-1,j},  sk_{i}, j,  K(\mathcal{R}_{ i}) )$, $\mathbb{C}_{i,j} : = \mathbb{C}_{i-1,j} \cup \{com_{enc}^{i,j} \}$, \item Send $(\textit{delivery-chunk},sid, j, c_{i,j}, \mathbb{C}_{i,j}) $ to the next hop. 
                  \end{itemize}
                \item  If $\R_{ i}$ received the valid $\textit{delivery-hash}$ message and all $n$ \textit{delivery-chunk} messages from its last hop, $\R_{ i}$ enter the lock stage of the payment phase. If $\R_{ i}$ has not entered the lock stage by round $4 + i$, $\R_{ i}$ aborts. 
            \end{itemize}    
        \item Customer $\C$: 
        \begin{itemize}
            \item  Upon receiving  $(\textit{delivery-hash}, sid, \mathbb{H}_m)$ from $\R_{ |p|}$:  If $\text{MT.Root}( \mathbb{H}_m) =com_m$, save $\mathbb{H}_m$. Otherwise, abort.  
               
            \item Upon receiving $( \textit{delivery-chunk}, sid, j,  c_{|p_k|,j}, \mathbb{C}_{|p|,j})$ from $\R_{|p|}$: 
                 \begin{itemize}[leftmargin=*]
                    \item  If the valid $\textit{delivery-hash}$ message has not received, abort. Parse $com_{enc}^{0,j}, \ldots, com_{enc}^{|p|,j} := \mathbb{C}_{|p_k|,j}$. 
                    \item Check if $Open(c_{|p|,j}, com_{enc}^{|p|,j}.h_c) = 1 \land ECOM.Ver(com_{enc}^{0,j}, Pk(\mathcal{P}), h_{sk}^{0}, j ) = 1)$ and  $com_{enc}^{0,j}.h_m $ is the $j$-th hash in  $\mathbb{H}_m)$.
                    \item For $i \in [1, |p|]$, check if $ECOM.Ver(com_{enc}^{i,j}, Pk(\mathcal{R}_{i}), h_{sk}^{i} ,j) = 1$. 
                    \item For $r \in [0, |p|-1]$, check if $com_{enc}^{r,j}.h_c = com_{enc}^{r+1, j}.h_m$.
                    \item If any check fails, abort. Otherwise, $\C$ saves ciphertext chunk set $\mathbb{CT}:= \mathbb{CT} \cup c_{|p|,j} $, and \textit{commitment chains} $\mathbb{COM}_{enc}:= \mathbb{COM}_{enc} \cup \mathbb{C}_{|p|,j} $.
                \end{itemize}
            \item  Once the relay path finish the hash delivery and the all $n$ ciphertext chunk delivery before round $T_1$, $\C$ enters the payment phase. Otherwise, abort. 
        \end{itemize}
    \end{enumerate}

    \end{mdframed}
    \caption{Delivery Phase of FairRelay  in single path scenario}
    \label{constrction:delivery2}
    \vspace{-0.3cm}
\end{figure*}

\begin{figure*}[htbp]
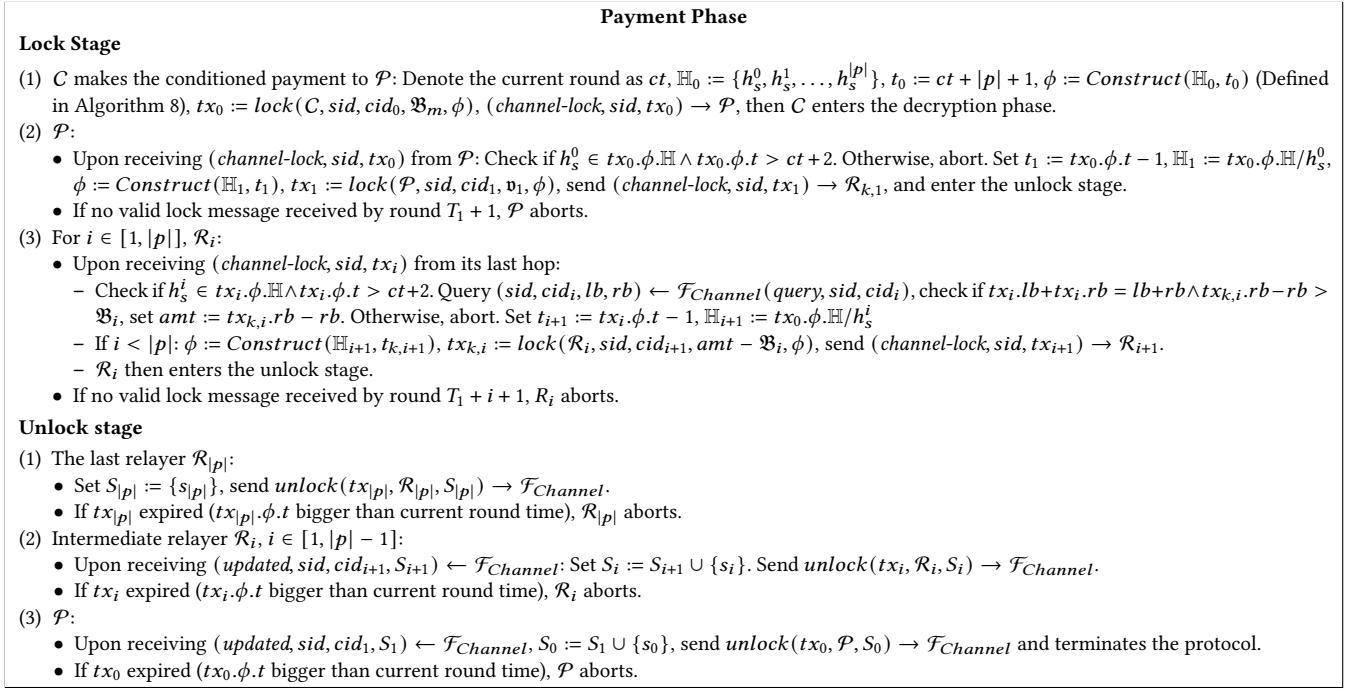
  
\setlength{\abovecaptionskip}{0.1cm}
    \begin{mdframed}[style=MyFrame]
    
        \begin{center}
            \textbf{Payment Phase }
        \end{center}

        \noindent \textbf{Lock Stage}
        \begin{enumerate}[leftmargin=*]
        \item $\C$ makes the conditioned payment to $\CP$: Denote the current round as $ct$, $\mathbb{H}_0 := \{h_s^0, h_s^1, \ldots, h_s^{|p|} \}$, $t_0:= ct + |p| + 1$, $\phi: = Construct(\mathbb{H}_0, t_0)$ (Defined in Algorithm \ref{al:construct}), $tx_0 := lock(\C, sid, cid_0, \mathfrak{B}_m, \phi)$, $(\textit{channel-lock}, sid, tx_0) \rightarrow \CP$, then $\C$ enters the decryption phase. 
       
        \item $\CP$: 
        \begin{itemize}[leftmargin=*]
            \item  Upon  receiving $(\textit{channel-lock}, sid, tx_0)$ from $\CP$: Check if $h_{s}^{0} \in tx_0.\phi.\mathbb{H} \land tx_0.\phi.t > ct + 2$. Otherwise, abort. Set $t_{1}:=  tx_0.\phi.t - 1$, $\mathbb{H}_{ 1} := tx_0.\phi.\mathbb{H}/h_{s}^{0} $, $\phi: = Construct(\mathbb{H}_{ 1}, t_{1})$, $tx_{1} := lock(\CP, sid, cid_{ 1}, \mathfrak{v}_{1}, \phi)$, send $ (\textit{channel-lock}, sid, tx_1) \rightarrow \R_{k, 1}$, and enter  the unlock stage. 
            \item If no valid lock message received by round $T_1 + 1$, $\CP$ aborts. 
        \end{itemize}
        \item For $i \in [1, |p|]$, $\R_{i}$: 
        \begin{itemize}[leftmargin=*]
            \item  Upon receiving $(\textit{channel-lock}, sid,tx_{ i})$ from its last hop: 
                  \begin{itemize}[leftmargin=*]
                          \item Check if $h_{s}^{i} \in tx_{i}.\phi.\mathbb{H} \land tx_i.\phi.t > ct + 2$. Query $(sid, cid_{i}, lb, rb) \leftarrow \FC(\textit{query},sid,cid_{i} )$, check if $tx_{i}.lb + tx_{i}.rb = lb + rb \land tx_{k,i}.rb - rb > \mathfrak{B}_i$, set  $amt := tx_{k,i}.rb - rb$. Otherwise, abort. Set $t_{i + 1}:=  tx_i.\phi.t - 1$, $\mathbb{H}_{ i + 1} := tx_0.\phi.\mathbb{H}/h_{s}^{i} $
                      \item If $i < |p|$:  $\phi: = Construct(\mathbb{H}_{ i+1}, t_{k, i+1})$, $tx_{k, i} := lock(\R_{i}, sid, cid_{i+1}, amt - \mathfrak{B}_{ i}, \phi)$, send $ (\textit{channel-lock}, sid,tx_{ i + 1}) \rightarrow \R_{i+1}$. 
                      \item $\R_{i}$ then enters the unlock stage. 
                  \end{itemize}
            \item  If no valid lock message received by round $T_1 + i + 1$, $R_{i}$ aborts. 
        \end{itemize}

        \end{enumerate}

        \noindent \textbf{Unlock stage}
        \begin{enumerate}[leftmargin=*]
            \item The last relayer $\R_{ |p|}$: 
            \begin{itemize}
                \item Set $S_{ |p|} := \{s_{|p|} \}$, send $unlock(tx_{ |p|}, \R_{|p|},S_{ |p|} )  \rightarrow \FC$. 
                \item If $tx_{|p|}$ expired ($tx_{|p|}.\phi.t$ bigger than current round time), $\R_{ |p|}$ aborts. 
            \end{itemize}
            \item Intermediate relayer $\R_{ i}$, $i \in [1, |p|-1]$: 
            \begin{itemize}
                \item Upon receiving $(\textit{updated}, sid, cid_{i+1}, S_{ i+1}) \leftarrow \FC$: Set  $S_{i} : = S_{ i+1} \cup \{s_{ i}\}$. Send $unlock(tx_{ i}, \R_{i}, S_{ i} )  \rightarrow \FC$.
                \item If $tx_{i}$ expired ($tx_{i}.\phi.t$ bigger than current round time), $\R_{i}$ aborts. 
            \end{itemize}
            \item $\CP$: 
            \begin{itemize}
                \item Upon receiving $(\textit{updated}, sid, cid_{ 1}, S_{ 1}) \leftarrow \FC$, $S_0 := S_{1}  \cup \{s_0\} $,  send $unlock(tx_0, \CP, S_0 )  \rightarrow \FC$ and terminates the protocol. 
                \item If $tx_{0}$ expired ($tx_{0}.\phi.t$ bigger than current round time), $\CP$ aborts. 
            \end{itemize}
             
        \end{enumerate}
    
    \end{mdframed}
    \caption{Payment Phase of FairRelay in single path scenario}
    \label{constrction:payment2}
    \vspace{-0.2cm}
\end{figure*}

\begin{figure*}[htbp]
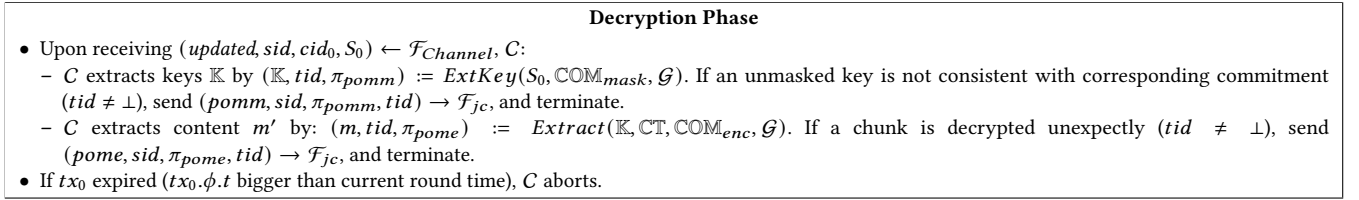

\setlength{\abovecaptionskip}{0.1cm}
    \begin{mdframed}[style=MyFrame]
        \begin{center}
            \textbf{Decryption Phase}
        \end{center}
\begin{itemize}[leftmargin=*]
    \item Upon receiving $(\textit{updated}, sid, cid_{0}, S_0) \leftarrow \FC$, $\C$: 
    \begin{itemize}
        \item  $\C$ extracts keys $\mathbb{K}$ by $(\mathbb{K}, tid, \pi_{pomm} ):= ExtKey(S_0, \mathbb{COM}_{mask}, \mathcal{G})$. If an unmasked key is not consistent with corresponding commitment ($tid \neq \bot$), send $(pomm, sid,  \pi_{pomm}, tid) \rightarrow \Fjc$, and terminate.
        \item $\C$ extracts content $m'$ by:  $(m, tid, \pi_{pome}) : = Extract (\mathbb{K}, \mathbb{CT}, \mathbb{COM}_{enc}, \mathcal{G})$. If a chunk is decrypted unexpectly ($tid \neq \bot$), send $(pome, sid, \pi_{pome}, tid) \rightarrow \Fjc$, and terminate. 
    \end{itemize}
    \item If $tx_{0}$ expired ($tx_{0}.\phi.t$ bigger than current round time), $\C$ aborts. 
\end{itemize}
    \end{mdframed}
    \caption{Decryption Phase of FairRelay in single path scenario}
    \label{constrction:decryption2}
    \vspace{-0.3cm}
\end{figure*}

\vspace{-0.3cm}
\section{Security Analysis in Sketch} \label{sec:security}

The formal security analysis is comprehensively detailed in Appendix \ref{sec:fullproof}. In the analysis, we initially demonstrate that our \textit{Enforceable A-HTLC} protocol \textit{UC-realizes} the ideal functionality $\Fex$. Subsequently, we elucidate the security properties addressed by $\Fex$, which tackle the atomicity issue in multi-path fee-secret exchanges. Finally, we assert the fairness and confidentiality as defined in Section \ref{sec:prob_def}. 
This section provides an intuitive, high-level overview of the security analysis.

Intuitively, the \textit{A-HTLC} payment scheme ensures fairness in a single-path fee-secret exchange: $\CP$ and $\R$ reveal their secrets only when their incoming channels are updated, enabling the transfer of content/relay fees to these nodes. Moreover, $\C$'s payment can only be updated if $\CP$ provides the preimages of the targeted hashes.
The \textit{Enforceable A-HTLC} scheme further guarantees fairness for $\CP$ in a multi-path fee-secret exchange: If any payment issued from $\CP$ is redeemed, $\CP$ can ensure that all payments are redeemed, or alternatively, $\CP$ can guarantee a compensation exceeding the total sum of all relay fees from $\Fjc$.
The commitment and proof of misbehavior in masking/encryption schemes ensure that once $\C$ has acquired all mask secrets, it can obtain the encryption key and decrypt the content $m$ from the ciphertext chunks. The encryption key remains masked by the mask secret, preventing leakage to adversaries and guaranteeing confidentiality.
By combining all these schemes, the requirements for fairness and confidentiality defined in Section \ref{sec:prob_def} are effectively addressed.

\section{Implementation and Evaluation} \label{sec:eval}

\noindent \textbf{Implementation}
We have implemented, tested, and evaluated our decentralized content delivery protocol within a simulated environment\footnote{FairRelay source code: https://github.com/7ujgt6789i/FairRelay}.
Regarding encryption for content chunks, we adopt \textit{Ciminion} \cite{dobraunig2021ciminion}, a scheme compatible with zero-knowledge proofs. For the generation and verification of the zk-SNARK proofs, we utilize the \textit{Groth16} \cite{groth2016size} proof system, with the corresponding circuits designed in \textit{Circom} \cite{circom}.
\begin{table}[ht]
    \begin{center}
        \centering
        \begin{tabular}{ccccc}
                \hline
                Contract & Operation & Tx Gas & ETH/\$ & Optimism L2/\$ \\
                \hline
                \multirow{7}{*}{\textbf{JC}} & Deploy & 2,421,188 & 93.44 & 0.55 \\
                & Join & 46,175 & 1.78 & 0.01 \\
                & Leave & 28,783 & 1.11 & 0.01 \\
                & Withdraw & 25,966 & 1.00 & 0.01 \\
                & Add & 90,924 & 3.51 & 0.02 \\
                & Remove & 30,159 & 1.16 & 0.01 \\
                & PoMM & 35429 &1.37 & 0.01 \\
                & PoME & 290,797 &11.22 & 0.07 \\
                & Enforce & 267,325 & 10.31 & 0.07 \\  
                & Response & 24,902 & 0.96 & 0.01 \\
                & Punish & 27,458 & 1.06 & 0.01 \\
                \hline
                \multirow{4}{*}{\textbf{PC}} & Deploy & 1,215,909 & 46.92 & 0.28 \\
                & Update & 90,912 & 3.50 &0.02\\
                & Close & 43,611 & 1.68 & 0.01 \\
                & Withdraw & 22,547 & 0.87 & 0.01 \\
                \hline
                \textbf{ENS} & Register & 266,996 & 10.30 & 0.06 \\
                \textbf{USDT} & Transfer & 54,128 & 2.09 & 0.01 \\
                \textbf{ERC20} & Deploy & 1,311,213 & 50.60 & 0.30 \\
                \hline
        \end{tabular}
        \caption{Gas cost of on-chain operations: on-chain section consists of Judge Contract (JC) and Payment Channel (PC), USD cost on Ethereum (ETH) and Optimism L2. Compare with deploying an ERC-20 contract, registering an ENS domain, and performing a USDT transfer.}
        \label{tab:gas-cost}
    \end{center}    
    \vspace{-1.05cm}
\end{table}

\subsection{On-chain Evaluation}
Table \ref{tab:gas-cost} provides an overview of the on-chain gas costs and their corresponding USD costs for all operations in our protocol. The table assumes that the price of Ether is set to 2270.13 USD (as of February 5, 2024). Additionally, the gas prices used in the calculations are set to 17 GWei on the Ethereum mainnet and 0.1 GWei on the Optimism Rollup \cite{optimism}.

\vspace{0.3em}
\noindent \textbf{One-time Costs.} 
The deployment of the \textit{Judge Contract} is a one-time global occurrence and incurs a gas cost of 2,421,188. In contrast, deploying a simple ERC-20 Contract typically costs 1,311,213 gas.
For each content provider or relayer participating in decentralized content delivery and earning profits, there is a total one-time gas cost of 107,099, which covers joining, leaving, and withdrawing deposits from the network.
The total one-time gas cost for deploying a payment channel contract is 1,282,069. It's worth noting that a payment channel can be utilized for multiple off-chain payments.
If a content provider intends to deliver content using our protocol, a one-time content register operation is necessary, incurring a gas cost of 90,924. In comparison, it costs roughly 266,996 gas to register a domain name in ENS \footnote{Ethereum Name Service: https://ens.domains/}. In summary, the one-time costs of FairRelay are acceptable. 

\vspace{0.3em}
\noindent \textbf{Optimistic Costs.} If content delivery concludes without any dispute, our protocol requires no on-chain operations, resulting in zero on-chain costs for all participants in a content delivery job. The optimistic costs are consistently zero and independent of the content chunk size, content chunk number, and the number of participants.

\vspace{0.3em}
\noindent \textbf{Pessimistic Costs.} If a customer observes misbehavior by certain nodes during the decryption phase, they have the option to file a complaint with the \textit{Judge contract}. The gas cost of verifying the proof of misbehavior \textit{PoME} on encryption on-chain remains constant at 290,797 gas (equivalent to 0.07 USD on Optimism L2) and is unaffected by the content chunk size, the number of content chunks, or the number of participants. Similarly, the gas cost of verifying the proof of misbehavior on masking \textit{PoMM} on-chain is also constant at 35,429 gas (0.01 USD). Consider a withdraw operations in Tornado Cash \cite{tornado} cost 301,233 gas (0.07 USD), the pessimistic costs of our protocol are acceptable.

\subsection{Off-chain Efficiency}
All efficiency tests are conducted on a customer-level laptop with 32GB memory and Intel i7-9750H CPU(2.60GHz).

\begin{figure*} 
    \centering    
    \setlength{\abovecaptionskip}{0.cm}
    \subfigure[Time costs of PoME with different chunk sizes.] {
        \label{fig:zkp_eff}     
        \includegraphics*[width = .305\textwidth]{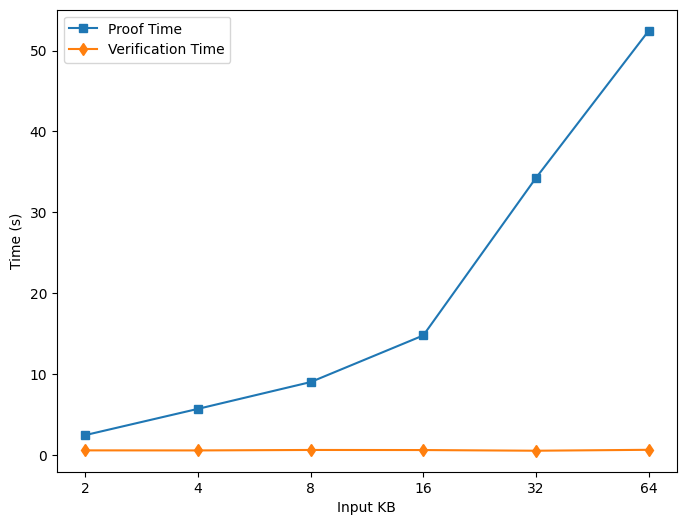}
    }     
    \subfigure[Memory costs of PoME with different chunk sizes.] { 
        \label{fig:zkp_eff}     
        \includegraphics*[width = .305\textwidth]{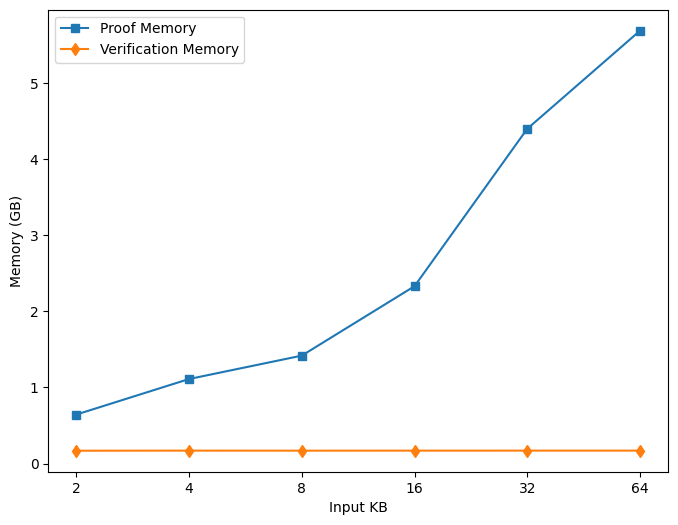} 
    }    
    \subfigure[Bandwidth usage transferring a 1GB file within one path: 32 bytes hash size and 65 bytes signature size.] {
        \label{fig:relay_eff}
        \includegraphics[width=.31\textwidth]{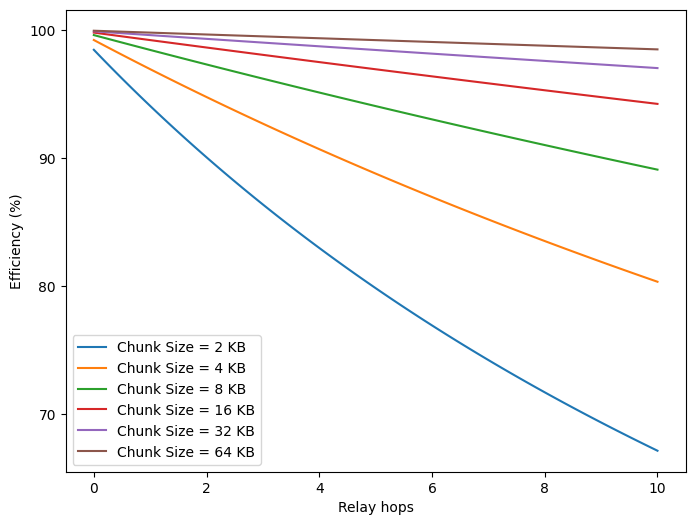}    
    }
    \caption{Off-chain Evaluation}     
    \label{fig:eva}     
\end{figure*}

\vspace{0.3em}
\noindent \textbf{zk-SNARK Efficiency. } 
In our protocol, only one zero-knowledge proof needs to be generated and verified in the proof of misbehavior on encryption stage. Fig. \ref{fig:zkp_eff} illustrates the time and memory costs on generating and verifying the \textit{PoME} along with content chunk size. 
For instance, when proving misbehavior on encrypting a 2KB chunk, the circuit contains 26,633 constraints. In comparison, FileBounty \cite{filebounty} requires 29,339 constraint circuits to verify a 64-byte chunk, which demonstrates the efficiency of our implementation. 
The time and memory costs of generating and verifying \textit{PoME} increase linearly with the size of the circuits to be proven, whereas the proof size, verification memory usage, and verification time remain constant are independent of the content chunk size, approximately 805B, 170MB, and 0.6s, respectively.

\vspace{0.3em}
\noindent \textbf{Encryption Efficiency.} To accelerate the proof generation in the \textit{PoME} generation, we use a zero-knowledge proof friendly symmetric encryption scheme \textit{Ciminion} in our protocol. The efficiency of off-chain encryption/decryption is about 110KB/s per thread, which could be accelerated through multi-threading optimization.

\subsection{Overhead} 
To demonstrate the efficiency of FairRelay, we initially compare our protocol with blockchain-based fair exchange protocols, FairSwap\cite{dziembowski_fairswap_2018}, FDE\cite{tas_atomic_nodate}, Bitstream\cite{linus_bitstream_nodate}, and FairDownload\cite{9929262}. Since most of protocols do not support multiple relayers, we compare the overheads of these protocols with FairRelay in scenarios involving the content provider $\CP$ and the customer $\C$. Table \ref{tab:compare} illustrates that our protocol achieves zero optimistic on-chain cost and constant pessimistic on-chain cost in two-party fair exchange.

\begin{table*}[ht]
  \resizebox{0.98\textwidth}{!}{%
  \begin{tabular}{@{}ccccccccc@{}}
    \toprule
  Protocol &
    Rounds &
    \begin{tabular}[c]{@{}l@{}}Commitment Scheme\end{tabular} &
    P-\textgreater C comm. &
    \begin{tabular}[c]{@{}l@{}}Onchain operations (Opt., Pess.) \end{tabular} &
    Optimistic On-chain cost &
    Pessimistic On-chain cost &
    Dynamic chunk size &
    Content Integrity \\ \midrule
  FairSwap &
    5 &
    Merkle Tree &
    \begin{tabular}[c]{@{}l@{}}$|m| + |h|$\end{tabular} &
    4,4 &
    \begin{tabular}[c]{@{}l@{}} $O(1)$ \end{tabular} & 
    \begin{tabular}[c]{@{}l@{}}$ O(log(n))$ \end{tabular} & 
    \color[HTML]{308014}\checkmark  &
    \color[HTML]{308014}\checkmark \\
   FairBounty &
    n &
    Merkle-Damgård &
    \begin{tabular}[c]{@{}l@{}}$|m| + n|h|$\end{tabular} &
    3,5 &
    \begin{tabular}[c]{@{}l@{}} $O(n)$\end{tabular} & 
    \begin{tabular}[c]{@{}l@{}} $O(n)$ \end{tabular} & 
    \color[HTML]{308014}\checkmark &
    \color[HTML]{B22222}\xmark \\
   FDE-ElGammal &
    3 &
    KZG &
    \begin{tabular}[c]{@{}l@{}} $8|m| + 6 \mathbb{G}$ \end{tabular} &
    4,4 &
    \begin{tabular}[c]{@{}l@{}} $O(1)$\end{tabular} & 
    \begin{tabular}[c]{@{}l@{}}$O(1)$\end{tabular} & 
    \color[HTML]{308014}\checkmark &
    \color[HTML]{308014}\checkmark \\
  FairDownload &
    n &
    Merkle tree &
    \begin{tabular}[c]{@{}l@{}}$|m| + 2n |\sigma| + (2n-2)|h| $ \end{tabular} &
    4,6 &
    \begin{tabular}[c]{@{}l@{}} $O(1)$ \end{tabular} &
   $ O(log(n))$ &
    \color[HTML]{308014}\checkmark &
    \color[HTML]{308014}\checkmark \\
  Bitstream &
    3 &
    Merkle Tree &
    \begin{tabular}[c]{@{}l@{}}$2|m| + |h| + |\sigma| $  \end{tabular} &
    0,1 &
    0 &
    \begin{tabular}[c]{@{}l@{}}$ O(log(n) ) $ \end{tabular} & 
    \color[HTML]{B22222}\xmark &
    \color[HTML]{308014}\checkmark \\
\hline
  \textbf{Our protocol} &
    3 &
    Merkle Tree &
    \begin{tabular}[c]{@{}l@{}} $|m| + (2n + 3)|h| + (n+1)|\sigma|  $\end{tabular} &
    0,2 &
    0 &
    \begin{tabular}[c]{@{}l@{}} $O(1)$ \end{tabular} & 
    \color[HTML]{308014}\checkmark &
    \color[HTML]{308014}\checkmark \\ \bottomrule
  \end{tabular}%
  }
  \caption{Compare FairRelay with related works. $n$ denotes the number of content chunks in $m$, $|h|$ denotes the size of hash , $|\mathbb{G}|$ denotes a group element,  $|\sigma|$ denotes the signature size, $|m|$ denotes the size of the content. Content integrity means $\C$ only pays to $\CP$ if the whole content is delivered correctly. In FDE-ElGammal, the ciphertext size $|ct|$ of a content chunk is about 8 times of the plaintext size \cite{tas_atomic_nodate}; In Bitstream, the content size is fixed to 32-bytes.}
  \label{tab:compare}
\vspace{-0.5cm}
\end{table*}

Next, we consider the overheads in relation to the number of hops in a relay path. For instance, let's consider a 1GB content $m$ relayed through an $n_r$-hop relay path. Before $\CP$ delivers the ciphertext, it conveys the hashes of each chunk to $\C$. Subsequently, each hop encrypts its incoming data and appends an encryption commitment $com_{enc} := (h_m, h_c, h_{sk}, id, \sigma)$. Since $h_m$ is included in the last encryption commitment and $h_{sk}$ is received during the setup phase, the index info is already reveal in the merkle proof stage, the size of $com_{enc}$ is $|h| + |\sigma|$. Along the relay path, the data overheads increase linearly with the number of previous hops: $total\_{overheads} =  n_r|com_{enc}| + setup\_overheads$. Fig. \ref{fig:relay_eff} depicts the bandwidth efficiency (i.e., the bandwidth used to transfer ciphertext divided by the total bandwidth usage) as a function of the number of hops for different chunk sizes. In the practical range of zero-knowledge chunk sizes (2KB - 64KB), it is feasible to select a chunk size that achieves high bandwidth efficiency.

\section{Related Works} \label{sec:related_works}

There are two technological approaches to achieve fair P2P content delivery. One is using a centralized trusted third party (TTP) to observe the resource usages and distribute fees. Such solutions include FloodGate \cite{nair_floodgate_2008}, Meson \cite{meson} and Saturn \cite{saturn}. Once the centralized trusted third party is compromised, their approach will fail. Thus, many researchers turn to using a decentralized trusted third party, usually a blockchain, yielding a dencentralized approach. The problem we study in the P2P content delivery context is related to the research of blockchain-based fair exchange, fair exchange using payment channels and atomic multiple channel updates. Here we review the related works in these research lines.

\vspace{0.3em}
\noindent \textbf{Blockchain-based Fair Exchange.}
Protocols \cite{ 9929262,  dziembowski_fairswap_2018, zkcp, eckey_optiswap_2020, filebounty,   ma_decentralized_2023,  tas_atomic_nodate} such as Zero-Knowledge Contingent Payment (ZKCP)  and FairSwap \cite{dziembowski_fairswap_2018} epitomize the blockchain-based fair exchange paradigm, wherein the content provider $\mathcal{P}$ encrypts the content $m$ into a verifiable ciphertext, and then finalizes the payment by disclosing the decryption key to $\mathcal{C}$. 
ZKCP employs zero-knowledge proofs to ensure the correctness of $\mathcal{P}$'s encryption, whereas FairSwap utilizes a proof of misbehavior scheme, thereby reducing the cost introduced by zero-knowledge proofs. 
Subsequent researches, including OptiSwap \cite{eckey_optiswap_2020}, FileBounty \cite{filebounty}, and \textit{Fair Data Exchange} \cite{tas_atomic_nodate}, have focused on reducing the on-chain and off-chain overheads.
Nevertheless, even in the absence of disputes, these protocols invariably necessitate at least one on-chain transaction to conclusively settle the fair exchange.

\vspace{0.3em}
\noindent \textbf{Payment Channel and Two-Party Fair Exchanges.}
Notably, a straightforward HTLC-based off-chain payment embodies an off-chain fair exchange, wherein the payer pays the payee in exchange for a predefined preimage.  Bitstream\cite{linus_bitstream_nodate} capitalizes on this concept, utilizing the HTLC preimage as a means of delivering the decryption key. However, Bitstream does not include relayers, which cannot be directly extended to solve the problem addressed in this paper.

\vspace{0.3em}
\noindent \textbf{Atomic Multi-channel Updates.} 
Atomic multi-channel update schemes \cite{aumayr2021blitz, malavolta_concurrency_2017, aumayr_thora_2022, malavolta_anonymous_2019, egger_atomic_2019, miller_sprites_2017} primarily focus on achieving atomic settlement of payments within payment channels.
Single-path multi-channel update schemes like AMHL \cite{malavolta_anonymous_2019} and Blitz \cite{aumayr2021blitz} provide "strong atomicity" \cite{malavolta_anonymous_2019} for all sub-payments by  replacing HTLC with homomorphic one-way functions and signatures, respectively. 
Multi-path multi-channel update schemes, such as Sprites \cite{miller_sprites_2017} and Thora \cite{aumayr_thora_2022}, offer a more versatile approach by utilizing a global on-chain event to ensure atomicity for all channel updates.
However, these schemes primarily address the atomicity of one-way payments and do not fully support two-way exchanges. For instance, Thora's "revoke-unless-all-paid" paradigm \cite{aumayr_thora_2022} is not suitable for the key release process since once a key is released, it cannot be revoked.
Consequently,  existing atomic multi-channel update frameworks are not well-suited for atomic multi-party exchanges within PCNs.

\vspace{0.3em}

FairDownload \cite{9929262} is most relevant to our work. 
Like our proposed scheme, FairDownload asks every relayer to encrypt the data for transmission to the subsequent node. However, unlike our work which employs multiple relayers per path, their scheme is restricted to delivery jobs where there is only one relayer in each delivery path. What's more, compare with our protocol, FairDownload incurs high on-chain costs even in the optimistic scenario.

 \vspace{-0.3 em}
\section{Conclusion}    \label{sec:conclusion}

In this work, we presents FairRelay, the first fair and cost-efficient P2P content delivery protocol built on payment channel networks.
We transform the fair content delivery problem into an \textit{atomic multi-party fee-secret exchange} problem. On top of that, we propose a multi-hop payment scheme based on \textit{A-HTLC}, ensuring that all participants get corresponding fees when the content is delivered by a single path. When the content is delivered via multiple paths, we design \textit{Enforceable A-HTLC}, which involves a \textit{acknowledge-pay-enforce} procedure to enforce the atomicity  of all payments. 

FairRelay has several positive impacts. On the one hand, FairRelay motivates owners of payment channels to use their idle bandwidth for profit, enhancing active participation in PCNs. The liquidity of locked funds in PCNs are also increased when those funds are used for relay fee payments.
On the other hand, \textit{Enforceable A-HTLC} provides solutions for some related problems in PCNs. For example, it can be used to enforce atomic multi-path payments.

Here, we outline several interesting questions for future work. 
\begin{itemize}[leftmargin=*,  topsep=3pt]
  
    \item \textit{Eliminating on-chain deposits. } An interesting optimization involves freeing the provider and relayers from making on-chain deposits, thereby lower the barrier to be a relayer.
    \item \textit{Fault tolerance. } Current FairRelay requires all relay paths to complete their tasks for fee settlement. A fault-tolerant payment scheme that accommodates node failures would be desirable.
    \item \textit{Defending front-running attack.}   FairRelay makes use of on-chain deposits for punishment. A possible attack is that an adversary drains its deposit before an honest party asks for compensation (known as front-running attack). In this work, we address this problem by slashing some compensation and require a large deposit to defend such attack (more discussion in Appendix \ref{dis:front-running}). In the future, we are interested in exploring other approaches to defend the attack, for example, removing the use of  on-chain deposits.  
\end{itemize}

\bibliographystyle{ACM-Reference-Format}
\bibliography{dCDN}

\appendix 
\section{FairRelay in multi-path scenario}
Fig. \ref{constrction:setup}, \ref{constrction:delivery}, \ref{constrction:payment} and \ref{constrction:decryption} formally demonstrate the FairRelay protocol in the multi-path scenario. 

\begin{figure*}[htbp]
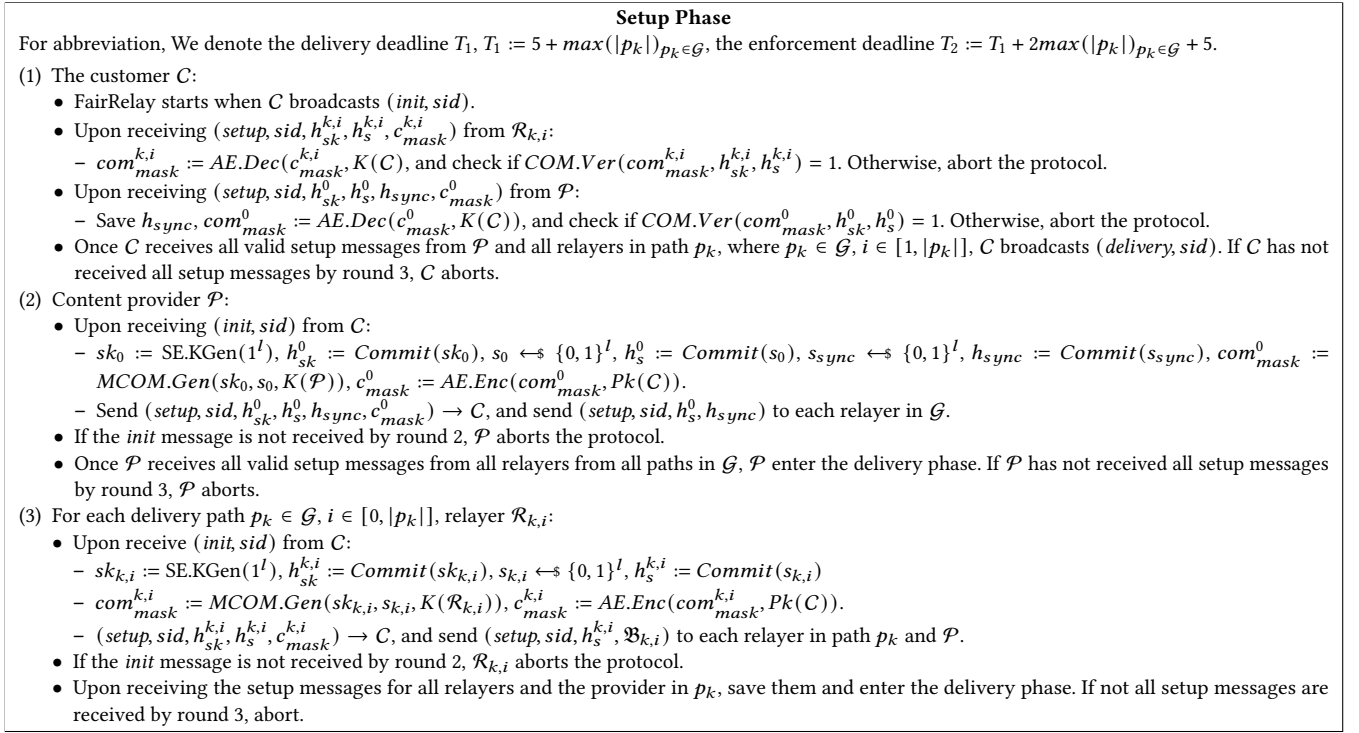

    \vspace*{-0.3cm}
    \begin{mdframed}[style=MyFrame]

        \begin{center}
            \textbf{Setup Phase}
        \end{center}
    For abbreviation, We denote the delivery deadline $T_1$,  $T_1 := 5 + max(|p_k|)_{p_k \in \mathcal{G}} $, the enforcement deadline $T_2:= T_1 + 2max(|p_k|)_{p_k \in \mathcal{G}}  + 5$. 

        \begin{enumerate}[leftmargin=*]
            \item The customer $\C$:
                \begin{itemize}
                \item FairRelay starts when $\C$ broadcasts $(\textit{init}, sid)$.
                    
                \item Upon receiving $(\textit{setup}, sid, h_{sk}^{k, i},h_{s}^{k, i},  c_{mask}^{k, i})$ from $\R_{k, i}$: 
                \begin{itemize}[leftmargin=*]
                    \item $com_{mask}^{k, i}:= AE.Dec(c_{mask}^{k, i}, K(\C) $, and check if $COM.Ver(com_{mask}^{k, i},  h_{sk}^{k, i}, h_{s}^{k, i}) =1$. Otherwise, abort the protocol. 
                \end{itemize}
                \item Upon receiving $(\textit{setup}, sid, h_{sk}^{0},h_{s}^{0}, h_{sync}, c_{mask}^{0})$ from $\CP$: 
                \begin{itemize}[leftmargin=*]
                    \item Save $h_{sync}$, $com_{mask}^{0}:= AE.Dec(c_{mask}^{0}, K(\C))$, and check if $COM.Ver(com_{mask}^{0},  h_{sk}^{0}, h_{s}^{0}) =1$. Otherwise, abort the protocol. 
                \end{itemize}
                \item Once $\C$ receives all valid setup messages from $\CP$ and all relayers in path $p_k$, where $p_k \in \mathcal{G}$, $i \in [1, |p_k|]$, $\C$ broadcasts $(\textit{delivery}, sid)$. If $\C$ has not received all setup messages by round 3, $\C$ aborts. 
                \end{itemize}

            \item Content provider $\CP$: 
            \begin{itemize}
                \item  Upon receiving $(\textit{init}, sid)$ from $\mathcal{C}$:  
                  \begin{itemize}[leftmargin=*]
                      \item $sk_{0}: =  \text{SE.KGen}(1^l)$, $h_{sk}^0 := Commit(sk_{0})$, $s_{0} \leftarrowS \{0, 1\} ^l $,  $h_{s}^0 := Commit(s_{0})$, $s_{sync}\leftarrowS \{0, 1\} ^l $, $h_{sync} := Commit(s_{sync}) $, $com_{mask}^0 := MCOM.Gen(sk_0, s_0,  K(\mathcal{P}) )$, $c_{mask}^{0}:= AE.Enc(com_{mask}^{0}, Pk(\C))$. 
                      \item  Send $(\textit{setup}, sid,h_{sk}^0, h_{s}^0, h_{sync},  c_{mask}^{0}) \rightarrow \C$, and send $(\textit{setup}, sid,  h_{s}^0, h_{sync}) $ to each relayer in $\mathcal{G}$. 
                  \end{itemize}
                   \item If the \textit{init} message is not received by round 2, $\CP$ aborts the protocol. 
                   \item Once $\CP$ receives all valid setup messages from  all relayers from all paths in $\mathcal{G}$, $\CP$ enter the delivery phase. If $\CP$ has not received all setup messages by round 3, $\CP$ aborts. 
            \end{itemize}
               
            \item For each delivery path $p_k \in \mathcal{G}$, $i \in [0, |p_k|]$, relayer $\mathcal{R}_{k,i}$: 
            \begin{itemize}[leftmargin=*]
                \item  Upon receive $(\textit{init}, sid)$ from $\mathcal{C}$:  
                  \begin{itemize}[leftmargin=*]
                      \item $sk_{k,i}: = \text{SE.KGen}(1^l)$, $h_{sk}^{k,i} := Commit(sk_{k, i})$, $s_{k,i} \leftarrowS \{0, 1\} ^l $, $h_{s}^{k,i} := Commit(s_{k, i})$
                       \item $com_{mask}^{k,i} := MCOM.Gen(sk_{k,i}, s_{k,i}, K(\R_{k,i})) $,  $c_{mask}^{k,i}:= AE.Enc(com_{mask}^{k,i}, Pk(\C))$.
                       \item $(\textit{setup}, sid, h_{sk}^{k, i},h_{s}^{k, i}, c_{mask}^{k,i}) \rightarrow \C$, and send $(\textit{setup}, sid, h_{s}^{k, i}, \mathfrak{B}_{k, i} ) $ to each relayer in path $p_k$ and $\CP$. 
                  \end{itemize}
                  \item If the \textit{init} message is not received by round 2, $\R_{k,i}$ aborts the protocol. 
                  \item  Upon receiving the setup messages for all relayers and the provider in $p_k$, save them and enter the delivery phase. If not all setup messages are received by round 3, abort. 
            \end{itemize}

        \end{enumerate}
    \end{mdframed}
    \caption{Setup Phase of FairRelay in multi-path scenario.}
    \label{constrction:setup}
\end{figure*}

\begin{figure*}[htbp]
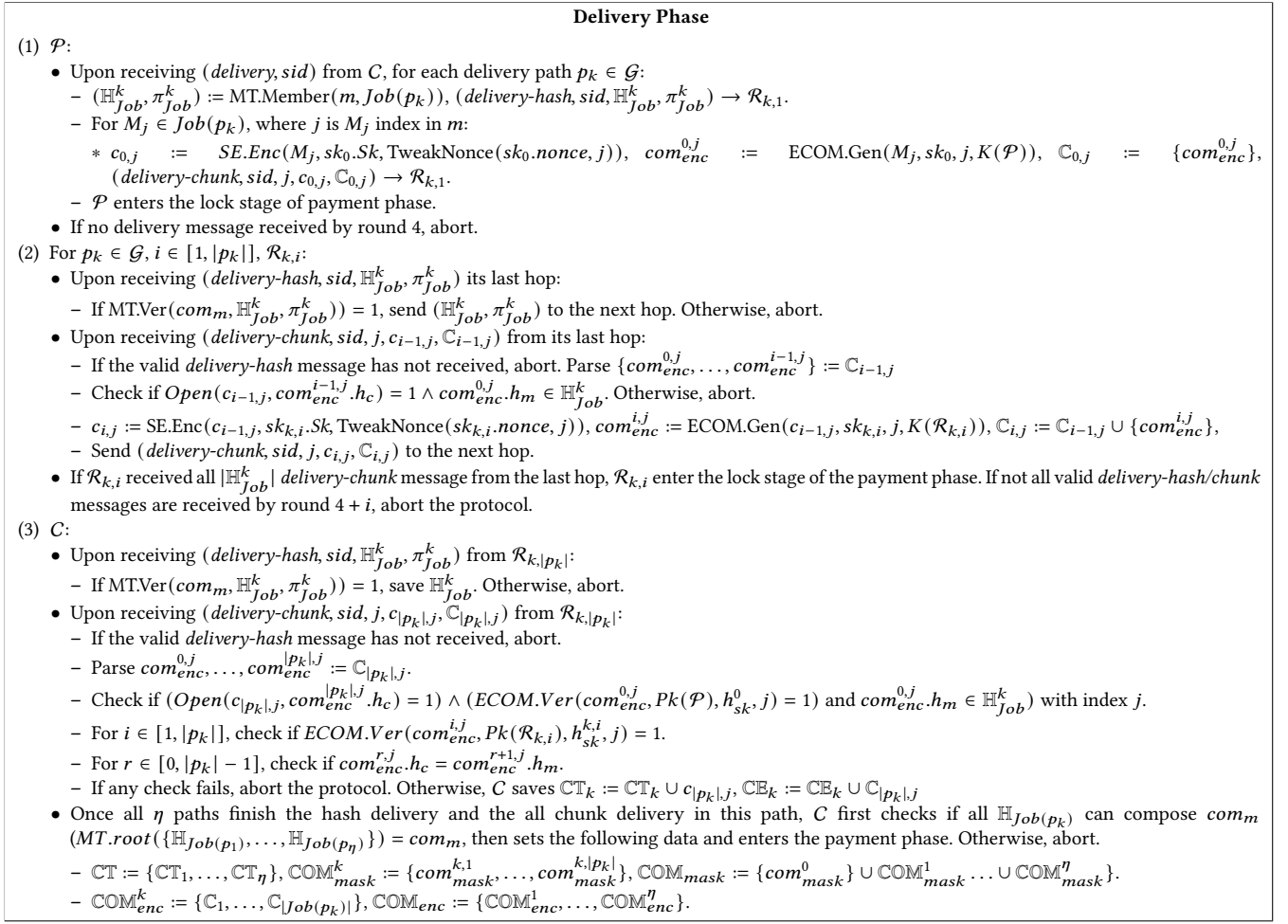

    \vspace*{-0.1cm}
    \begin{mdframed}[style=MyFrame]
        \begin{center}
            \textbf{ Delivery Phase }
        \end{center}
    
        \begin{enumerate}[leftmargin=*]

        \item $\CP$: 
            \begin{itemize}
                \item Upon receiving $(\textit{delivery}, sid)$ from $\C$, for  each delivery path $p_k \in \mathcal{G}$: 
                \begin{itemize}[leftmargin=*]
                      \item $(\mathbb{H}_{Job}^{k}, \pi_{Job}^{k}): = \text{MT.Member}(m, Job(p_k))$,  $(\textit{delivery-hash}, sid, \mathbb{H}_{Job}^{k}, \pi_{Job}^{k}) \rightarrow \mathcal{R}_{k, 1}$. 
                      \item For $M_j \in Job(p_k)$, where $j$ is  $M_j$ index in $m$:
                          \begin{itemize}[leftmargin=*]
                              \item $  c_{0,j} := SE.Enc(M_j, sk_0.Sk, \text{TweakNonce}( sk_{0}.nonce , j))$, $com_{enc}^{0,j} := \text{ECOM.Gen}( M_j,  sk_{0}, j, K(\CP) )$, $\mathbb{C}_{0,j} : = \{com_{enc}^{0,j}\}$, $( \textit{delivery-chunk}, sid, j, c_{0,j},\mathbb{C}_{0,j}) \rightarrow \mathcal{R}_{k,1}$.
                          \end{itemize}
                        \item $\CP$ enters the lock stage of payment phase. 
                \end{itemize}
                \item  If no delivery message received by round 4, abort. 
               
            \end{itemize}
        \item For $p_k\in \mathcal{G}$, $i \in [1, |p_k|]$, $\R_{k, i}$: 
            \begin{itemize}
                \item  Upon receiving $(\textit{delivery-hash}, sid, \mathbb{H}_{Job}^{k}, \pi_{Job}^{k})$ its last hop: 
                \begin{itemize}[leftmargin=*]
                      \item If $\text{MT.Ver}(com_m,\mathbb{H}_{Job}^{k},\pi_{Job}^{k}) ) =1$,  send $(\mathbb{H}_{Job}^{k}, \pi_{Job}^{k})$ to the next hop. Otherwise, abort. 
                \end{itemize}
                \item  Upon receiving $( \textit{delivery-chunk}, sid, j,  c_{i-1,j}, \mathbb{C}_{i-1,j})$ from its last hop: 
                \begin{itemize}[leftmargin=*]
                    \item If the valid $\textit{delivery-hash}$ message has not received, abort. Parse $\{ com_{enc}^{0,j}, \ldots,  com_{enc}^{i-1,j} \} := \mathbb{C}_{i-1,j}$
                      \item Check if $ Open(c_{i-1,j}, com_{enc}^{i-1,j}.h_c) = 1 \land com_{enc}^{0,j}.h_m \in \mathbb{H}_{Job}^{k}$. Otherwise, abort. 
                      \item $c_{i,j} := \text{SE.Enc}( c_{i-1,j}, sk_{k, i}.\textit{Sk}, \text{TweakNonce}(sk_{k, i}.nonce, j) )$, $com_{enc}^{i,j} := \text{ECOM.Gen}(c_{i-1,j},  sk_{k,i}, j,  K(\mathcal{R}_{k, i}) )$, $\mathbb{C}_{i,j} : = \mathbb{C}_{i-1,j} \cup \{com_{enc}^{i,j} \}$, \item Send $(\textit{delivery-chunk},sid, j, c_{i,j}, \mathbb{C}_{i,j}) $ to the next hop. 
                  \end{itemize}
                \item If $\R_{k, i}$ received all $|\mathbb{H}_{Job}^{k}|$ \textit{delivery-chunk} message from the last hop, $\R_{k, i}$ enter the lock stage of the payment phase.  If not all valid $\textit{delivery-hash/chunk}$ messages are received by round $4 + i$, abort the protocol. 
            \end{itemize}    
        \item $\C$: 
        \begin{itemize}
            \item Upon receiving $(\textit{delivery-hash}, sid, \mathbb{H}_{Job}^{k}, \pi_{Job}^{k})$ from $\R_{k, |p_k|}$: 
                \begin{itemize}[leftmargin=*]
                      \item If $\text{MT.Ver}(com_m,\mathbb{H}_{Job}^{k},\pi_{Job}^{k}) ) =1$, save $\mathbb{H}_{Job}^{k}$. Otherwise, abort. 
                \end{itemize}
            \item  Upon receiving $( \textit{delivery-chunk}, sid, j,  c_{|p_k|,j}, \mathbb{C}_{|p_k|,j})$ from $\R_{k, |p_k|}$: 
                 \begin{itemize}[leftmargin=*]
                    \item If the valid $\textit{delivery-hash}$ message has not received, abort.
                    \item Parse $com_{enc}^{0,j}, \ldots, com_{enc}^{|p_k|,j} := \mathbb{C}_{|p_k|,j}$. 
                    \item Check if $(Open(c_{|p_k|,j}, com_{enc}^{|p_k|,j}.h_c) = 1) \land (ECOM.Ver(com_{enc}^{0,j}, Pk(\mathcal{P}), h_{sk}^{0}, j ) = 1)$ and  $com_{enc}^{0,j}.h_m \in \mathbb{H}_{Job}^{k})$ with index $j$. 
                    \item For $i \in [1, |p_k|]$, check if $ECOM.Ver(com_{enc}^{i,j}, Pk(\mathcal{R}_{k, i}), h_{sk}^{k, i},j) = 1$. 
                    \item For $r \in [0, |p_k|-1]$, check if $com_{enc}^{r,j}.h_c = com_{enc}^{r+1, j}.h_m$.
                    \item If any check fails, abort the protocol. Otherwise, $\C$ saves $\mathbb{CT}_k:= \mathbb{CT}_k \cup c_{|p_k|,j} $, $\mathbb{CE}_k:= \mathbb{CE}_k \cup \mathbb{C}_{|p_k|,j} $
                \end{itemize}
            \item Once all $\eta$ paths finish the hash delivery and the all chunk delivery in this path, $\C$ first checks if  all $\mathbb{H}_{Job(p_k)}$ can compose $com_m$ ($MT.root(\{\mathbb{H}_{Job(p_1)} , \ldots, \mathbb{H}_{Job(p_{\eta})} \} ) = com_m$, then sets the following data and enters the payment phase. Otherwise, abort. 
                \begin{itemize}[leftmargin=*]
                    \item $\mathbb{CT}:= \{ \mathbb{CT}_1, \ldots,\mathbb{CT}_{\eta} \}$, $ \mathbb{COM}_{mask}^k:=  \{com_{mask}^{k, 1}, \ldots, com_{ mask}^{k, |p_k| } \}$, $ \mathbb{COM}_{mask}:=  \{com_{mask}^{0}\} \cup \mathbb{COM}_{mask}^1 \ldots \cup  \mathbb{COM}_{mask}^{\eta} \}$.
                    \item $\mathbb{COM}_{enc}^k:= \{ \mathbb{C}_1, \ldots, \mathbb{C}_{|Job(p_k)|} \}$, $\mathbb{COM}_{enc}:= \{ \mathbb{COM}_{enc}^1, \ldots, \mathbb{COM}_{enc}^{\eta} \}$.
                \end{itemize}
        \end{itemize}
    \end{enumerate}

    \end{mdframed}
    \caption{Delivery Phase of FairRelay in multi-path scenario.}
    \label{constrction:delivery}
\end{figure*}

\begin{figure*}[htbp]
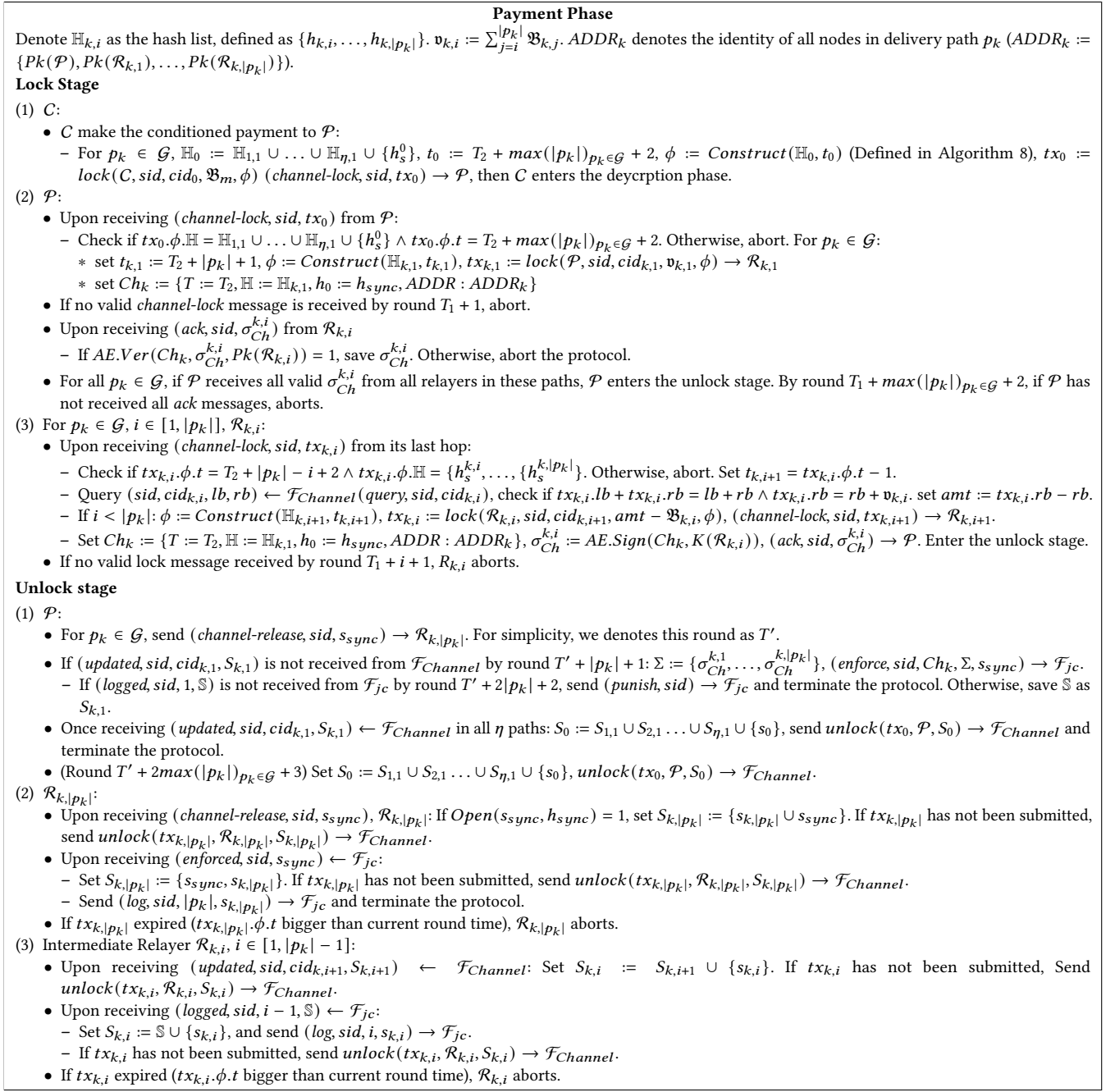
  
    \begin{mdframed}[style=MyFrame]
    
        \begin{center}
            \textbf{Payment Phase }
        \end{center}

    \noindent
    Denote  $\mathbb{H}_{k, i}$ as the hash list, defined as $\{h_{k,i}, \ldots, h_{k,|p_k|}\}$. $\mathfrak{v}_{k, i}:= \sum_{j=i}^{|p_k|} \mathfrak{B}_{k, j}$.
    $ADDR_k$ denotes the identity of all nodes in delivery path $p_k$ ($ADDR_k:= \{Pk(\CP), Pk(\R_{k, 1}), \ldots, Pk(\R_{k, |p_k|}) \}$).

        \noindent \textbf{Lock Stage}
        \begin{enumerate}[leftmargin=*]
        \item $\C$: 
        \begin{itemize}[leftmargin=*]
            \item  $\C$ make the conditioned payment to $\CP$: 
                  \begin{itemize}[leftmargin=*]
                      \item For $p_k \in \mathcal{G}$, $\mathbb{H}_0 := \mathbb{H}_{1, 1} \cup  \ldots \cup \mathbb{H}_{\eta, 1} \cup \{ h_s^0\}$, $t_0:= T_2 + max(|p_k|)_{p_k \in \mathcal{G}} + 2$, $\phi: = Construct(\mathbb{H}_0, t_0)$ (Defined in  Algorithm \ref{al:construct}), $tx_0 := lock(\C, sid, cid_0, \mathfrak{B}_m, \phi)$ $(\textit{channel-lock}, sid, tx_0) \rightarrow \CP$, then $\C$ enters the deycrption phase. 
                  \end{itemize}
        \end{itemize}
        \item $\CP$: 
        \begin{itemize}[leftmargin=*]
            \item  Upon  receiving $(\textit{channel-lock}, sid, tx_0)$ from $\CP$: 
                  \begin{itemize}[leftmargin=*]
                      \item Check if $ tx_0.\phi.\mathbb{H} = \mathbb{H}_{1, 1} \cup  \ldots \cup \mathbb{H}_{\eta, 1} \cup \{ h_s^0\} \land tx_0.\phi.t = T_2 + max(|p_k|)_{p_k \in \mathcal{G}} + 2$. Otherwise, abort. 
                      For $p_k \in \mathcal{G}$: 
                        \begin{itemize}[leftmargin=*]
                            \item set  $t_{k,1}:= T_2 + |p_k| + 1$, $\phi: = Construct(\mathbb{H}_{k, 1}, t_{k,1})$, $tx_{k, 1} := lock(\CP, sid, cid_{k, 1}, \mathfrak{v}_{k,1}, \phi) \rightarrow \R_{k, 1}$
                            \item set $Ch_k:= \{ T:= T_2, \mathbb{H}:= \mathbb{H}_{k,1}, h_{0}:= h_{sync}, ADDR: ADDR_k \}$
                        \end{itemize}
                  \end{itemize}

            \item  If no valid \textit{channel-lock} message is received by round  $T_1 + 1$, abort. 
            \item  Upon receiving $(\textit{ack}, sid, \sigma_{Ch}^{k,i}) $ from  $\mathcal{R}_{k, i}$
            \begin{itemize}[leftmargin=*]
                \item If $AE.Ver(Ch_k, \sigma_{Ch}^{k,i}, Pk(\mathcal{R}_{k,i})) = 1$, save $ \sigma_{Ch}^{k,i}$. Otherwise, abort the protocol. 
            \end{itemize}
            \item For all $p_k \in \mathcal{G}$, if $\mathcal{P}$ receives all valid $\sigma_{Ch}^{k,i}$ from all relayers in these paths, $\CP$ enters the unlock stage.  By round $T_1 + max(|p_k|)_{p_k \in \mathcal{G}} + 2$, if $\CP$ has not received all \textit{ack} messages, aborts. 
        \end{itemize}
        \item For $p_k \in \mathcal{G}$, $i \in [1, |p_k|]$, $\R_{k, i}$: 
        \begin{itemize}[leftmargin=*]
            \item  Upon receiving $(\textit{channel-lock}, sid,tx_{k, i})$ from its last hop: 
                  \begin{itemize}[leftmargin=*]
                      
                          \item Check if $tx_{k, i}.\phi.t = T_2 + |p_k| - i + 2 \land tx_{k, i}.\phi.\mathbb{H} = \{h_{s}^{k, i}, \ldots, \{h_{s}^{k, |p_k|}\}$. Otherwise, abort. Set $ t_{k, i+1} = tx_{k, i}.\phi.t -1$. 
                          \item Query $(sid, cid_{k,i}, lb, rb) \leftarrow \FC(\textit{query},sid,cid_{k,i} )$, check if $tx_{k,i}.lb + tx_{k,i}.rb = lb + rb \land  tx_{k,i}.rb = rb +\mathfrak{v}_{k,i} $. set $amt := tx_{k,i}.rb - rb$. 
                      \item If $i < |p_k|$:  $\phi: = Construct(\mathbb{H}_{k, i+1}, t_{k, i+1})$, $tx_{k, i} := lock(\R_{k,i}, sid, cid_{k, i+1}, amt - \mathfrak{B}_{k, i}, \phi)$, $ (\textit{channel-lock}, sid,tx_{k, i+1 }) \rightarrow \R_{k, i+1}$. 
                      \item Set  $Ch_k:= \{ T:= T_2, \mathbb{H}:= \mathbb{H}_{k,1}, h_{0}:= h_{sync}, ADDR: ADDR_k \}$, $\sigma_{Ch}^{k,i} := AE.Sign( Ch_k, K(\mathcal{R}_{k,i}))$, $(\textit{ack}, sid, \sigma_{Ch}^{k,i}) \rightarrow \CP$. Enter the unlock stage. 
                  \end{itemize}
            \item  If no valid lock message received by  round $T_1 + i + 1$, $R_{k, i}$ aborts. 
        \end{itemize}

        \end{enumerate}

        \noindent \textbf{Unlock stage}
        \begin{enumerate}[leftmargin=*]
            \item $\CP$: 
            \begin{itemize}
                \item For $p_k \in \mathcal{G}$, send $(\textit{channel-release},sid, s_{sync}) \rightarrow \R_{k, |p_k|}$. For simplicity, we denotes  this round as $T'$.  
                
                \item If $(\textit{updated}, sid, cid_{k,1}, S_{k,1})$ is not received from $\FC$ by round  $T' + |p_k| + 1$: $\Sigma := \{\sigma_{Ch}^{k,1}, \ldots, \sigma_{Ch}^{k,|p_k|} \}$, $(\textit{enforce}, sid, Ch_k, \Sigma, s_{sync} ) \rightarrow \Fjc$. 
                \begin{itemize}
                    \item  If $(\textit{logged}, sid, 1, \mathbb{S})$ is not received from $\Fjc$ by round  $T' + 2|p_k| + 2$, send $(\textit{punish}, sid) \rightarrow \Fjc$ and terminate the protocol. Otherwise, save $\mathbb{S}$ as $S_{k, 1}$. 
                \end{itemize}

                \item Once receiving $(\textit{updated}, sid, cid_{k, 1}, S_{k, 1}) \leftarrow \FC$ in all $\eta$ paths: $S_0 := S_{1, 1} \cup S_{2, 1} \ldots \cup S_{\eta, 1} \cup \{s_0\} $,  send $unlock(tx_0, \CP,S_0 )  \rightarrow \FC$ and terminate the protocol. 
                \item (Round $T' + 2max(|p_k|)_{p_k \in \mathcal{G}} + 3$) Set $S_0 := S_{1, 1} \cup S_{2, 1} \ldots \cup S_{\eta, 1} \cup \{s_0\} $,  $unlock(tx_0, \CP,S_0 )  \rightarrow \FC$.
                
            \end{itemize}
            \item $\R_{k, |p_k|}$: 
            \begin{itemize}
                \item Upon receiving $(\textit{channel-release},sid, s_{sync})$, $\mathcal{R}_{k, |p_k|}$: If $Open(s_{sync}, h_{sync}) = 1$, set $S_{k, |p_k|} := \{s_{k, |p_k|} \cup s_{sync} \}$.  If $tx_{k,|p_k|}$ has not been submitted, send $unlock(tx_{k, |p_k|}, \R_{k,|p_k|},S_{k, |p_k|} )  \rightarrow \FC$. 
                  \item Upon receiving $(\textit{enforced}, sid, s_{sync}) \leftarrow \Fjc$: 
                  \begin{itemize}[leftmargin=*]
                    
                    \item Set $ S_{k, |p_k|}:= \{s_{sync}, s_{k, |p_k|} \}$. If $tx_{k, |p_k|}$ has not been  submitted, send $unlock(tx_{k, |p_k|}, \R_{k,|p_k|}, S_{k, |p_k|} )  \rightarrow \FC$. 
                    \item Send $(\textit{log}, sid,  |p_k|,  s_{k, |p_k|}) \rightarrow \mathcal{F}_{jc}$ and terminate the protocol. 
                \end{itemize}
                \item If $tx_{k, |p_k|}$ expired ($tx_{k, |p_k|}.\phi.t$ bigger than current round time), $\R_{k, |p_k|}$ aborts. 
            \end{itemize}
            \item Intermediate Relayer $\R_{k, i}$, $i \in [1, |p_k|-1]$: 
            \begin{itemize}
                \item Upon receiving $(\textit{updated}, sid, cid_{k, i+1}, S_{k, i+1}) \leftarrow \FC$: Set  $S_{k, i} : = S_{k, i+1} \cup \{s_{k, i}\}$. If $tx_{k, i}$ has not been  submitted, Send $unlock(tx_{k, i}, \R_{k,i}, S_{k, i} )  \rightarrow \FC$.
                 \item Upon receiving $(\textit{logged}, sid, i-1, \mathbb{S})\leftarrow \Fjc$:
                    \begin{itemize}[leftmargin=*]
                        \item Set $ S_{k, i}:=  \mathbb{S} \cup \{ s_{k, i} \}$, and send $(\textit{log}, sid, i, s_{k, i}) \rightarrow \mathcal{F}_{jc}$. 
                         \item If $tx_{k,i}$ has not been submitted, send $unlock(tx_{k, i}, \R_{k,i}, S_{k, i} )  \rightarrow \FC$. 
                    \end{itemize}
                \item If $tx_{k,i}$ expired ($tx_{k, i}.\phi.t$ bigger than current round time), $\R_{k, i}$ aborts. 
            \end{itemize}

        \end{enumerate}
    
    \end{mdframed}
    \caption{Payment Phase of FairRelay}
    \label{constrction:payment}
\end{figure*}

\begin{figure*}[htbp]
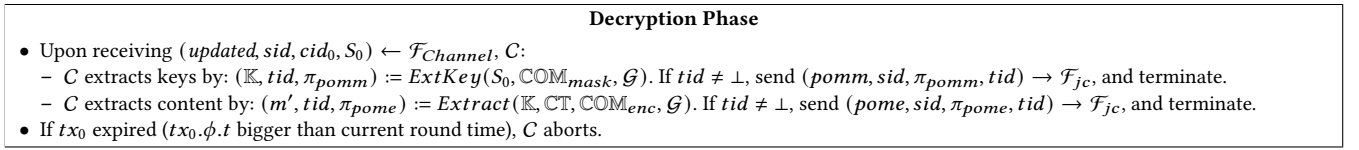

    \begin{mdframed}[style=MyFrame]
        \begin{center}
            \textbf{Decryption Phase}
        \end{center}
\begin{itemize}[leftmargin=*]
    \item Upon receiving $(\textit{updated}, sid, cid_{0}, S_0) \leftarrow \FC$, $\C$: 
    \begin{itemize}
        \item  $\C$ extracts keys by: $(\mathbb{K}, tid, \pi_{pomm} ):= ExtKey(S_0, \mathbb{COM}_{mask}, \mathcal{G})$. If $tid \neq \bot$, send $(pomm, sid,  \pi_{pomm}, tid) \rightarrow \Fjc$, and terminate.
        \item $\C$ extracts content by:  $(m', tid, \pi_{pome}) : = Extract (\mathbb{K}, \mathbb{CT}, \mathbb{COM}_{enc}, \mathcal{G})$. If $tid \neq \bot$, send $(pome, sid, \pi_{pome}, tid) \rightarrow \Fjc$, and terminate. 
    \end{itemize}
    \item If $tx_{0}$ expired ($tx_{0}.\phi.t$ bigger than current round time), $\C$ aborts. 
\end{itemize}
    \end{mdframed}
    \caption{Decryption Phase of FairRelay}
    \label{constrction:decryption}
\end{figure*}

\section{Auxiliary primitives, building blocks and functions} \label{app:pritmives}
In this section, we first outline the construction of cryptographic primitives and building blocks essential to our protocol. Subsequently, we elaborate on the auxiliary functions employed in the FairRelay protocol.

\subsection{Primitives}

\noindent \textbf{Global Random Oracle.} We utilize a \textit{restricted-programmable }  global random oracle, modeled as ideal functionality $\FR$. $\FR$ responds to all queries with uniformly random sampled $\lambda$-bit string $r\leftarrowS \{0, 1\}^\lambda$. 
Every party has oracle access to the global functionality $\FR$, and $\FR$ is implemented by a hash function in the real world. 
$\FR$ will logs all query response pairs $(q, r)$ in a query history set (denoted as $Q$. If a new query is logged in $Q$ before, such that $(q, r) \in Q$, $\FR$ will return the logged random $r$. In the global random oracle settings, $\FR$ would return the same random value for all the same query no matter which session it comes from. 
\textit{restricted-programmable} refers to the capability of an ideal adversary $\Sim$ to inspect and alter the log history of the random oracle within the confines of its session, while an honest user retains the ability to verify if a query history mapping $(q, r)$ has been programmed. We devise the \textit{restricted-programmable} global random oracle in accordance with the framework delineated in \cite{ro1}, as delineated below.

\vspace{20pt}

\begin{mdframed}[userdefinedwidth= \linewidth,frametitlealignment=\center,  frametitle={ Random Oracle Ideal Functionality $\FR$}, frametitlerule=true, frametitlebackgroundcolor=gray!20, linecolor=gray!50, font=\small,  innertopmargin=2pt, innerbottommargin=2pt, innerrightmargin=5pt, innerleftmargin=5pt, font=\small ]

    \textbf{Local Var}: $\FR$ maintains the input history query set $Q$, mapping query value $q \in \{0, 1\} ^ *$ to the output $r  \in \{0, 1\} ^ \lambda $. 
    $\FR$ also stores set $P$ as programmed set and $Q_{sid}$ for all sessions $sid$. 
    
    \noindent \textbf{Query.} Upon receiving $( \textit{query}, \textit{sid}, q)$ from $sid'$:
    \begin{itemize}[leftmargin=*]
        \item If $(sid, q, r) $ no in $ Q$, randomly sample $r \leftarrowS \{0, 1\} ^ \lambda $, add $(sid, q, r)$ to $Q$. Output $(\textit{query}, \textit{sid}, q, r)$.
        \item If $sid' \neq sid$ or query from a adversary , store $(q, r)$ in $Q_{sid}$
    \end{itemize}
    
    \noindent \textbf{Program.}
    \begin{itemize}[leftmargin=*]
        \item  Upon receiving $(\textit{program}, sid, q, r)$ from a adversary $\mathcal{A}$ in session $sid$, if there exists a record $(sid, q, r')$ in $Q$, abort. Otherwise, if $r \in \{0,1\}^ \lambda$ store $(sid, q, r)$ in $Q$ and $sid, q)$ in $P$. 
        \item Upon receiving $(\textit{isProgrammed}, sid, q)$ from a party of session $sid$, if $(sid, q) \in P$, return $(\textit{programmed}, sid, q)$
    \end{itemize}
    
    \noindent \textbf{Observe} Upon receiving $(\textit{lookup}, sid)$ from the adversary of sesion $sid$, respond with $(\textit{lookuped}, sid, Q_{sid})$
\end{mdframed}

\vspace{10pt}

\noindent \textbf{Commitment Scheme.} We build a commitment scheme based on $\FR$, defined in Algorithm \ref{algo:com}: 

\vspace{10pt}

\begin{breakablealgorithm}
    \small
    \caption{Commitment Scheme} \label{algo:com}
    \begin{algorithmic}
    \Function{Commit}{x}    \Comment{$x \in \{0, 1\} ^*$}
        \State  $d \leftarrowS \{0, 1\} ^K \text{s.t } \FR (\textit{isProgrammed},sid, x || d ) \neq 1$  \Comment{ Choose an un-corrupted $d$}
        \State $c \gets \FR( \textit{query}, sid, x || d)$
        \State \Return $(c, d)$
    \EndFunction
    \Statex
    \Function{Open}{m, h}  
    \State parse $(c, d) \gets h$   \Comment{ $c \in \{0, 1\} ^\lambda, d\in \{0, 1\}^K $}
    \State $c' = \FR(\textit{query}, sid, x||d)$
    \If{$c = c' \land \FR(\textit{isProgrammed},sid, x || d ) \neq 1$}
    \State \Return 1
    \EndIf 
    \State \Return 0
    \EndFunction    
    \end{algorithmic}
\end{breakablealgorithm}

\vspace{0.3em}
\noindent \textbf{Encryption Scheme.} We utilize a symmetric encryption scheme $SE$, which satisfies the IND-CPA security \cite{goldreich2009foundations}, and provides following interfaces. 

\begin{itemize}[leftmargin=*]
    \item $(Sk, nonce):= SE.KGen(1^\mu)$, where $Sk$ is the master encryption key and $nonce$ is randomly sampled. We use $sk$ to denotes this tuple, $l$ is the security parameters where $|sk| = \mu$
    \item $c:= SE.Enc(m, sk.Sk, r)$. Encrypt a message $m$ with a random nonce $r$.  $\forall c:= SE.Enc(m, sk.Sk, r)$, $c':= SE.Enc(m', sk.Sk, r')$, if $m' \neq m \land r' \neq r $, then this encryption holds the IND-CPA security. 
    \item $m:= SE.Dec(c, sk.Sk, r)$. Decrypt a ciphertext $c$ uses a master encryption key $sk.Sk$ and a nonce $r$. 
\end{itemize}

\vspace{0.3em}
\noindent \textbf{Asymmetric encryption scheme.}
We use an asymmetric encryption scheme captured by an ideal functionality $\Fsig$, presented in work \cite{SigUC}.   $\Fsig$ has following interfaces: 
\begin{itemize}[leftmargin=*]
    \item $(iPk, iK) := AE.KGen(\lambda)$. Generate keypair with security parameter $\lambda$, where $iPk$ is the public key, $iK$ is the secret key. 
    \item $\sigma :=AE.Sign( m, iK)$. Sign message $m$ and generate signature $\sigma$. 
    \item $AE.Ver(m, \sigma, iPk) \mapsto \{0, 1\}$. Verify the signature $\sigma$ of $m$ with a public key $iPk$. 
    \item $c:= AE.Enc(m, iPk)$ Encrypt a message $m$ with public key $iPk$. 
    \item $m:= AE.Dec(c, iK)$ Encrypt a ciphertext $c$ with private key $iK$. 
\end{itemize}

\vspace{0.3em}
\noindent \textbf{zk-SNARK.} A \textit{zk-SNARK} system $Snk$ can demonstrate that a specific witness (comprising a public witness $x$ and a private witness $w$) validates a given claim $C$, where $C(x, w) = 1$, without disclosing the confidential witness. This system upholds the principles of \textit{correctness}, \textit{soundness}, and \textit{zero knowledge} \cite{groth2016size}. $Snk$ has following interfaces:  
\begin{itemize}[leftmargin=*]
    \item $(vp, pp) :=   Snk.\text{Setup}(C) $ Setup the \textit{verification parameter} $vp$ and \textit{proving parameter} $pp$. 
    \item  $\pi := Snk.\text{Proof}(pp, x, w) $. Create a proof $\pi$ claiming $C(x,w) =1$. 
    \item $Snk.\text{Verify}(vp, x, \pi) \mapsto \{0, 1\}$. 
\end{itemize}

\vspace{0.3em}
\noindent \textbf{Merkle Multi-proof.} A \textit{Merkle Multi-proof} system $MT$ can prove a set $x$ is a subset of $y$, where $y$ is commitment by a merkle root \cite{merklemulti}. The \textit{Merkle Multi-proof} $MT$ provides following interfaces: 
\begin{itemize}[leftmargin=*]
    \item $root := MT.Root(y)$. Compute the Merkle root of an ordered list $y$.
    \item $ (\mathbb{H}, \pi_{merkle} ) := MT.Member(x, y)$. Generate a proof $\pi_{merkle}$ claiming $x$ is a subset of $y$, where $\mathbb{H}$ is the hashes of all elements in $x$. 
    \item $MT.Ver(\mathbb{H}, \pi_{merkle}, root) \mapsto \{0, 1\} $. Verify if $\mathbb{H}$ is a subset of merkle leaves committed by $root$. 
\end{itemize}

\subsection{Building Blocks}
\subsubsection{Commitment and PoM on Masking}\label{app:pomm}

Here we formally define the commitment scheme \textit{MCOM} (shown in Algorithm \ref{algo:mcom}) and the proof of misbehavior scheme \textit{POMM} on masking (shown in Algorithm \ref{algo:pomm}). 

\vspace{1cm}
\begin{breakablealgorithm}
    \small
    \caption{Mask Commitment}   \label{algo:mcom}
    \begin{algorithmic}[1]
        \Function{MCOM.Gen}{$sk, s, iK$}
            \State $h_{sk} \gets Commit(sk)$ 
            \State $h_s \gets Commit(s)$
            \State $ck \gets s \oplus sk$ \Comment{ mask result}
            \State $\sigma \gets AE.Sign(\{h_{sk}, h_s, ck\}, iK)$ 
            \State $com_{mask} \gets \{h_{sk}, h_s, ck, \sigma\}$
            \State \Return $com_{mask}$
        \EndFunction
        \Statex
        \Function{MCOM.Ver}{$com_{mask}, iPk, h_{sk}', h_s'$}
        \State Parse $\{h_{sk}, h_s, ck, \sigma\} \gets com_{mask}$
        \If{$AE.Ver(\{h_{sk}, h_s, ck\}, \sigma, iPk) = 1 \land h_{sk}' = h_{sk} \land h_s' = h_s$ }
        \State \Return 1
        \EndIf 
        \State \Return 0
        \EndFunction

    \end{algorithmic}
\end{breakablealgorithm}

\vspace{0.3cm}
\begin{breakablealgorithm}
    \small
    \caption{Proof of Misbehavior on Masking (POMM)}    \label{algo:pomm}
    \begin{algorithmic}[1]
        \Function{PoMM.Gen}{$com_{mask}, iPk, s'$}
            \State $tid \gets iPk$
            \State Parse $\{h_{sk}, h_s, ck, \sigma\} \gets com_{mask}$
            \State $\pi_{pomm}   \gets (h_{sk}, h_s, ck, \sigma, s') $
            \State \Return $(\pi_{pomm}, tid)$
        \EndFunction
        \Statex
        \Function{PoMM.Ver}{$\pi_{pomm}, tid$}
            \State Parse $ (h_{sk}, h_s, ck, \sigma, s') \gets \pi_{pomm} $
            \If{$AE.Ver(\{h_{sk}, h_s, ck\}, \sigma, tid) = 1 \land Open(s', h_s) = 1 \land Open(ck \oplus s', h_{sk} ) = 1$}
            \State \Return 1
            \EndIf 
            \State \Return 0 
        \EndFunction
    \end{algorithmic}
\end{breakablealgorithm}

We now formally define the properties secured by the $MCOM$ and $PoMM$. 

\begin{lemma}   \label{lemma:pomm_binding}
     For any mask commitment $com_{mask}$ correctly signed by a public key $iPk$, and the secret $s$ is correctly revealed: If $com_{mask}$ is accurately generated from function $MCOM.Gen()$, the correct encryption key $sk$ committed by $com_{mask}.h_{sk}$ can be extracted. Otherwise, a valid $\pi_{pomm}$ could be generated. In other words, $P_1$ with \textit{probability} 1 when the committer is honest, $P_2$ with \textit{probability} 1 when the committer is dishonest. The sum of $P_1$ and $P_2$ has a \textit{probability} of 1.

\end{lemma}

\[ \small P_1 = 
\Pr\left[
\begin{array}{l}
Open(sk,h_{sk}) = 1
\end{array} 
\middle| 
\begin{array}{lcl}
(h_{sk}, h_s, ck, \sigma) \gets com_{mask} \\
AE.Ver(\{h_{sk}, h_s, ck\}, \sigma, iPk) = 1 \land \\
Open(s', h_s) = 1 \land \\
sk' = ck \oplus s'

\end{array}\right]
\]

\[ 
P_2 = 
\Pr\left[
\begin{array}{l}

PoMM.Ver(\pi, tid) = 1
\end{array} 
\middle| 
\begin{array}{l}
(h_{sk}, h_s, ck, \sigma) \gets com_{mask} \\
AE.Ver(\{h_{sk}, h_s, ck\}  \sigma, iPk) \\ = 1 \land\\
Open(s', h_s)=1 \land\\
sk' = ck \oplus s' \\
(\pi, tid) \gets \\ PoMM.Gen(com_{mask}, iPk, s') 
\end{array}\right] 
\]

\begin{proof}
This property holds if the  commitment scheme and the signature scheme are secure, and the \textit{zk-SNARK} scheme used in proof generation holds the \textit{correctness}. 
\end{proof}

\begin{lemma}  \label{lemma:pomm_hiding}
    For any mask commitment $com_{mask}$ honestly constructed, PPT adversary can learns the encryption key $sk$ and the mask secret $s$ with a negligible probability, shown below.
\end{lemma}

\[ 
\Pr\left[
\begin{array}{l}

\mathcal{A}(com_{mask}) = (sk, s)
\end{array} 
\middle| 
\begin{array}{l}
com_{mask} \gets MCOM.Gen(sk, s, iK)\\

\end{array}\right] 
\]

\begin{proof}
It is evident that this property holds true if the commitment scheme and the one-time pad encryption are secure.
\end{proof}

\subsubsection{Commitment and PoM on Encryption} \label{app:pome}
We formally define encryption scheme and proof of misbehavior on encryption scheme.

\vspace{7pt}

\begin{breakablealgorithm}
    \small
    \caption{Encryption Commitment}
    \begin{algorithmic}[1]
        \Function{ECOM.Gen}{$m_i,  sk, id, iK$}
            \State $c \gets SE.Enc(m, sk.Sk, \text{TweakNonce}(sk.nonce, id))$
            \State $h_m \gets Commit(m_i)$ 
            \State $h_c \gets Commit(c)$
            \State $h_{sk} \gets Commit(Sk)$ 
            \State $\sigma\gets AE.Sign(\{h_m, h_c, h_{sk}, id \}, iK)$ 
            \State $com_{enc} \gets \{h_m, h_c, h_{sk}, id, \sigma\}$
            \State \Return $com_{enc}$
        \EndFunction
        \Statex
        \Function{ECOM.Ver}{$com_{enc}, iPk, h_{sk}', id'$}
        \State Parse $\{h_m, h_c, h_{sk}, id, \sigma\} \gets com_{enc}$
        \State \Return 1 \textbf{if} $AE.Ver(\{h_m, h_c, h_{sk}, id\}, \sigma, iPk) = 1$ \textbf{and} $id = id'$ \textbf{and} $h_{sk} = h_{sk}'$
        \State \Return 0
    \EndFunction
    \end{algorithmic}
\end{breakablealgorithm}

\vspace{7pt}

\begin{breakablealgorithm}
    \small
    \caption{Proof of Misbehavior on Encryption (POME)}
    \begin{algorithmic}[1]
        \State \textbf{require} Claim $C(x, w)$:
            \State \hspace{\algorithmicindent} $ (h_m, h_c, h_{sk}, id) \gets x$, $ (c', sk') \gets w$
            \State \hspace{\algorithmicindent} $Open(c', h_c) = 1$    \Comment{Ciphertext is correct}
            \State \hspace{\algorithmicindent} $Open(sk', h_{sk}) = 1$  \Comment{Key is correct}
            \State \hspace{\algorithmicindent} $m' \gets SE.Dec(c', sk'.Sk, \text{TweakNonce}(sk'.nonce, id))$
            \State \hspace{\algorithmicindent} $Open(m', h_{m}) =0$ \Comment{Decryption result is incorrect}
        \State \textbf{require} $(pp, vp) \gets Snk.Setup(C)$
        \Function{PoME.Gen}{$com_{enc}, c', iPk, sk'$}\Comment{$c'$ and  $sk'$ are secret witness}
            \State $tid \gets iPk$
            \State Parse $\{h_m, h_c, h_{sk}, id, \sigma\} \gets com_{enc}$
            \State $x \gets  (h_m, h_c, h_{sk}, id)$, $w \gets (c', sk')$
            \State $\pi \gets Snk.Prove(pp, x, w)$, $\pi_{pome} \gets (x, \sigma, \pi)$
            \State \Return $(\pi_{pome}, tid)$
        \EndFunction
        \Statex
        \Function{PoME.Ver}{$\pi_{pome}, tid$}
            \State Parse $(x, \sigma, \pi) \gets \pi_{pome}$
            \State Parse $(h_m, h_c, h_{sk}, id) \gets x$
             \If{$AE.Ver(\{h_{m}, h_c, h_{sk}, id\}, \sigma, tid) = 1$}
            \State \Return $Snk.Verify(vp, x, \pi)$
            \EndIf 
            \State \Return 0 
        \EndFunction

    \end{algorithmic}
\end{breakablealgorithm}

We then formally define the security properties secured by the $MCOM$ and $POMM$. 

\begin{lemma}  \label{lemma:pome_binding}
    For any encryption commitment $com_{enc}$ correctly signed by public key $iPk$, once the encryption key $sk'$ and ciphertext $c'$ are correctly revealed: If $com_{enc}$ is generated honestly, a plaintext $m$ committed by $com_{enc}.h_{m}$ can be extracted. Otherwise, a proof of misbehavior on encryption can be generated towards $iPk$. In other words, $P_3$ with \textit{probability} 1 when the committer is honest, $P_4$ with \textit{probability} 1 when the committer is dishonest. The sum of $P_3$ and $P_4$ has a \textit{probability} of 1.

\end{lemma}

\[ P_3 = 
\Pr\left[
\begin{array}{l}
Open(m', h_m) = 1
\end{array} 
\middle| 
\begin{array}{lcl}
(h_m, h_c, h_{sk}, i, \sigma) \gets com_{enc} \\
AE.Ver((h_m, h_c, h_{sk}, i), \sigma, iPk) \\ = 1 \land \\
Open(c', h_c )  = 1\land \\
Open(sk',h_{sk} )  = 1 \\
m' \gets  SE.Dec(c', sk'.Sk, \\ \text{TweakNonce}(sk'.nonce , i))
\end{array}\right]
\]

\[ P_4 = 
\Pr\left[
\begin{array}{l}
PoME.Ver(\pi, tid) = 1
\end{array} 
\middle| 
\begin{array}{l}
(h_m, h_c, h_{sk}, id, \sigma) \gets com_{enc} \\
AE.Ver(\{h_m, h_c, h_{sk}, id\}, \\ \sigma, iPk) = 1 \land\\
Open(c',h_c  )  = 1\land\\
Open(sk',h_{sk}  )  = 1\\
(\pi, tid) \gets PoME.Gen(\\com_{enc}, c', iPk, sk')
\end{array}\right] 
\]

\begin{proof}
This property holds if the encryption scheme, commitment scheme, the signature scheme are secure, and the \textit{zk-SNARK} scheme used in proof generation guarantees the \textit{correctness}.
\end{proof}

\begin{lemma}    \label{lemma:pome_hiding}
    Consider an adversary holding an encryption commitment $com_{enc}$ honestly constructed and the corresponding ciphertext $c$. The probabilistic polynomial-time adversary can learn the plaintext $m$ and the encryption key $sk$ with a negligible probability:
\end{lemma}

\[ 
\Pr\left[
\begin{array}{l}
\mathcal{A}(com_{enc}, c) = (m, sk)
\end{array} 
\middle| 
\begin{array}{l}
com_{enc} \gets \\ ECOM.Gen(m, sk,i, iK)\\
c \gets  SE.Enc(m, sk.Sk, \\ \text{TweakNonce}(sk.nonce , i) ) 
\end{array}\right] 
\]

\begin{proof}
This property holds if the commitment scheme and the symmetric encryption are secure.
\end{proof}

\subsubsection{Encryption Commitment Chain}

For a  commitment pair  $(com_{mask}, com_{enc})$ committed by the same committer with public key $iPk$, we consider a  commitment pair is \textit{paired} if: 
\begin{itemize}[leftmargin=*]
    \item parse $(h_m, h_c, h_{sk}, id, \sigma) \gets com_{enc}$
    \item parse $(h_{sk}', h_s, ck, \sigma') \gets com_{mask} $
    \item $h_{sk} = h_{sk}' \land AE.Ver(\{h_{sk}', h_s, ck\},\sigma', iPk) = 1 \land$
    \item $ AE.Ver(\{h_m, h_c, h_{sk}, id\},\sigma, iPk) = 1$
\end{itemize}

\begin{corollary}   \label{coro:pair}
    For a \textit{paired} commitment pair $(com_{mask}, com_{enc})$ committed by $iPk$, once the mask secret $s'$ and ciphertext $c'$ are correctly opened: If $(com_{mask}, com_{enc})$ are generated honestly, a plaintext $m$ committed by $com_{enc}.h_{m}$ can be extracted. Otherwise, a proof of misbehavior can be generated towards $iPk$.
\end{corollary}
    
\begin{proof}
    This corollary follows directly from Lemmas \ref{lemma:pomm_binding}, \ref{lemma:pome_binding}. 
\end{proof}

For a plaintext chunk $c_0$ with index $id$, there exist $n$ encoders $u_1, \ldots, u_n$: $u_i$ has encryption key $sk_i$, encrypting $c_0$ layer by layer and delivering the content to $u_{n+1}$. We consider $u_{n+1}$ knows all public keys of each user $u_i$ and receives a \textit{valid} tuple $(c_n, \mathbb{C}_1, \mathbb{C}_2,id)$ from $u_n$. We consider a tuple to be \textit{valid} if $(c_n, \mathbb{C}_1, \mathbb{C}_2, id)$ satisfies:
\label{sec:valid_tuple}

\begin{itemize}[leftmargin=*]
    \item Parse $\{\cdot com_{enc}^i \cdot\} \gets \mathbb{C}_1, \{\cdot com_{mask}^i \cdot\} \gets \mathbb{C}_2$
    \item $Open(c_n, com_{enc}^n.h_c) = 1 \land $ 
    \item For $i \in [1, n-1]$, $com_{enc}^i.h_c = com_{enc}^{i+1}.h_m \land $
    \item For $i \in [1, n]$: 
    \begin{itemize}
        \item $com_{enc}^i.id = id \land $
        \item  Parse $\{h_m, h_c, h_{sk}, i , \sigma \} \gets com_{enc}^i$
        \item Parse $(h_{sk}', h_s, ck, \sigma') \gets com_{mask}^i $
        \item $h_{sk} = h_{sk}' \land AE.Ver(\{h_{sk}', h_s, ck\},\sigma', Pk(u_i)) = 1 \land$
        \item $AE.Ver(\{h_m, h_c, h_{sk}, i\},\sigma, Pk(u_i) ) =1  $
    \end{itemize}
\end{itemize}

\begin{corollary} \label{coro:valid}
    For $u_{n+1}$ who holds a valid tuple $(c_n, \mathbb{C}_1, \mathbb{C}_2 , id)$, once all $u_i$ ($i \in [1, n])$ open their secrets committed by $com_{mask}^i.h_{s}$: If all commitments in $\mathbb{C}_1, \mathbb{C}_2$ are honestly constructed, $u_{n+1}$ can get the content committed $m$ by $com_{enc}^1.h_{m}$; Otherwise, a proof of misbehavior can be generated towards user $u_r$, where $u_r$ is the closest cheater to $u_{n+1}$. 
\end{corollary}

\begin{proof}
    This corollary follows directly from applying Corollary \ref{coro:pair} and Lemma.\ref{lemma:pome_binding} recursively. 
\end{proof}

\subsection{Functions used in FairRelay}    \label{app:functions}

In this section we details functions used in our protocol. 

\vspace{3 pt}
\begin{breakablealgorithm}
    \small
    \begin{algorithmic}
        \caption{Nonce Tweak}
        \Function{TweakNonce}{$nonce, id$}  \Comment{$nonce \in \{0, 1\} ^\mu$}
        \State $h \gets \FR(\textit{query}, sid, nonce || id)$
        \State \Return $h[\mu:]$    \Comment{Cut first $\mu$ bits}
        \EndFunction
    \end{algorithmic}
\end{breakablealgorithm}

\vspace{2 pt}

\begin{breakablealgorithm}
    \caption{Lock and Unlock}\label{algo:lock_unlock}
         
    \begin{algorithmic} 
        \Function{lock}{$v, sid, cid, amt, \phi$} 
        \State $(sid, cid, lb, rb) \leftarrow \FC(\textit{query}, sid, cid)$
        \If{v is the left party of channel $cid$}
        \State $x \gets ( cid, lb - amt, rb + amt, \phi)$
        \Else
        \State $x \gets ( cid, lb +  amt, rb - amt, \phi)$
        \EndIf 
        \State $\sigma \gets AE.Sign(x, K(v))$
        \State $tx \gets \{x, \sigma\}$
        \State \Return $tx$
        \EndFunction
        \Statex
        \Function{Unlock}{$tx, w, s$}
        \State Parse $(x, \sigma ) \gets tx$
        \State Parse $(sid, cid, lb, rb, \phi) \gets x$
        \State $\sigma' \gets AE.Sign(x, K(w))$
        \State  $(\textit{update}, sid, cid, lb, rb, \phi, s) \gets \textbf{Compose}(x, \sigma, \sigma', s)$   \Comment{Generate acceptable calls to $\FC$}
        \State \Return $(\textit{update}, sid, cid, lb, rb, \phi, s)$
        \EndFunction
    \end{algorithmic}
\end{breakablealgorithm}

\vspace{2pt}
\begin{breakablealgorithm}
    \small
    \caption{Payment Condition Construction} \label{al:construct} 
    \begin{algorithmic} 
        \Function{Construct}{$\mathbb{H}, t$}
        \State \textbf{return Function} $\phi(x, ct)$:
        \State \quad \quad \quad store $\phi.\mathbb{H} \gets\mathbb{H} , \phi.t \gets t$
        \State \quad \quad \quad parse $s_1, \ldots, s_n \gets x$
        \State \quad \quad \quad  parse $h_1, \ldots, h_n \gets \mathbb{H}$
        \State \quad \quad \quad  \textbf{for} $i \gets 1$  \textbf{to} $n$
        \State \quad \quad \quad \quad \textbf{if} $Open(s_i, h_i) = 0$ \textbf{, return } $0$
        \State \quad \quad \quad \textbf{EndFor}
        \State \quad \quad \quad  \textbf{if} $ct > t$ \textbf{, return } $0$
        \State \quad \quad \quad  \textbf{return } $1$
        \EndFunction
    \end{algorithmic}
\end{breakablealgorithm}

\vspace{10pt}
\begin{breakablealgorithm}

    \caption{Extract Encryption Keys}\label{alg:extkey}
    \begin{algorithmic}
\Function{ExtKey}{$\mathbb{S}, \mathbb{COM}_{mask}, \mathcal{G}$} \Comment{Pass identities of parties in$ \mathcal{G}$}
    \State Parse ${s}_0, \{ s_{1,1}, \ldots, s_{1,|p_1|} \}, \ldots, \{ s_{\eta,1}, \ldots, s_{\eta,|p_{\eta}|} \} \gets \mathbb{S}$
    \State Parse $com_{m}^0, \{ com_{m}^{1,1}, \ldots, com_{m}^{1,|p_1|} \}, \ldots, \{ com_{m}^{\eta,1}, \ldots, com_{m}^{k,|p_{\eta}|} \} \gets \mathbb{COM}_{mask}$

    \State $\mathbb{K} \gets \emptyset$
    \If{$Open(com_{m}^0.ck \oplus s_0,  com_{m}^0.h_{sk} ) = 0$}
        \State $(\pi, tid) \gets POMM.Gen(com_{m}^0, Pk(\CP), s_0)$
        \State \Return $(\emptyset, tid, \pi)$
    \Else 
        \State $\mathbb{K} \gets \mathbb{K} \cup (com_{m}^0.ck \oplus s_0)$
    \EndIf

    \For{$k \gets 1$ to $\eta$}
        \For{$i \gets 1$ to $|{p_k}|$}
            \If{$Open(com_{m}^{k,i}.ck \oplus s_{k,i}, com_{m}^{k,i}.h_{sk} ) = 0$}
                \State $(\pi, tid) \gets POMM.Gen(com_{m}^{k,i}, Pk(\R_{k, i}), s_{k,i})$
                \State \Return $(\emptyset, tid, \pi)$
            \Else
                \State $\mathbb{K} \gets \mathbb{K} \cup (com_{m}^{k,i}.ck \oplus s_{k,i})$
            \EndIf
        \EndFor
    \EndFor
    
\State \Return $(\mathbb{K}, \bot, \bot)$
\EndFunction
    \end{algorithmic}
\end{breakablealgorithm}

\vspace{10 pt}

\begin{breakablealgorithm}
    \small
    \caption{Extract Content}\label{alg:extcontent}
    \begin{algorithmic}
        
        \Function{ExtChunk}{ $c', \mathbb{C}, j, \mathbb{Sk}$}  \Comment{Decrypt a ciphertext chunk layer by layer}
            \State $ n \gets |\mathbb{Sk}|$
            \State parse $sk_0, \ldots, sk_n \gets  \mathbb{Sk}$
            \State parse $com_{enc}^0, \ldots, com_{enc}^n \gets \mathbb{C}$
            \State $c \gets c'$
            \For{$i \gets n$ to $0$}
                \State $com_{enc} \gets com_{enc}^i$
                \State $c_p \gets SE.Dec(c, sk_i.Sk, TweakNonce(sk_i.nonce, j) )$
                \If{$Open(c_p, com_{enc}.h_m) = 0$} \Return $(\bot, i, c_p, com_{enc})$
                \EndIf 
                \State $c \gets c_p$
            \EndFor 
            \State \Return $(c, \bot, \bot, \bot)$
        \EndFunction
        \Statex 
        \Function{ExtPath}{$\mathbb{CT}, \mathbb{COM}_{enc}, \mathbb{Sk}$}
            \State $ n \gets |\mathbb{CT}|$     \Comment{Get chunk numbers}
            \State parse $ c_1, \ldots, c_n \gets \mathbb{CT}$
            \State parse $ \mathbb{C}_1, \ldots, \mathbb{C}_n \gets \mathbb{COM}_{enc}$
            \State $PT \gets \emptyset$     \Comment{Plaintext chunk list}
            \For{$i \gets 1$ to $n$}        \Comment{Range each ciphertext chunk}
                \State $idx \gets \mathbb{C}_i[0].id$   \Comment{Get the chunk index in content}
                \State $(m, cht, c, com_{enc}) \gets ExtChunk(c_i, idx, \mathbb{Sk})$
                \If{$cht \neq \bot$} \Return $ (\bot, cht, c, com_{enc})$
                \EndIf 
                \State $PT \gets PT \cup (m, idx)$
            \EndFor 
            \State \Return $ (PT, \bot, \bot, \bot)$
        \EndFunction
        \Statex
        
        \Function{Extract}{$\mathbb{K}, \mathbb{CT},\mathbb{COM}_{enc}, \mathcal{G}$}
            \State parse $\mathbb{CT}_{1}, \ldots, \mathbb{CT}_{n_f} \gets \mathbb{CT}$
            \State $sk_0,\{ sk_{1,1}, \ldots, sk_{1, |p_1}|\}, \ldots, \{sk_{\eta,1}, \ldots, sk_{\eta, |p_\eta|}\} \gets \mathbb{K}$
            \State $\mathbb{COM}_{enc}^1, \ldots, \mathbb{COM}_{enc}^{\eta} \gets \mathbb{COM}_{enc}$
             \State $Chunks \gets \emptyset$
            \For{$i \gets 1$ to $\eta$}
                \State $\mathbb{Sk} \gets \{sk_0, sk_{i,1}, \ldots,  sk_{i,|p_i|} \}$
                \State $(PT, cht, c, com_{enc}) \gets ExtPath(\mathbb{CT}_i, \mathbb{COM}_{enc}^i, \mathbb{Sk})$
                \If{$cht \neq \bot \land cht \neq 0$}    \Comment{ a relay cheats on encryption }
                    \State \Return $PoME.Gen(com_{enc}, c, Pk(\R_{i, cht}), sk_{i, cht})$
                \EndIf 
                \If{$cht \neq \bot \land cht = 0$}    \Comment{ provider cheats on encryption}
                    \State \Return $PoME.Gen(com_{enc}, c, Pk(\CP), sk_0)$
                \EndIf 
                   
                \State $Chunks \gets PT \cup Chunks$
            \EndFor
            \State \Return $Rebuild(Chunks)$    \Comment{Rebuild full content from chunks}
        \EndFunction
    \end{algorithmic}
\end{breakablealgorithm}

\begin{figure*}[htbp] 
\begin{mdframed}[frametitle={Ideal Functionality of $\Fex$}, frametitlealignment=\centering, frametitlerule=true, frametitlebackgroundcolor=gray!20, font=\small, innertopmargin=2pt, innerbottommargin=2pt, innerrightmargin=5pt, innerleftmargin=5pt]

Consider a \textit{Enforceable A-HTLC} consisting of a payer $u_0$ (provider $\CP$), and $n$ payees $\{ u_1, \ldots, u_n \}$ ($u_i$ is the i-th relay $\R_i$). $u_0$ pays $\mathfrak{B}_i$ to $u_i$ in exchange of a secret $s_i$ committed by $h_i$ ($Open(s_i, h_i) = 1)$. 
$u_0$ and $u_i$ had   an agreed-on enforcement deadline $T$. 
We denote $\{h_i, \ldots, h_n\}$ as $\mathbb{H}_i$. $cid_i$ denotes the payment channel identifier between $u_{i-1}$ and $u_i$, and 
$ ADDR: \{Pk(u_0),Pk(u_1), \ldots, Pk(u_n) \}$. We consider  $n$, $ADDR$ as public knowledge configured on a trusted setup (e.g, a PKI). $\Fex$ owns the private key of each party.

\noindent \textbf{Local variables}
\begin{itemize}[leftmargin=*]
    \item$states$: is a $n$-tuple, where $states[i]$ is the state of the $i$-th channel.  Each state is initialized as $\bot$.
        \item$\mathfrak{B}_{max}$: Refund amount configured globally. 
        \item$cheater$: Identity of the cheater, initialized as $\bot$.
        \item$\textit{ifEnf}$: Enforcement phase flag; $\textit{eid}$: Enforcement identifier, initialized as 0. 
        \item $ct$: Current round retrieved from global clock $\Fclk$
\end{itemize}

    \begin{center}
        \textbf{Lock Phase }
    \end{center}

    (Round 1) Upon receiving $(\textit{lock},  \textit{sid},  \mathbb{H}_1, T, h_0, amt, tx_1) \leftarrow u_0$: 
    \begin{itemize}[leftmargin=*]
        \item  If $u_0$ is honest: 
        \begin{itemize}
            \item $\phi:= Construct(  \{ h_0 \} \cup \mathbb{H}_1,  T + n + 1)$, $tx_1 := lock(u_0, sid, cid_1, amt, \phi) $
        \end{itemize}
        \item Save $ T,  h_0, \mathbb{H}_1$,  
    leak $(\textit{lock},  \textit{sid}, u_0, tx_1)$ to $\Sim$, set $states[1]:= \textit{locked}$, set $\textit{ifEnf}:= 0$.
    \end{itemize}
 
    (Round 1 + i) ($0< i $) Upon receiving $(\textit{lock}, \textit{sid}, T^*, h_0^*, \mathbb{H}_{1}^*, h_i, \mathfrak{B}_i, tx_{i+1}) \leftarrow u_i$: 

    \begin{itemize}[leftmargin=*]
        \item If $u_i$ is honest: 
        \begin{itemize}[leftmargin=*]
            \item Check if $tx_i$ is properly signed by $u_i$ and  $tx_i.\phi.t = T + n - i +2$.
            \item  $(sid, cid_i, lb, rb) \leftarrow \FC(\textit{query}, sid, cid_i)$, check  $tx_i.lb + tx_i.rb = lb + rb \land  tx_i.rb > rb +\mathfrak{B}_i$. Set $amt_i = tx_i.rb - rb$. 
            \item Parse $h_1^*, \ldots, h_n^* := \mathbb{H}_{1}^*$,  $h_0', h_i', \ldots, h_n' := tx_{i}.\phi.\mathbb{H}$. For $j \in [i, n]$ check if $h_j' = h_j^*$. Check if $h_0' = h_0^*$. If any check fails, $\Fex$ aborts.
            \item If $i < n $, set $\phi:= Construct(tx_{i}.\phi.\mathbb{H} \backslash h_i , tx_{i}.\phi.t-1 )$, $tx_{i+1} := lock(u_i, sid, cid_i, amt_i - \mathfrak{B}_i, \phi)$. 
        \end{itemize}
        
    \end{itemize}

    \begin{itemize}[leftmargin=*]
        \item If $T^* = T\land h_0^* = h_0 \land \mathbb{H}_{1}^* = \mathbb{H}_{1}$, set $\textit{ifEnf}:= \textit{ifEnf} + 1$, and set $Ch:= (h_0, \mathbb{H}_{1}^*, T^*, ADDR)$, $\sigma_i:= AE.Sign(Ch, K(u_i))$. Leak $(\textit{ack}, u_i, \sigma_i)$ to $\Sim$. 
        \item Save $tx_{i+1}$, if $i < n$,  set $states[i+1]:= \textit{locked}$,  leak $(\textit{locked}, \textit{sid}, u_i, tx_{i+1})$ to $\Sim$.  
    \end{itemize}
    
    \begin{center}
        \textbf{Unlock Phase}
    \end{center}

    (Round $T'$) Upon receiving $( \textit{release}, \textit{sid}, s_0) \leftarrow u_0$: If $Open(s_0, h_0) = 1 \land T' < T \land \textit{ifEnf} = n$ , save $s_0$, leak $(\textit{release}, \textit{sid}, u_0, s_0)$ to $\Sim$, and  set $states[n]:= \textit{unlockable}$. 
    
    \vspace{2pt }
    
    Upon receiving $(\textit{unlock}, \textit{sid}, s_n) \leftarrow u_n$: if $states[n] = \textit{unlockable}$ and $tx_{n}.\phi(\{s_0, s_n\}, ct) = 1$, save $s_n$, send $unlock(tx_n, u_n,\{s_0, s_n\})  \rightarrow \FC$, then set $states[n]:= \textit{unlocked}$.

    Upon receiving $(\textit{unlock}, \textit{sid}, s_i) \leftarrow u_i$ ($i< n$): if $states[i] = \textit{locked}$, and  $states[i+1] = \textit{unlocked}$, set $\mathbb{S} = \{s_0, s_i, \ldots, s_n\}$. If $tx_{i}.\phi(\mathbb{S} , ct) = 1$, send  $unlock(tx_{i}, u_n, \{s_i\} \cup \mathbb{S})  \rightarrow \FC$, then set $states[i]:= \textit{unlocked}$.

    \begin{center}
        \textbf{Enforcement Phase}
    \end{center}
   
    (Round $T''$) Upon receiving $(\textit{challenge}, \textit{sid}, s_0) \leftarrow u_0$: If $T'' < T \land Open(s_0, h_0) = 1 \land \textit{ifEnf} = n$, leak $(\textit{challenge}, \textit{sid}, u_0, s_0)$ to $\Sim$. 
    If $states[n] = \textit{locked} $, set $states[n]:= \textit{unlockable}$. Parse $h_1', \ldots, h_n' := \mathbb{H}_1$, and set $eid:= n$, $\textit{ifEnf} :=0$. 

    (Round $T'' + 1$) Upon receiving $(\textit{response}, \textit{sid}, s_n,  \textit{ifUnlock}) \leftarrow u_n$: 
    \begin{itemize}[leftmargin=*]
        \item If $Open(s_n, h_n') = 1 \land \textit{ifUnlock} =1 \land states[n] = \textit{unlockable}$ and $tx_{n}.\phi(\{s_0, s_n\}, T'' + 1) = 1$, save $s_n$, send $unlock(tx_n, u_n,\{s_0, s_n\})  \rightarrow \FC$, then set $states[n]:= \textit{unlocked}$. 
        \item If $Open(s_n, h_n') = 1 \land eid:= n$, save $s_n$, leak $(\textit{response}, \textit{sid}, s_n, u_n)$ to \Sim, $eid:= n-1$
        \item If no valid \textit{response} message received in this round, set $cheater:= u_n$.
    \end{itemize}

    (Round $T'' + n - i + 1$) Upon receiving $(\textit{reponse}, \textit{sid}, s_i, \textit{ifUnlock} ) \leftarrow u_i$ ($0<i< n$):
    \begin{itemize}[leftmargin=*]
        \item If  $Open(s_i, h_i') = 1 \land eid:= i \land states[n] = \textit{locked } \land \textit{ifUnlock} =1 \land cheater = \bot \land tx_{i}.\phi(\{s_0, s_i, \ldots, s_n\}, T'' +  n - i + 1) = 1$: send $unlock(tx_n, u_n,\{s_0, s_n\})  \rightarrow \FC$, and set $states[i]:= \textit{unlocked}$. 
        \item If $Open(s_i, h_i') = 1 \land eid:= i$, save $s_i$, leak $(\textit{reponse}, \textit{sid}, s_i, \R_i)$ to \Sim, and set $eid:= eid-1$. 
        \item If no valid \textit{response} message received in this round and $cheater = \bot$, set $cheater:= u_i$, $eid := 0$
    \end{itemize}
    (Round $T'' + n + 1$) Upon receiving $(\textit{punish}, \textit{sid} ) \leftarrow u_0$:  If $cheater \neq \bot$, send $(\textit{transfer}, sid, cheater, u_0, \mathfrak{B}_{max}) \rightarrow \FL$, and leak $(\textit{punished}, \textit{sid} ) $.
\end{mdframed}

 \caption{Ideal Functionality of $\Fex$}
  \label{fig:uc_if} 
\end{figure*}

\begin{figure*}[htb]
\setlength{\abovecaptionskip}{0.1cm}
\begin{mdframed}[frametitle={Enforceable \textit{A-HTLC} Protocol $\Pi$}, frametitlealignment=\centering, frametitlerule=true, frametitlebackgroundcolor=gray!20, font=\small, innertopmargin=2pt, innerbottommargin=2pt, innerrightmargin=5pt, innerleftmargin=5pt]

Consider an \textit{Enforceable A-HTLC} protocol involving a payer $u_0$ (referred to as the provider $\CP$), $n$ payees denoted as $\{ u_1, \ldots, u_n \}$ (where $u_i$ represents the i-th relay $\R_i$). In this protocol, $u_0$ makes a payment of $\mathfrak{B}_i$ to $u_i$ in exchange for a secret $s_i$ that is committed by $h_i$ (where $Open(s_i, h_i) = 1$).
Both $u_0$ and $u_i$ agree upon an enforcement deadline denoted as $T$. We use $\{h_i, \ldots, h_n\}$ to represent the set of commitments, which is denoted as $\mathbb{H}_i$. The payment channel identifier between $u_{i-1}$ and $u_i$ is denoted as $cid_i$, and the addresses are represented by $ADDR: \{Pk(u_0),Pk(u_1), \ldots, Pk(u_n) \}$.
In this protocol, we assume that the values of $n$ and $ADDR$ are public knowledge and are configured as part of a trusted setup, such as a Public Key Infrastructure (PKI).

    \begin{center}
        \textbf{Lock Phase }
    \end{center}

    \begin{itemize}[leftmargin=*]
        \item (Round 1) $u_0$ : 
        \begin{itemize}
                  \item Upon receiving $(\textit{lock},  sid, \mathbb{H}_1, T, h_0, amt) \leftarrow \env$, $u_0$ set $\phi:= Construct(  \{ h_0 \} \cup \mathbb{H}_1,  T + n + 1)$, $tx_1:= lock(u_0, sid, cid_1, amt, \phi)$, and send $ (\textit{channel-lock}, sid, tx_1 ) \rightarrow u_1$.
                  Then $u_0$ will enter the unlock/enforcement phase. If no valid \textit{release} or \textit{challenge} message arrives before round $T + n + 1$, $u_0$ terminates. 
                  \item If no valid \textit{lock} message received in this round, $u_0$ terminates.
        \end{itemize}
        
        \item  $u_i$ ($i> 0$) in round $(i + 1)$: 
                \begin{itemize}
                    \item Upon receiving $(\textit{lock}, \textit{sid}, T, h_i, \mathfrak{B}_i, h_0, \mathbb{H}_1) \leftarrow \env$
                    \begin{itemize}
                          \item Check if  $u_i$ had received lock message $ (\textit{channel-lock}, sid, tx_{i}) \leftarrow u_{i-1}$, where $tx_i$ should be signed by $u_{i-1}$ and $tx_{i}.\phi.t = T + n - i + 2$. 
                          \item Parse $ h_1, \ldots, h_n := \mathbb{H}_1$ . 
                          \item $(sid, cid_i, lb, rb) \leftarrow \FC(\textit{query}, sid, cid_i)$, check if $tx_i.lb + tx_i.rb = lb + rb \land  tx_i.rb > rb +\mathfrak{B}_i$. set $amt = tx_i.rb - rb$. 
                          \item Parse $h_0', h_i', \ldots, h_n' := tx_{i}.\phi.\mathbb{H}$. For $j \in [i, n]$, check if $h_j' = h_j$. Check if $h_0' = h_0$. 
                          \item If any check fail, $u_i$ aborts the protocol.
                          \item Set $st_i := \textit{locked}$, $Ch:= (h_0, \mathbb{H}_1, T,  ADDR )$,  $\sigma_{Ch}^i:= AE.Sign(Ch, K(u_i))$, $\sigma_{Ch}^i \rightarrow u_0$. 
                          \item If $i < n $, $\phi:= Construct(tx_{i}.\phi.\mathbb{H} \backslash h_i , tx_{i}.\phi.t-1 )$,  $tx_{i+1}:= lock(u_i, sid, cid_{i+1},amt - \mathfrak{B}_i, \phi ) $, $(\textit{channel-lock}, sid, tx_{i+1} ) \rightarrow u_{i+1}$. 
                          \item Then $u_i$ will enter the next phase. If no new \textit{unlock} or \textit{release} message received until round $tx_{i}.\phi.t$, $u_i$ aborts the protocol.
                    \end{itemize}
                    \item If no valid \textit{lock} message received in round $1 + i$, $u_i$ terminates.
                \end{itemize}

    \end{itemize}

    \begin{center}
        \textbf{Unlock Phase}
    \end{center}

    \begin{itemize}[leftmargin=*]
        \item  $u_0$: Upon receiving $( \textit{release}, \textit{sid}, s_0) \leftarrow \env  $ in round $T'$.
        \begin{itemize}
            \item If $T' < T \land Open(s_0, h_0) = 1$ and $u_0$ had received all $\sigma_{Ch}^i$ from $u_i$, where $Ch:= (h_0, \mathbb{H}_1, T, ADDR )$,  $AE.Ver(Ch, \sigma_{Ch}^i, Pk(u_i) ) = 1$, send $(\textit{released}, \textit{sid},  s_{0}) \rightarrow u_n$. 
        \end{itemize}

        \item $u_n$: Upon receiving $(\textit{unlock}, sid, s_n) \leftarrow \env$ in round $ct$: 
                  \begin{itemize}
                      \item  If $u_i$ had received  $(\textit{released} , \textit{sid},  s_{0})$ from  $u_0$, $u_i$ sets $\mathbb{S}_n := \{ s_{0}, s_n \}$. 
                      \item If $ tx.\phi(\{s_{0}, s_{n}\}, ct) = 1 \land st_n = \textit{locked}$, $u_i$  sends $unlock(cid_n, tx_n,  u_{n} , \mathbb{S}_n)\rightarrow \FC$, and sets $st_n:= \textit{unlocked}$. 
                  \end{itemize}

        \item $u_i$ ($0<i< n$):  Upon receiving $(\textit{unlock}, sid, s_i) \leftarrow \env$:
        \begin{itemize}
            \item If $u_i$ had received $(\textit{updated}, sid, cid_{i+1},  \mathbb{S}_{i+1} ) \leftarrow \FC$, $u_i$ sets $\mathbb{S}_{i}:= \mathbb{S}_{i+1} \cup \{ s_i \}$. 
            \item If $st_i = \textit{locked} \land tx.\phi(\mathbb{S}_{i}, ct) = 1$, $u_i$ sends $unlock( tx_i, u_i, \mathbb{S}_{i}) \rightarrow \FC$, then set $st_i:= \textit{unlocked}$. 
        \end{itemize}

    \end{itemize}

        \begin{center}
        \textbf{Enforcement phase}
    \end{center}

    \begin{itemize}[leftmargin=*]
        \item  $u_0$:
              \begin{itemize}[leftmargin=*]
                  \item  Upon receiving $(\textit{challenge}, sid, s_0) \leftarrow \env$ in round $(T'')$, if $T''  < T  \land Open(s_0, h_0) =1$ and $u_0$ had received all $\sigma_{Ch}^i$ from $u_i$, where $Ch:= (h_0, \mathbb{H}_1, T, ADDR )$,  $AE.Ver(Ch, \sigma_{Ch}^i, Pk(u_i) ) = 1$: set $\Sigma := \{ \sigma_{Ch}^1, \ldots, \sigma_{Ch}^n \}$, $(\textit{enforce},sid,Ch,\Sigma, s_{0}) \rightarrow \Fjc$.
                  \item In round $(T'' + n + 1)$: 
                  \begin{itemize}
                      \item Upon receiving $(\textit{punish}, sid ) \leftarrow \env$, send $(\textit{punish}, sid) \rightarrow \Fjc$, then terminate.
                        \item If no \textit{punish} message received in this round, $u_n$ terminates. 
                  \end{itemize}
              \end{itemize}
        \item In round $(T'' +  1)$, $u_n$:  
        \begin{itemize}
            \item  Upon receiving $(\textit{response}, sid, s_n) \leftarrow \env$, if $u_n$ had received $(\textit{enforced}, sid, s') \leftarrow \Fjc$ :
            \begin{itemize}[leftmargin=*]
                  \item If $st_n = \textit{locked} \land tx_n.\phi(\{s', s_{n}\}, ct) = 1$, send $unlock(cid_n, tx_n, \R_{i}, \{s', s_n\}) \rightarrow \FC$, set $st_n:= \textit{unlocked}$
                  \item send $(\textit{log}, sid, i, s_n) \rightarrow \Fjc$, then $u_n$ terminates.
              \end{itemize}
            \item If no \textit{response} message received in this round, $u_n$ terminates.
        \end{itemize}
    
        \item In round $(T'' + n -i+ 1)$, $u_i$ ($0<i < n$): 
        \begin{itemize}
            \item Upon receiving $(\textit{response}, sid, s_i) \leftarrow \env$, if $u_i$ had received $(\textit{logged}, sid, i, \mathbb{S})$ from $\Fjc$: 
              \begin{itemize}[leftmargin=*]
                  \item If $st_i = \textit{locked} \land  tx_i.\phi(\{\mathbb{S}, s_{i}\}, ct) = 1$, send $unlock(cid_i, tx_i, u_{i}, \mathbb{S} \cup \{s_i\}) \rightarrow \FC$, set $st_i:= \textit{unlocked}$
                  \item Send $(\textit{log}, sid, i, s_i) \rightarrow \Fjc$, then $u_i$ terminates.
              \end{itemize}
            \item If no \textit{response} message received in this round, $u_i$ terminates.
        \end{itemize}
        
    \end{itemize}
\end{mdframed}
 \caption{Enforceable \textit{A-HTLC} Protocol Formal Description $\Pi$}
\label{fig:uc_protocol}
\vspace{-0.3cm}
\end{figure*}

\begin{figure*}

\begin{mdframed}[frametitle={Simulator}, frametitlealignment=\centering, frametitlerule=true, frametitlebackgroundcolor=gray!20, font=\small, innertopmargin=2pt, innerbottommargin=2pt, innerrightmargin=5pt, innerleftmargin=5pt]

    \begin{center}
        \textbf{Lock Phase}
    \end{center}

    (Round 1) Upon receiving $(\textit{lock}, sid, \mathbb{H}_1, T, h_0, amt) \leftarrow \env$ in the ideal world:
    \begin{itemize}[leftmargin=*]
        \item If $u_0$ is honest:
        \begin{itemize}
            \item If $\Sim$ learns $(\textit{lock}, sid, u_0, tx_1)$ from $\Fex$ leakage, $\Sim$ simulates the lock message by sending $(sid, tx_1)$ to $u_1$ on behalf of $u_0$.
        \end{itemize}
        \item If $u_0$ is dishonest:
        \begin{itemize}[leftmargin=*]
            \item $\Sim$ learns the locking message $(sid, tx_1^*) \rightarrow u_1$ in the real world. Then $\Sim$ sends $(\textit{lock}, sid, \mathbb{H}_1, T, h_0, tx_1^*) \rightarrow \Fex$. If no locking message is captured, $\Sim$ aborts the simulation.
            \item $\Sim$ then simulates messages by replaying the captured messages in the ideal world.
        \end{itemize}
    \end{itemize}

(Round $i + 1$) Upon receiving $(\textit{lock}, \textit{sid}, T, h_i, \mathfrak{B}_i, h_0, \mathbb{H}_1) \leftarrow \env$:
\begin{itemize}[leftmargin=*]
    \item If $u_i$ is honest:
    \begin{itemize}[leftmargin=*]
        \item If $\Sim$ learns $(\textit{ack}, u_i, \sigma_i)$ from $\Fex$ leakage, $\Sim$ sends $(sid, \sigma_i) \rightarrow u_0$ on behalf of $u_i$. If $\Sim$ learns $(\textit{locked}, sid, u_i, tx_{i+1})$ from $\Fex$, $\Sim$ sends $(sid, tx_{i+1}) \rightarrow u_{i+1}$ on behalf of $u_i$.
    \end{itemize}
    \item If $u_i$ is dishonest:
    \begin{itemize}[leftmargin=*]
        \item $\Sim$ learns if $u_i$ sent $(sid, tx_{i+1})$ in the real world. If no $tx_{i+1}$ is sent in this round, $\Sim$ sets $tx_{i+1}' = \bot$. 
        $\Sim$ then learns if $u_i$ sends $\sigma_{Ch}^i$ towards $u_0$. As $u_i$ is corrupted, $\Sim$ extracts $Ch^i$ from $\sigma_{Ch}^i$. After the extraction, $\Sim$ sends $(\textit{lock}, sid, CH^i.T, Ch^i.h_0, Ch^i.\mathbb{H}_1, \bot,  \bot, tx_{i+1}') \rightarrow \Fex$. If $u_i$ does not send $\sigma_{Ch}^i$, $\Sim$ sends $(\textit{lock}, sid, \bot, \bot, \bot, \bot, \bot, tx_{i+1}') \rightarrow \Fex$
        \item $\Sim$ then simulates the communication by replaying the captured messages in the ideal world.
    \end{itemize}
\end{itemize}

    \begin{center}
        \textbf{Unlock Phase}
    \end{center}
    (Round $T'$) Upon $\tilde{u_0}$ receiving $( \text{release}, \textit{sid}, s_0) \leftarrow \env  $: 
    \begin{itemize}[leftmargin=*]
        \item  $u_0$ is honest:
        \begin{itemize}
            \item As $\Fex$ leaks $(\textit{release}, sid, u_0, s_0)$ to $\Sim$, $\Sim$ simulates communication by sends $(\textit{released}, sid, s_0)$ to $u_n$ on behalf of $u_0$. 
        \end{itemize}
        \item $u_0$ is dishonest:
        \begin{itemize}[leftmargin=*]
            \item If $\Sim$ learns $(\textit{release}, sid, s_0^*)$ sent from $u_0$, $\Sim$ will send $( \text{release}, \textit{sid}, s_0^*) \rightarrow \Fex  $. If no message captured, $\Sim$ send nothing to $\Fex$. 
            \item $\Sim$ then simulates the communication by replaying captured message in ideal world.     
        \end{itemize}
    \end{itemize}
     Upon $\tilde{u_i}$ ($i \in [1, n]$) receiving  $( \text{release}, \textit{sid}, s_i) \leftarrow \env  $: 
     \begin{itemize}[leftmargin=*]
         \item If $u_i$ is honest: $\Sim$ learns $(\textit{release}, u_i, \textit{sid}, s_i)$ from $\Fex$'s leakage, simulating the release message by sending $(\textit{release},  \textit{sid}, s_0)$ to $u_n$. 
         \item If $u_i$ is dishonest: 
         \begin{itemize}[leftmargin=*]
             \item If $u_i$ sends $(\textit{release},  \textit{sid}, s_0^*)$ to $u_n$, $\Sim$ sends  $ (\text{release}, \textit{sid}, s_0^*) \rightarrow \Fex$. Otherwise, $\Sim$ sends nothing to $\Fex$.
         \end{itemize}
          
     \end{itemize}

    \begin{center}
        \textbf{Enforcement Phase}
    \end{center}
    (Round $T''$) Upon $\tilde{u_0}$ receiving $(\textit{challenge}, sid, s_0) \leftarrow \env$: 
    \begin{itemize}[leftmargin=*]
        \item If $u_0$ is honest: 
        \begin{itemize}
            \item If $\Sim$ learns the $s_0$ and $\sigma_i$ from $\Fex$'s leakage, set $\Sigma:= \{\sigma_1, \ldots, \sigma_n\}$, $Ch:= (h_0, \mathbb{H}_1, T, ADDR)$,  $\Sim$ will simulate $\Fjc$'s execution by running $(\textit{enforced}, sid, s_0) \leftarrow \Fjc(\textit{enforce}, sid, Ch, \Sigma, s_0)$. 
        \end{itemize}
        \item If $u_0$ is dishonest: 
        \begin{itemize}[leftmargin=*]
            \item If $u_0$ sends $(\textit{enforce}, sid, Ch^*, \Sigma^*, s_0')$ to $\Fjc$ in real world, $\Sim$ sends $(\textit{challenge}, sid, s_0')   \rightarrow  \Fex$. $\Sim$ then simulates $\Fjc$'s execution in the ideal world ($ (\textit{enforce}, sid, Ch^*, \Sigma^*, s_0') \rightarrow \Fjc $). If $\Fjc$ leaks the challenge message, $\Sim$ simulates $\Fjc$'s response:  $(\textit{enforced}, sid, s_0') \leftarrow \Fjc$. 
            
        \end{itemize}
    \end{itemize}
    (Round $T'' + n - i + 1$) Upon $\tilde{u_i}$ receiving $(\textit{responce}, sid, s_i) \leftarrow \env$ ($i \in [1, n]$): 
    \begin{itemize}[leftmargin=*]
        \item If $u_i$ is honest: 
        \begin{itemize}[leftmargin=*]
            \item If $\Sim$ learns $(\textit{response}, sid, s_i, u_i)$ from $\Fex$,  $\Sim$ will simulate the $\Fjc$ by running $(\textit{logged}, sid, \mathbb{S}) \leftarrow \Fjc(\textit{log}, sid,i, s_i)$. 
        \end{itemize}
        \item If $u_i$ is dishonest: 
        \begin{itemize}[leftmargin=*]
            \item If $\Sim$ captures the $(\textit{log}, sid, i, s_i') \rightarrow \Fjc$ message sent by $u_i$ in real world, and the channel $cid_{i}$ updated from $tx_i$, $\Sim$ sends $(\textit{responce}, sid, s_i', 1)   \rightarrow  \Fex$. 
            If $\FC$ does not captured \textit{updated} message from $\Sim$, $\Sim$ sends $(\textit{responce}, sid, s_i', 0)   \rightarrow  \Fex$
            \item $\Sim$ replays the contract execution in the ideal world.  
        \end{itemize}
    \end{itemize}
    
   (Round $T'' + n + 1$) Upon $\tilde{u_0}$ receiving $(\textit{punish}, sid) \leftarrow \env$ 
   \begin{itemize}[leftmargin=*]
        \item If $u_0$ is honest: 
        \begin{itemize}[leftmargin=*]
            \item If $\Sim$ learns \textit{punished} message from $\Fex$,  $\Sim$ will simulates the contract execution $(\textit{punished}, sid) \leftarrow \Fjc(\textit{punish}, sid)$. If $\Sim$ learns \textit{punished-fail} message from $\Fex$, $\Sim$ will simulate  $(\textit{punish}, sid) \rightarrow  \Fjc$ without any output from $\Fjc$
        \end{itemize}
        \item If $u_n$ is dishonest: 
        \begin{itemize}[leftmargin=*]
            \item $\Sim$ watches if $u_n$ calls the $\Fjc$ in real world, if yes, $(\textit{punish}, sid) \leftarrow \Fex $. $\Sim$ also simulates the $\Fjc$ to relay the $\Fjc$ execution in the ideal world. If no message sent from $u_n$, $\Sim$ will forward nothing to $\Fex$. 
        \end{itemize}
   \end{itemize}
\end{mdframed}

\caption{Simulator Construction}
\label{fig:uc_sim}
\end{figure*}

\section{Security proof}    \label{sec:fullproof}
In this section, we first analyze our \textit{Enforceable A-HTLC} protocol in the UC framework, addressing the security properties ensured by \textit{Enforceable A-HTLC}. We then use the \textit{Enforceable A-HTLC} ideal functionality as building blocks, arguing for the fairness and confidentiality defined in Section \ref{sec:prob_def}.

\subsection{\textit{Enforceable  A-HTLC} in UC} \label{sec:uc}
Here, we define the \textit{ Enforceable A-HTLC} protocol (denoted as $\Pi$) in UC framework, and then define the ideal functionality $\Fex$. 

\subsubsection{The UC-security definition}
We define $\Pi$ as a \textit{hybrid} protocol that accesses a list of preliminary ideal functionalities: secure communication channel $\Fan$, global clock $\Fclk$, \textit{restricted-programmable} global random oracle $\FR$, signature scheme $\Fsig$, public ledger $\FL$, \textit{Judge Contract} $\Fjc$, and payment channel $\FC$. 

We denote these functionalities as $\Fpre$. The environment $\env$ supplies inputs to all parties in $\Pi$ and the adversary $\mathcal{A}$ with a security parameter $\lambda$ and auxiliary input $z$, $z \in \{0, 1\} ^ *$. The ensemble of outputs of executing protocol $\Pi$ are represented as $\text{EXEC}_{\Pi, \mathcal{A}, \env}^{\Fpre}\{\lambda, z\}$. 

In the ideal world, all parties do not interact with each other but simply forwarding the message from $\env$ to $\Fex$. $\Fex$ has access to following ideal functionalities:  global clock $\Fclk$, \textit{restricted-programmable} random oracle $\FR$, public ledger $\FL$, and payment channel $\FC$. We denotes these ideal functionalities as $\Fideal$. 
Let $\text{EXEC}_{\Fex, \Sim, \env}^{\Fideal}\{\lambda, z\}$ be  the ensemble of outputs of executing ideal functionality $\Fex$.

\begin{definition}[\textbf{UC Security}]
    \label{def:uc}
    A protocol $\Pi$ \textit{UC-realizes} an ideal functionality $\Fex$ if, for any probabilistic polynomial-time (PPT) adversary $\mathcal{A}$, there exists a simulator $\Sim$ such that, for any environment $\env$ with $z \in \{0, 1\} ^*$, $\lambda \in \mathbb{N}$, $\text{EXEC}_{\mathcal{F}, \mathcal{S}, \mathcal{E}}$ and $\text{EXEC}_{\tau, \mathcal{A}, \mathcal{E}}$ are computationally indistinguishable: 
    \[
    \text{EXEC}_{\Pi, \mathcal{A}, \env}^{\Fpre}\{\lambda, z\} \approx_{c} \text{EXEC}_{\Fex, \Sim, \env}^{\Fideal}\{\lambda, z\}
    \]
\end{definition}

\subsubsection{Ideal functionality of \textit{Enforceable A-HTLC}}
In this section, we present the ideal functionality $\Fex$, shown in Fig.\ref{fig:uc_if}.

\subsubsection{ Enforceable A-HTLC Protocol}
In this section, we formally present the Enforceable A-HTLC Protocol, shown in Fig.\ref{fig:uc_protocol}.

\subsubsection{Simulator Construction}

In this section, we present the simulator constuction in Fig.\ref{fig:uc_sim} and formal proof that  $\Pi$ \textit{UC-realize} the ideal functionality $\Fex$. We denotes a dummy party in the ideal world as $\tilde{u_i}$. 
Fig. \ref{fig:uc} demonstrates the real world execution and the ideal world execution. 
\begin{figure*}
    \centering
    \includegraphics[width=0.9\textwidth] {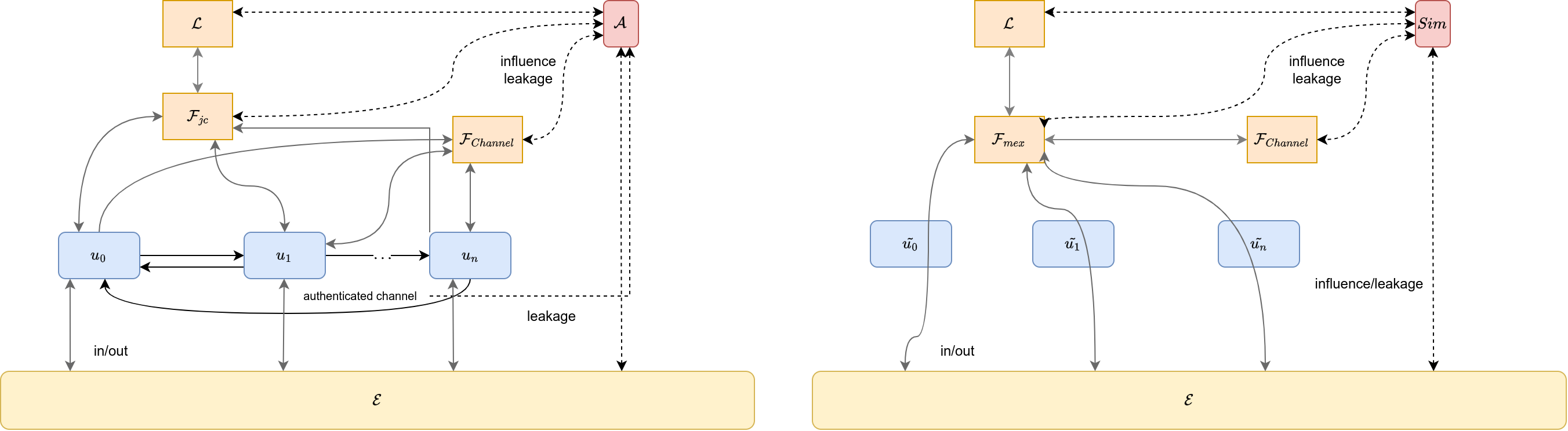}
    \caption{Setup of a Simulation with honest parties. The left part is the real $\Fpre$ hybrid world, while the right part is the $\Fideal$ ideal world.}
    \label{fig:uc}
\end{figure*}

\subsubsection{UC Proof}

In this section, we formally argue that $\Pi$ \textit{uc-realize}s $\Fex$. In our description, we write $m^i$ as the message observed at round $i$. 

\begin{lemma}   \label{lemma:uc_lock}
    The lock phase of $\Pi$ \textit{UC-realize}s the lock phase of $\Fex$. 
\end{lemma}
\begin{proof} 
We assume the protocol $\Pi$ starts on round $1$, and note that in the real world $\env$ controls $\mathcal{A}$.  We first define the messages used in the lock phase: 
    \begin{itemize}[leftmargin=*]
    \sloppy 
        \item $m_1$: the \textit{lock} message sent from $\env$ to $u_i$. 
        \item $m_2$: the \textit{locked} message sent from $u_i$ to $u_{i+1}$.
        \item $m_3$: the \textit{ack} message sent from $u_{i+1}$ to $u_0$.
    \end{itemize}
    We proceed to compare the messages received by $\env$ in the real world and the ideal world, considering various adversary assumptions. During the lock phase, the output of $u_0$ is only relevant to $\env$'s input $m_1$, while the output of $u_i$ ($i > 0$) is relevant to $\env$'s input $m_1$ and the output of $u_{i-1}$. Thus, we consider the following cases:

    \begin{itemize}[leftmargin=*]
    \sloppy 
    \item (Round 1): we consider $\env$ sends $m_1$ to $u_{0}$ in this round.   
    \begin{itemize}[leftmargin=*]
        \item Case 1: Assuming $u_0$ is honest. Once $\env$ sends $m_1$ to $u_0$ in round one, an honest $u_0$ will send the \textit{lock} message $m_2$ to $u_1$, resulting in $\text{EXEC}_{\Pi, \mathcal{A}, \env} := \{ m_1^1, m_2^{1}, \ldots\}$. In the ideal world, the Simulator will learn $tx_1$ from the leakage of $\Fex$ and simulate the \textit{locked} message, achieving the same output.
        \item Case 2: Assuming $u_0$ is corrupted. $\Sim$ captures the message $m_2$ sent from $u_0$ in the real world and extracts $tx_1'$ from $m_2$. $\Sim$ will then forward $tx_1'$ to $\Fex$.
    \end{itemize}
    \item(Round $1 + i $): we consider $\env$ sends $m_1$ to $u_{i}$  ($i \in [1, n]$) in this round.  
    \begin{itemize}[leftmargin=*]
        \item  Case 1: Assuming $u_{i-1}$ is corrupted, while $u_{i}$ remains honest. In the real world, during round $i$, if the corrupted $u_{i-1}$ sends an invalid $m_2^{i}$ to $u_{i}$, $u_{i}$ will refuse to proceed to the next lock and abort the protocol, resulting in $\text{EXEC}_{\Pi, \mathcal{A}, \env} := \{ \ldots,m_1^{i},  m_2^{i}, m_1^{1+i} \}$. In the ideal world, $\Sim$ captures the $m_2^{i}$ sent from $u_i$, extracts $tx_i'$, and forwards $tx_i'$ to $\Fex$. Since $\Fex$ follows the same logic as $\Pi$, when $u_{i}$ is honest, the execution result in the ideal world is: $\text{EXEC}_{\Fex, \Sim, \env} := \{ \ldots,m_1^{i},  m_2^{i}, m_1^{1+i} \}$.
        \item Case 2: Assuming $u_{i-1}$ and $u_{i}$ are both honest. In the real world, during round $i$, $u_{i-1}$ sends a valid $m_2^i$ to $u_{i}$, and $u_{i}$ proceeds to the next lock, resulting in $\text{EXEC}_{\Pi, \mathcal{A}, \env} := \{ \ldots, m_1^{i},  m_2^{i},m_3^{i}, m_1^{1+i}, m_2^{1 + i}, m_3^{1 + i}, ,\ldots\}$. In the ideal world, $\Sim$ learns the leakage from $\Fex$ and constructs $m_2$ and $m_3$, resulting in $\text{EXEC}_{\Fex, \Sim, \env} := \{ \ldots, m_1^{i},  m_2^{i},m_3^{i}, m_1^{1+i}, m_2^{1 + i}, m_3^{1 + i}, ,\ldots\}$.
        \item Case 3: Assuming $u_{i}$ is corrupted. $\Sim$ can observe everything sent by the corrupted $u_i$ and simulate the communication in the ideal world. Then $\Sim$ will extract the corresponding input from the messages sent by $u_i$ and forward it to $\Fex$.
    \end{itemize}
\end{itemize}

\end{proof}

\begin{lemma}   \label{lemma:uc_unlock}
    The lock-unlock phase of protocol $\Pi$ \textit{UC-realizes} the lock-unlock phase of $\Fex$. 
\end{lemma}
\begin{proof}
We first define new messages used in the unlock phase: 
\begin{itemize}[leftmargin=*]
    \item $m_4$: the \textit{release} message sent from $\env$ to $u_0$.
    \item $m_5$: the \textit{released} message sent from $u_0$ to $u_n$.
    \item $m_6$: the \textit{unlock} message sent from $\env$ to $u_i$ $i \in [1, n]$.
    \item $m_7$: the \textit{updated} message leaked from $\FC$ to $\mathcal{A}$ and $\Sim$. 
\end{itemize}

We subsequently compare the messages received by $\env$ in both the real world and the ideal world, considering different assumptions regarding the adversary. It is assumed that $\env$ transmits $m_4$ during round $T'$.
The output of $u_0$ is solely pertinent to $\env$'s input and the $m_3$ transmitted during the lock phase.
The update message $m_7$ concerning channel $cid_n$ solely pertains to $u_0$'s \textit{released} message and the environment's input to $u_n$, while $m_7$ concerning channel $cid_i$ solely pertains to $u_{i+1}$'s \textit{unlock} message and the environment's input to $u_i$. Therefore, the following conditions are considered:

\begin{itemize}[leftmargin=*]
    \sloppy 
    \item (Round $T'$): we consider $\env$ sends $m_4$ to $u_{0}$ in round $T'$. 
    \begin{itemize}
        \item Case 1: Assuming $u_0$ is honest, along with all $u_i$ being honest during the lock phase. In the real world, upon receiving $m_4$, $u_0$ will send $m_5$, resulting in the execution $\text{EXEC}_{\Pi, \mathcal{A}, \env} := \{ m_4^{T'}, m_5^{T'}, \ldots\}$. In the ideal world, $\Fex$ will disclose $s_0$ to $\Sim$, which will simulate this execution with the output $\text{EXEC}_{\Fex, \Sim, \env} := \{ m_4^{T'}, m_5^{T'}, \ldots\}$.
        \item Case 2:  Assuming $u_0$ is honest, but at least one $u_i$ is corrupted during the lock phase. In the real world, $u_0$ will refrain from sending $m_4$ if not all signatures are received during the lock phase, resulting in $\text{EXEC}_{\Pi, \mathcal{A}, \env} := \{ m_4^{T'}, \ldots\}$. In the ideal world, $\Fex$ will also prevent the transmission of $m_4$ from $\env$ (the $ifEnf$ flag maintained by $\Fex$ is less than $n$), resulting in $\text{EXEC}_{\Fex, \Sim, \env} := \{ m_4^{T'}, \ldots\}$.
        \item Case 3: Assuming $u_0$ is corrupted. If $u_0$ sends an invalid $s_0'$ in $m_5$ in the real world, $\Sim$ will forward $s_0'$ to $\Fex$ and simulate $m_5$, resulting in $\text{EXEC}_{\Fex, \Sim, \env} := \{ m_4^{T'}, m_5^{T'}, \ldots\}$. If $u_0$ fails to send any $m_5$ in round $T'$, $\Sim$ will not provide any input to $\Fex$.
    \end{itemize}
    \item (Round $T' + 1$): we consider $\env$ sends $m_6$ to $u_n$ in this round. 
    \begin{itemize}
        \item  Case 1: Assuming both $u_0$ and $u_n$ are honest. In the real world, $u_0$ will send $m_5$ to $u_n$, and $u_n$ will update its channel, resulting in $\text{EXEC}_{\Pi, \mathcal{A}, \env} := \{ m_4^{T'}, m_5^{T'}, m_6^{T' + 1}, m_7^{T' + 1}, \ldots\}$. In the ideal world, $\Sim$ will forward $s_0$ to $\Fex$ and simulate $m_6$, while $m_7$ will be generated by $\Fex$, resulting in $\text{EXEC}_{\Fex, \Sim, \env} := \{ m_4^{T'}, m_5^{T'}, m_6^{T' + 1}, m_7^{T' + 1}, \ldots\}$.
        \item Case 2: Assuming $u_0$ is corrupted and $u_n$ is honest. In the real world, $u_n$ will unable to send $m_7$ as $u_0$ donot reveal the correct secrets, resulting $\text{EXEC}_{\Pi, \mathcal{A}, \env} := \{ m_4^{T'}, m_5^{T'}, m_6^{T' + 1}\}$. In the ideal world, $\Fex$ will also refuse to submit $m_7$, resulting $\text{EXEC}_{\Fex, \Sim, \env} := \{ m_4^{T'}, m_5^{T'}, m_6^{T' + 1}\}$
        \item Case 3: Assuming $u_n$ is corrupted. $\Sim$ captures all message sent in the real world and relay in the ideal world, resulting the same outputs. 
    \end{itemize}
    \item(Round $T_r + 1$): we consider $\env$ sends $m_6$ to $u_{i+1}$ in this round. 
    \begin{itemize}
        \item Case 1: Assuming $u_i$ is honest, while $u_{i+1}$ is corrupted. If $u_{i+1}$ fails to trigger the \textit{updated} message $m_7$ in the real world during round $T_r$, $u_i$ would never trigger the channel update of $cid_i$ since the unlock condition would never be satisfied, resulting in $\text{EXEC}_{\Pi, \mathcal{A}, \env} := \{ m_6^{T_r},  m_6^{T_r+ 1}, \ldots\}$.
        In the ideal world, $\Fex$ would maintain $states[i+1]$ as \textit{locked}, preventing $u_i$ from triggering the update call, resulting in $\text{EXEC}_{\Pi, \mathcal{A}, \env} := \{ m_6^{T_r},  m_6^{T_r+ 1}, \ldots\}$.
        \item Case 2: Assuming both $u_i$ and $u_{i+1}$ are honest ($i \in [1, n-1]$). If $u_{i+1}$ sends $m_7$, $u_i$ will also send $m_7$, resulting in $\text{EXEC}_{\Pi, \mathcal{A}, \env} := \{ m_6^{T_r}, m_7^{T_r}, m_6^{T_r+ 1},m_7^{T_r + 1}, \ldots\}$. In the ideal world, $\Fex$ will send $m_7$, resulting in $\text{EXEC}_{\Fex, \Sim, \env} :=  \{ m_6^{T_r}, m_7^{T_r}, m_6^{T_r+ 1},m_7^{T_r + 1}, \ldots\}$.
        \item Case 3: Assuming $u_{i+1}$ is corrupted ($i \in [1, n-1]$). $\Sim$ can observe everything sent by corrupted $u_i$ in the real world and simulate the communication in the ideal world.
    \end{itemize}

\end{itemize}

\end{proof}

\begin{lemma}   \label{lemma:uc_enforce}
    The lock-enforcement phase of protocol $\Pi$ \textit{UC-realize}s the lock-enforcement phase of $\Fex$. 
\end{lemma}

\begin{proof}

First define messages used in the enforcement phase: 

\begin{itemize}[leftmargin=*]
    \item $m_8$: the \textit{challenge} message sent from $\env$ to $u_0$.
    \item $m_9$: the \textit{enforce} message send from $u_0$ to $\Fjc$.
    \item $m_{10}$: the \textit{enforced} message leaked from $\Fjc$.
    \item $m_{11}$: the \textit{response} message sent from $\env$ to $u_i$.
    \item $m_{12}$: the \textit{log} message sent from $u_i$ to $\Fjc$. 
    \item $m_{13}$: the \textit{logged} message leaked from $\Fjc$. 
    \item $m_{14}$: the \textit{punish} message sent from $\env$ to $u_0$.
    \item $m_{15}$: the \textit{punish} message sent from $u_0$ to $\Fjc$.
    \item $m_{16}$: the \textit{transfer} message sent from $\Fjc$ to $\FL$
    \item $m_{17}$: the \textit{punished} message leaked from $\Fjc$
\end{itemize}

We then compare the messages that $\env$ receives in the real world and ideal world under different adversary assumptions. We assume that $\env$ sends $m_8$ in round $T''$. In the real world, $m_9$ is solely pertinent to $\env$'s input and the $m_3$ messages transmitted during the lock phase. $u_n$'s output is solely pertinent to $\env$'s input and the message $m_{10}$, and $u_i$'s output ($i \in [1, n-1]$) is solely pertinent to $\env$'s input, the message $m_{10}$, and $u_{i+1}$'s output. The punishment result ($m_{15}, m_{16}, m_{17}$) is pertinent to $\env$'s input and the message $m_{13}$ sent from $u_i$. 
Since the enforcement phase follows a round-by-round setting, we can demonstrate the comparison as follows:

\begin{itemize}[leftmargin=*]
        \item (Round $T''$): we consider $\env$ sends $m_8$ to $u_0$ in this round. 
        \begin{itemize}[leftmargin=*]
            \sloppy
            \item Case 1: Assuming $u_0$ is honest, and all user $u_i$ are honest in the lock phase. In the real world,  $u_0$ will submit $m_{9}$ to $\Fjc$, then $\Fjc$ will broadcast $m_{10}$, resulting $\text{EXEC}_{\Pi, \mathcal{A}, \env} := \{  m_8^{T''}, m_9^{T''},  m_{10}^{T''}, \ldots \}$
             In the ideal world,  $\Sim$ will simulate  $\Fjc$'s execution, resulting $\text{EXEC}_{\Fex, \Sim, \env} := \{  m_8^{T''}, m_9^{T''},  m_{10}^{T''}, \ldots \}$. 
    
            \item Case 2: Assuming $u_0$ is honest, and there exist corrupted $u_i$ in the lock phase. In the real world, honest $u_0$ will never send $m_8$ to $\Fjc$ since not all valid $m_3$ messages are received. The same applies to $\Fex$ in the ideal world. In the following cases, we assume all user $u_i$ are honest in the lock phase. 
            \item Case 3: Assuming $u_0$ is corrupted. If $u_0$ reveals an incorrect secret $s_0'$ ($Open(s_0', h_0) \neq 1$) in message $m_{9}$, $\Fjc$ will refuse to send $m_{10}$, and no further outputs will be observed: $\text{EXEC}_{\Pi, \mathcal{A}, \env} :=  \{ m_8^{T''}, m_9^{T''}\}$. In the ideal world, $\Sim$ will simulate the $\Fjc$ interaction $m_9$, and $\Fex$ will refuse to send further response messages $m_{11}$ from $\tilde{u_i}$: $\text{EXEC}_{\Fex, \Sim, \env} :=  \{ m_8^{T''}, m_9^{T''} \}$.
        \end{itemize}
        \item (Round $T'' +  1$):  we consider $\env$ sends $m_{11}$ to $u_n$ in this round.
        \begin{itemize}[leftmargin=*]
            \sloppy
            \item  Case 1: Assuming $u_0$ is corrupted, and $u_n$ is honest. If $u_0$ reveals an incorrect secret $s_0'$ ($Open(s_0', h_0) \neq 1$) in message $m_{9}$, $\Fjc$ will refuse to send $m_{10}$, and no further outputs will be observed: $\text{EXEC}_{\Pi, \mathcal{A}, \env} :=  \{ m_8^{T''}, m_9^{T''},\ldots, m_{11}^{T'' +  i + 1}, \ldots \}$. 
            In the ideal world, $\Sim$ will simulate the $\Fjc$ interaction $m_9$, and $\Fex$ will refuse to send further response messages $m_{11}$ from $\tilde{u_i}$: $\text{EXEC}_{\Fex, \Sim, \env} :=  \{ m_8^{T''}, m_9^{T''}, \ldots, m_{11}^{T'' + i + 1}, \ldots \}$.
            \item Case 2: Assuming $u_0$ is honest, and $u_n$ is honest. In the real world, $u_n$ will send $s_n$ to $\Fjc$, and $\Fjc$ will broadcast the \textit{log} message, resulting $\text{EXEC}_{\Pi, \mathcal{A}, \env} := \{  m_8^{T''}, m_9^{T''},  m_{10}^{T''}, m_{11}^{T'' + 1},  m_{12}^{T'' + 1}, m_{13}^{T'' + 1}, \ldots \}$. In the ideal world, $\Sim$ leaks $s_n$ to $\Sim$, while in turn, $\Sim$ simulates the execution of $\Fjc$, resulting $\text{EXEC}_{\Fex, \Sim, \env}  := \{  m_8^{T''}, m_9^{T''},  m_{10}^{T''}, m_{11}^{T'' + 1},  m_{12}^{T'' + 1}, m_{13}^{T'' + 1}, \ldots \}$. If the payment channel is not updated, $u_n$ in the real world will send $m_7^{T'' + 1}$, so as the $\Fex$ in the ideal world. 
            
            \item Case 3: Assuming $u_n$ is corrupted. If $u_n$ sends an invalid $m_12$ to $\Fjc$, the real world execution is: $\text{EXEC}_{\Pi, \mathcal{A}, \env} := \{  \ldots, m_{11}^{T'' + 1},  m_{12}^{T'' + 1} \}$. In the ideal world, $\Sim$ will simulate the invalid interaction, resulting $\text{EXEC}_{\Fex, \Sim, \env}  := \{ \ldots, m_{11}^{T'' + 1},  m_{12}^{T'' + 1}\}$. 
        \end{itemize}
        \item (Round $T'' + n - i +  1$):  we consider $\env$ sends $m_{11}$ to $u_i$ $(i \in [1, n-1])$ in this round.
        \begin{itemize}
            \sloppy 
            \item Case 1: Assuming $u_0$ is honest, $u_{i}, \ldots, u_{n}$ are honest. In the real world, $u_i$ will  send $s_i$ to $\Fjc$, and $\Fjc$ will broadcast the \textit{log} message, resulting $\text{EXEC}_{\Pi, \mathcal{A}, \env} := \{  \ldots, m_{11}^{T'' + n - i +  1},  m_{12}^{T'' + n - i +  1}, m_{13}^{T'' + n - i +  1}, \ldots \}$. In the ideal world, $\Sim$ leaks $s_i$ to $\Sim$, while in turn, $\Sim$ simulates the execution of $\Fjc$, resulting $\text{EXEC}_{\Fex, \Sim, \env}  := \{  \ldots, m_{11}^{T'' + n - i +  1},  m_{12}^{T'' + n - i +  1}, m_{13}^{T'' + n - i +  1}, \ldots \}$. 
            \item Case 2: Assuming $u_0$ is honest, $u_i$ is honest, and there exist a $u_r$ $(r> i)$ is corrupted. In the real world, $u_i$ will refuse to send $m_12$ to $\Fjc$. resulting $\text{EXEC}_{\Pi, \mathcal{A}, \env} := \{  \ldots, m_{11}^{T'' + n - i +  1}, \ldots \}$. In the ideal world, $\Fex$ will leak nothing to $\Sim$ in this round, and $\Sim$ will do nothing, resulting  $\text{EXEC}_{\Fex, \Sim, \env} := \{  \ldots, m_{11}^{T'' + n - i +  1}, \ldots \}$.
            \item Case 3: Assuming $u_0$ is corrupted, $u_i$ is honest. As $u_0$ is corrupted, $u_i$ will refuse to send $m_12$ to $\Fjc$, so the real world execution and ideal world execution is the same as Case 2. 
            \item Case 4: Assuming $u_i$ is corrupted. $\Sim$ can observe everything sent by corrupted $u_i$ in the real world and simulate the communication in the ideal world.
        \end{itemize}
        \item (Round $T'' + n + 1$): we consider $\env$ sends $m_{14}$ to $u_0$ in this round.
        \begin{itemize}
            \sloppy 
            \item Case 1: Assuming all users are honest. In the real world, $u_0$ will send $m_{15}$, but $\Fjc$ will not trigger $m_{16}$ and $m_{17}$ as no party is corrupted, resulting $\text{EXEC}_{\Pi, \mathcal{A}, \env} := \{  \ldots, m_{14}^{T'' + n+  1}, m_{15}^{T'' + n+  1} \}$. In the ideal world, $\Fex$ will not trigger $m_{16}$ and $\Sim$ only simulates the message towards $\Fjc$ with no response, resulting $\text{EXEC}_{\Fex, \Sim, \env} := \{  \ldots, m_{14}^{T'' + n+  1}, m_{15}^{T'' + n+  1} \}$.
            \item Case 2: Assuming $u_0$ is honest, and there exist a $u_r$ $(r> 0)$ is corrupted. In the real world,  $u_0$ will send $m_{15}$, and $\Fjc$ will trigger $m_{16}$ to $\FL$ and $m_{17}$, resulting $\text{EXEC}_{\Pi, \mathcal{A}, \env} := \{  \ldots, m_{14}^{T'' + n+  1}, m_{15}^{T'' + n+  1},  m_{16}^{T'' + n+  1},  m_{17}^{T'' + n+  1} \}$. In the ideal world, $\Fex$ will send $m_{16}$ to $\FL$, and $\Sim$ will simulates the interaction between $u_0$ and $\Fjc$, resulting $\text{EXEC}_{\Fex, \Sim, \env} := \{  \ldots, m_{14}^{T'' + n+  1}, m_{15}^{T'' + n+  1},  m_{16}^{T'' + n+  1},  m_{17}^{T'' + n+  1} \}$
            \item Case 3: Assuming $u_0$ is corrupted. In the real world, $u_0$ submits an invalid $m_{15}$ to $\Fjc$, $\Fjc$ will do nothing about this call, so the execution output of both real world and ideal world is identical to Case 1. If $u_0$ sends nothing in this round, $\text{EXEC}_{\Pi, \mathcal{A}, \env} := \{  \ldots, m_{14}^{T'' + n+  1}\}$. In the ideal world, $\Sim$ will forward nothing to $\Fex$, resulting  $\text{EXEC}_{\Fex, \Sim, \env} := \{  \ldots, m_{14}^{T'' + n+  1}\}$.
        \end{itemize}
       
\end{itemize}

\end{proof}

\begin{theorem}
    \textit{Enforceable A-HTLC} protocol $\Pi$ \textit{UC-realize}s the ideal functionality $\Fex$.
\end{theorem}

\begin{proof}
    This theorem follows from Lemmas \ref{lemma:uc_lock}, \ref{lemma:uc_unlock} and \ref{lemma:uc_enforce}. 
\end{proof}

\subsubsection{Security properties of $\Fex$}

Let us discuss the security properties guaranteed by $\Fex$ in multi-party money-secret exchanges.  
Since our \textit{enforceable A-HTLC} protocol realizes the ideal functionality $\Fex$, these security properties are also achieved by our protocol in the real world.
In a $\Fex$ with $n$ payee and a deadline $T$: 

\begin{lemma}   \label{lemma:fex1}
    In $\Fex$, an honest $u_i$ reveals $s_i$ only if it receives $\mathfrak{B}_i$. 
\end{lemma}
\begin{proof}
    We first consider the unlock phase of $\Fex$: 
    For an honest payee $u_i$ ($i \in [1, n-1]$) selling secret $s_i$, once the channel $cid_{i+1}$ is updated (redeeming the payment $\sum_{j=i}^n \mathfrak{B}_j$ from $u_i$), $\FC$ will send the redeem secret $\mathbb{S}$ to $u_i$. Since the transaction locking time $tx_i.\phi.t > tx_{i+1}.\phi.t$, $u_i$ can update the channel $cid_i$, redeeming the payment $\sum_{j=i}^n \mathfrak{B}_j$ from $u_{i-1}$.

    We then consider the enforcement phase of $\Fex$: Once $u_0$ submits the \textit{enforce} message to $\Fex$ at round $T'$ ($T' < T$), $u_i$ has to reveal its secret $s_i$ in round $T' + n - i + 1$ if $u_0$ and all its last hops ($u_{i+1}, \ldots, u_n$) reveal their secrets. Since the timelock for updating channel $cid_i$ is $t_i := T + n - i + 2$, and $t_i > T' + n - i + 1$, $u_i$ has enough time to unlock its payment after revealing its secret.

    In summary, in the unlock phase, once $u_{i+1}$ updates channel $cid_{i+1}$, $u_i$ can always update channel $cid_i$. In the enforcement phase, even if $u_i$ has to reveal its secret, $u_i$ can still update channel $cid_i$. Therefore, an honest $u_i$ reveals $s_i$ only if it receives the fee $\mathfrak{B}_i$.

\end{proof}

\begin{lemma}    \label{lemma:fex2}
    In $\Fex$, an honest $u_0$ pays $\sum_{i=1}^n \mathfrak{B}_i$ only if it receives all the secrets $\{s_1, \ldots, s_n\} $ from all payees. 
\end{lemma}

\begin{proof}
    In the lock phase, an honest $u_0$ requires $u_i$ to provide $s_0, \ldots, s_n$ in order to redeem this payment, protected by $\FC$. Once this payment is redeemed, $u_0$ will receive these secrets from $\FC$.
\end{proof}

\begin{lemma}    \label{lemma:fex3}
   \textit{Controllable.} In $\Fex$, an honest $u_0$ pays nothing until $u_0$ sends the \textit{release} or \textit{enforce} message.
\end{lemma}

\begin{proof}   
    In the lock phase, an honest $u_0$ adds $h_0 := \text{Commit}(s_0)$ to $u_1$'s redeem condition. Since the secret $s_0$ is kept private in $u_0$, $u_1$ cannot update this payment in $\FC$ until $u_0$ sends the \textit{release} or \textit{enforce} message, which includes $s_0$.
\end{proof}

\begin{lemma}    \label{lemma:fex4}
   \textit{Enforceable.} If $\Fex$ completes the \textit{lock} phase with $ifEnf = n$, $\Fex$ guarantees that the payer $u_0$ can enforce all payees to reveal their secrets before round $T + n$, or $\CP$ will receive a refund of amount $\mathfrak{B}_{\text{max}}$.
\end{lemma}

\begin{proof}
    Once $ifEnf = n$, $\Fex$ allows $u_0$ to reveal $s_0$ by sending an \textit{enforce} message to $\Fex$ in round $T'$ $(T' < T)$. Then, each payee $u_i$ has to reveal its secret $s_i$ to $\Fex$. If a payee $u_r$ is the first to fail in revealing its secret, $\Fex$ allows $u_0$ to submit a \textit{punish} message to $\Fex$ at the end of the enforcement phase, transferring $\mathfrak{B}_{\text{max}}$ from $u_r$ to $u_0$.
\end{proof}

\begin{lemma}    \label{lemma:fex5}
   Consider a payer $u_0$ initiating $\eta$ instances of $\Fex$ with different payees in $\eta$ payment paths at round $T''$, and all $u_0$ sets the same enforcement deadline $T$ for each instance. We assume each path $p_i$ has $n_i$ payees. If all $\Fex$ instances complete the \textit{lock} phase with $ifEnf = n_i$, the payer $u_0$ can enforce all payees in all $\Fex$ instances to reveal their secrets before round $T + \max(n_i) + 1$, or $u_0$ will receive a refund of at least amount $\mathfrak{B}_{\text{max}}$.
\end{lemma}

\begin{proof}
   The lemma follows directly from Lemma \ref{lemma:fex4}. 
\end{proof}

\subsection{Fairness and Confidentiality in FairRelay}    \label{sec:proof}

In this section, we formally argue that FairRelay guarantees the fairness and confidentiality properties defined in Section \ref{sec:prob_def}.  

Consider a relay-assisted content exchange $\mathcal{G}$ in the FairRelay protocol, where $\CP$ delivers content $m = \{m_1, \dots, m_n\}$ committed by $com_m$ to $\C$ through $\eta$ delivery paths. Each relay path $p$ is assigned a delivery job $Job(p)$, where $Job(p) \subset m$. All paths should complete their delivery job before time $T_1$.
$\CP$ generates a new encryption key $sk_0$ and a mask $s_0$. The $i^{th}$ relayer in path $p_k$, denoted as $\R_{k, i}$, generates its encryption key $sk_{k,i}$ and the mask secret $s_{k,i}$.
All participants agree on the delivery deadline $T_1$ and the challenge deadline $T_2$. At round $T_1$, if $\eta$ paths ($\mathbb{P}$) finish the ciphertext delivery, $\C$ issues a conditioned payment to $\CP$. In the next round, $\CP$ issues $\eta$ $\Fex$ for each relay path in $\mathbb{P}$ in exchange for their secret $s_{k, i}$.

\begin{theorem}
    FairRelay guarantees the fairness for the customer $\C$ (Definition \ref{def:exf_r}).
    \label{pro:fair_r}
\end{theorem}

\noindent \begin{proof}
\sloppy 
Before $\C$ makes the payment in round $T_1$, $\C$ receives one mask commitment from each user, ciphertext of chunks along with encryption commitment chain, and the Merkle multi-proof for each plaintext chunk.

$\C$ can compose the following tuple for any chunk $m_{idx} \in m$ with index $idx$ relayed in path $p_k$: $(idx, c_{|p_k|, idx}, \mathbb{C}_1, \mathbb{C}_2, \pi_{\text{merkle}})$, where:
- $c_{|p_k|, idx}$ is the ciphertext encrypted ($|p_k| + 1$) times by $\{\CP, \R_{k, 1}, \ldots, \R_{k, |p_k|}\}$ sequentially;
- $\mathbb{C}_1$ is the encryption commitment chain $\{com_{\text{enc}}^{0, idx}, com_{\text{enc}}^{1, idx}, \ldots, com_{\text{enc}}^{|p_k|, idx}\}$,
- $\mathbb{C}_2$ is the mask commitments $\{com_{\text{mask}}^0, com_{\text{mask}}^{k, 1}, \ldots, com_{\text{mask}}^{k, |p_k|}\}$,
- $\pi_{\text{merkle}}$ proves that $com_{\text{enc}}^{0, idx}.h_m$ is the $idx^{th}$ leaf of the Merkle tree with root $com_m$.

Additionally, $(c_{|p_k|, idx}, \mathbb{C}_1, \mathbb{C}_2, idx)$ forms a \textit{valid}( defined in \ref{sec:valid_tuple}) tuple signed by $\{\CP, \R_{k, 1}, \ldots, \R_{k, |p_k|}\}$. 

Then $\C$ performs a conditioned off-chain payment of amount $\mathfrak{B}_m$ to $\CP$ in exchange for all secrets in $\mathbb{P}$. $\FC$ guarantees that once this conditioned payment is redeemed by $\CP$ and updated on $\FC$, $\C$ will receive all requested mask secrets to decrypt these \textit{valid} tuples. 
According to Corollary \ref{coro:valid}, once all mask secrets $s_0, s_{k, 1}, \ldots, s_{k, |p_k|}$ are revealed, $\C$ can either obtain $m_{idx}'$ committed by $com_{\text{enc}}^{0, idx}.h_m$ when all nodes in this path are honest, or $\C$ can generate a proof of misbehavior towards the closest corrupted node.

If all providers and relayers are honest, $\C$ will obtain all chunks $m_{idx} \in m$ and compose the content $m$ committed by $com_m$. Otherwise, $\C$ can generate at least one proof of misbehavior towards the \textit{Judge Contract}, which will refund an amount of $\mathfrak{B}_{\text{max}}$ to $\C$. Since the \textit{Judge Contract} requires all content prices to be lower than $\mathfrak{B}_{\text{max}}$, we can consider $\C$ pays nothing when $\C$ does not receive $m$.

\end{proof}

\begin{theorem}
    FairRelay guarantees the fairness for the honest content provider $\CP$ (Definition \ref{def:exf_p}).
    \label{pro:fair_p}
\end{theorem}

\begin{proof}

We contend that the revelation of $m$ is contingent upon $\CP$ receiving payment. Each constituent chunk $m_r \in m$ undergoes encryption using $sk_0$ and is subsequently masked by $s_0$. The inherent hiding properties (Lemma \ref{lemma:pomm_hiding}, \ref{lemma:pome_hiding}) ensure that the disclosure of $m_r$ remains unattainable as long as $s_0$ is kept confidential. Only upon the redemption of payment by $\CP$ from $\C$, will $s_0$ be disclosed to $\C$ via $\FC$.

Subsequently, we argue that $\CP$ refrains from remitting relay fees until it receives payment from $\C$. As  $\C$ first initiates a conditioned payment in return for all secrets within $\mathbb{P}$. Consequently, $\CP$ issues $\eta$ $\Fex$ messages to each path  in exchange for the relayers' mask secrets. In the event that not all paths successfully complete the lock phase, $\CP$ terminates the protocol by ceasing to transmit \textit{release} or \textit{enforce} messages to $\Fex$. In such instances, no party involved in this transaction receives payment. Conversely, if all paths progress to the unlock or enforcement phase, the veracity of Lemma \ref{lemma:fex5} ensures that $\CP$ receives all secrets atomically, prior to a predetermined deadline, after which $\CP$ can update the payment from $\C$.

Furthermore, the \textit{soundness} of \textit{zk-SNARK} guarantees that no party can submit a proof of misbehavior to an honest $\CP$, thereby reducing its balance. Consequently, we stipulate that $\CP$ only pays the relay fee upon receiving payment from $\C$.

\end{proof}

\begin{theorem}
    FairRelay guarantees the fairness for any relayer $\R$   (Definition \ref{def:exf_c}).
    \label{pro:fair_c}
\end{theorem}

\begin{proof}
Within path $p_k$, it is imperative that the ciphertext received by $\C$ is encrypted using the encryption key $sk_{k, i}$ belonging to $\R_{k, i}$. The underlying hiding properties (Lemma \ref{lemma:pomm_hiding}, \ref{lemma:pome_hiding}) ensure that the ciphertext transmitted along this path remains incompletely decryptable as long as $s_{k, i}$ is kept confidential. Lemma \ref{lemma:fex1} guarantees that an honest $\R_{k, i}$ only discloses its secret $s_{k, i}$ upon receiving the payment $\mathfrak{B}_{k, i}$. Consequently, the content relayed on path $p_k$ ($Job_{p_k}$) can only be unveiled if all $\R_{k, i}$ entities receive a payment of $\mathfrak{B}_{k, i}$.

Additionally, the \textit{soundness} of \textit{zk-SNARK} guarantees that no party can present a proof of misbehavior to an honest $\R_{k, i}$, resulting in a reduction of its balance.
As a consequence, we establish that $\C$ cannot gain access to the ciphertext chunks relayed by $\R$ unless $\R$ receives its relay fee.

\end{proof}

\begin{theorem}
    FairRelay guarantees the Confidentiality  (Definition \ref{def:conf}).
\end{theorem}

\begin{proof}
Since each chunk $m_r \in m$ is encrypted using $sk_0$ and masked with $s_0$, the hiding properties (Lemma \ref{lemma:pomm_hiding}, \ref{lemma:pome_hiding}) ensure that no one can obtain $m_r$ as long as $s_0$ remains confidential.

The honest $\CP$ privately sends the mask commitment to the honest $\C$ by encrypting the commitment with $\C$'s public key. Considering that the masked secret $s_0$ is publicly revealed over the payment channel, only $\CP$ and $\C$ possess the encryption key $sk_0$. Consequently, as long as $\CP$ and $\C$ act honestly, no relayer collision can access the plaintext content $m$.

\end{proof}

\section{Discussions}

\noindent \textbf{Termination. }
In the worst case, an honest party participating in FairRelay terminates the protocol when the round reaches the expiration time of its incoming conditioned payment. However, if the incoming channel is not even locked, an honest party is able to abort the protocol earlier than the worst case scenario.

\noindent \textbf{From a single payment channel to a multi-hop payment.} \label{app:loose}
In section \ref{sec:construction} we assumed that the payment path between $\C$ to $\CP$ can be considered as a single payment channel. We can simply expand this condition to the full payment path by utilizing multi-hop  payment schemes. The most naive solution is to repeat the condition payment hop by hop with incremental timelock (similar to HTLC multi-hop payment). We can also integrate this payment with advanced multi-hop payment schemes like AMHL\cite{malavolta_anonymous_2019}, Sprites\cite{miller_sprites_2017} or Blitz\cite{aumayr2021blitz}, ensuring the \textit{strong atomicity} over this multi-hop payment.

\noindent \textbf{Tackle Front-running Attacks. }   \label{dis:front-running}
In our protocol, both relayers and providers have a global security deposit locked on-chain, which is shared among multiple FairRelay instances. However, this setting enables dishonest relayers or providers to conduct front-running attacks against honest customers.
In a front-running attack, a dishonest party (denoted as $v$) sends a proof of misbehavior against itself to a colluded party. The colluded party can then submit these proofs of misbehavior to the \textit{Judge Contract}, constantly draining the security deposit.
To mitigate this problem, a slashing scheme can be added to the proof of misbehavior handler in the Judge Contract. This ensures that no rational party would attempt such attacks, as their security deposit would be burned.
Additionally, the customer checks the service provider's deposit before making the conditioned payment, and the payment has an expiration time. This means that the dishonest node has only a limited time-window to perform such attacks. Moreover, the throughput of the blockchain is limited, which restricts the extractable deposit.
By increasing the amount of the security deposit, we can guarantee that there is always a sufficient amount for the customer to claim a refund. By combining these two solutions, the front-running attack can be effectively mitigated.

\noindent \textbf{Fault Tolerance. }
In our current multi-path FairRelay protocol, the protocol will securely terminate if a single relayer fails to fulfill its task. In future iterations, fault tolerance schemes could be integrated into our existing design to enhance the robustness of the content delivery protocol.

\end{document}